\documentclass[pra,aps,twocolumn,groupedaddress,nofootinbib,superscriptaddress,preprintnumbers]{revtex4-2}

\usepackage{graphicx}
\usepackage[nointegrals]{wasysym} %
\usepackage[export]{adjustbox}
\usepackage{amsmath,amsfonts,amssymb,latexsym}
\usepackage{hhline}
\usepackage{bm}
\usepackage{verbatim}
\usepackage{enumitem}
\usepackage{mathrsfs}
\usepackage{slashed}
\usepackage{empheq}

\newcommand{\Tr}{\text{Tr}}

\newcommand{\End}{\text{End}}

\newcommand{\lan}{\langle}
\newcommand{\ran}{\rangle}

 \newcommand{\unit}{\mathbf{1}}

\newcommand{\da}{{\dagger}}

\newcommand{\ob}[1]{\mkern 1.5mu\overline{\mkern-1.5mu#1\mkern-1.5mu}\mkern 1.5mu}

\newcommand{\ra}{\rightarrow}

\newcommand{\wt}{\widetilde}

\newcommand{\uvx}{{\mathbf{\hat x}}}
\newcommand{\uvy}{{\mathbf{\hat y}}}

\newcommand{\bfzero}{{\mathbf{0}}}

\renewcommand{\(}{\left(}
\renewcommand{\)}{\right)}
\renewcommand{\[}{\left[}
\renewcommand{\]}{\right]}
\newcommand{\mt}{\mapsto}
\newcommand{\xra}{\xrightarrow}
\newcommand{\tp}{\otimes}

\newcommand{\twp}{{2\pi}}

\newcommand\bpm            {\begin{pmatrix}}
	\newcommand\epm           {\end{pmatrix}}

\newcommand{\ms}{\medskip}
\newcommand{\bs}{\bigskip}

\def\app#1#2{%
	\mathrel{%
		\setbox0=\hbox{$#1\sim$}%
		\setbox2=\hbox{%
			\rlap{\hbox{$#1\propto$}}%
			\lower1.1\ht0\box0%
		}%
		\raise0.25\ht2\box2%
	}%
}

\newcommand{\tw}{\textwidth}

\newcommand{\vs}{\varsigma}

\newcommand{\inv}{^{-1}}

\newcommand{\ope}\odot

\usepackage{manfnt}

\newcommand{\bi}{\begin{itemize}}
	\newcommand{\ei}{\end{itemize}}

\newcommand{\igpfc}[1]{\vcenter{\hbox{\includegraphics[width=.5\textwidth]{#1}}}}
\newcommand{\igptc}[1]{\vcenter{\hbox{\includegraphics[width=.3\textwidth]{#1}}}}
\newcommand{\igptfc}[1]{\vcenter{\hbox{\includegraphics[width=.35\textwidth]{#1}}}}
\newcommand{\igptwoc}[1]{\vcenter{\hbox{\includegraphics[width=.2\textwidth]{#1}}}}
\newcommand{\igpfoc}[1]{\vcenter{\hbox{\includegraphics[width=.45\textwidth]{#1}}}}

\usepackage{amsthm}

\newcommand\bd            {\begin{definition}}
	\newcommand\ed            {\end{definition}}
\newcommand\bt            {\begin{theorem}}
	\newcommand\et            {\end{theorem}}
\newcommand\be            {\begin{equation}}
	\newcommand\ee            {\end{equation}}
\newcommand\ba            {\begin{aligned}}
	\newcommand\ea            {\end{aligned}}
\newcommand\bea{\begin{equation}\begin{aligned}}
		\newcommand\eea{\end{aligned}\end{equation}}

\usepackage[dvipsnames]{xcolor}

\usepackage{hyperref} 
\hypersetup{final}
\hypersetup{colorlinks, citecolor=red, linkcolor=red, urlcolor=red} 

\newcommand{\sss}{\subsubsection}
\renewcommand{\ss}{\subsection}

\renewcommand{\a}{\alpha}

\renewcommand{\d}{\delta}

\newcommand{\g}{\gamma}
\newcommand{\G}{\Gamma}
\newcommand{\s}{\sigma}

\newcommand{\ep}{\varepsilon} %
\renewcommand{\l}{\lambda}
\renewcommand{\L}{\Lambda}
\renewcommand{\t}{\theta}

\renewcommand{\o}{\omega}

\renewcommand{\r}{\rho}

\renewcommand{\c}{\chi}
\newcommand{\z}{\zeta}

\newcommand{\bfg}{\mathbf{g}}
\newcommand{\bfh}{\mathbf{h}}

\newcommand{\zt}{\mathbb{Z}_2}
\newcommand{\zn}{\mathbb{Z}_N}

\newcommand{\cc}{\mathbb{C}}

\newcommand{\EE}{\mathbb{E}}

\newcommand{\nn}{\mathbb{N}}

\newcommand{\qq}{\qquad}

\newcommand{\zz}{\mathbb{Z}}

\newcommand{\mcc}{\mathcal{C}}

\newcommand{\mcf}{\mathcal{F}}
\newcommand{\mco}{\mathcal{O}}
\newcommand{\mcu}{\mathcal{U}}

\newcommand{\mcg}{\mathcal{G}}
\newcommand{\mcs}{\mathcal{S}}

\newcommand{\mch}{\mathcal{H}}

\newcommand{\mcp}{\mathcal{P}}
\newcommand{\mcx}{\mathcal{X}}

\newcommand{\mcm}{\mathcal{M}}

\newcommand{\mcw}{\mathcal{W}}

\newcommand{\mcr}{\mathcal{R}}
\newcommand{\mcq}{\mathcal{Q}}

\newcommand{\sfP}{\mathsf{P}}

\usepackage[mathscr]{eucal} %

\newcommand{\scd}{\mathscr{D}}

\newcommand{\kb}[2]{|{#1}\rangle\langle{#2}|}
\renewcommand{\k}[1]{|#1\rangle}
\newcommand{\ket}[1]{|#1\rangle}
\newcommand{\proj}[1]{|#1\rangle\langle#1|}
\newcommand{\bra}[1]{{\langle #1}|}

\usepackage{dcolumn}

\usepackage{pifont}
\usepackage{ulem}

\usepackage{enumitem}

\usepackage{amsthm}
\usepackage{bm}
\renewcommand{\bot}{\bigotimes}

\DeclarePairedDelimiterX{\infdivx}[2]{(}{)}{%
	#1\;\delimsize\|\;#2%
}

\newcommand{\spt}{\mathsf{SPT}}
\newcommand{\ssb}{\mathsf{SSB}}
\newcommand{\para}{\mathsf{PARA}}

\newtheorem{theorem}{Theorem} %
\newtheorem*{theorem*}{Theorem}
\newtheorem{prop}{Proposition}
\newtheorem{corollary}{Corollary}
\newtheorem{claim}{Claim}
\newtheorem{lemma}[theorem]{Lemma}
\newtheorem{definition}{Definition}

\renewcommand\qq{\qquad}

\renewcommand{\bot}{\bigotimes} 

\newcommand{\maj}{\mathsf{maj}}

\usepackage{accents}
\newcommand{\ub}[1]{\underline{#1}}
\newcommand{\ut}[1]{\underaccent{\sim}{#1}} 
\newcommand{\rem}{\mathsf{rem}}

\begin{document}

	\title{Exact Quantum Algorithms for Quantum Phase Recognition:\\ Renormalization Group and Error Correction }
	
	\author{Ethan Lake}
	\affiliation{Department of Physics, Massachusetts Institute of Technology, Massachusetts Institute of Technology, Cambridge, MA 02139, USA}
	
	\author{Shankar Balasubramanian}
	\affiliation{Department of Physics, Massachusetts Institute of Technology, Massachusetts Institute of Technology, Cambridge, MA 02139, USA}
	\affiliation{Center for Theoretical Physics, Massachusetts Institute of Technology, Massachusetts Institute of Technology, Cambridge, MA 02139, USA}
	
	\author{Soonwon Choi}
	\affiliation{Department of Physics, Massachusetts Institute of Technology, Massachusetts Institute of Technology, Cambridge, MA 02139, USA}
	\affiliation{Center for Theoretical Physics, Massachusetts Institute of Technology, Massachusetts Institute of Technology, Cambridge, MA 02139, USA}

    \preprint{MIT-CTP/5498}
	\begin{abstract}

		We explore the relationship between renormalization group (RG) flow and error correction by constructing quantum algorithms that exactly recognize 1D symmetry-protected topological (SPT) phases protected by finite internal Abelian symmetries.
		For each SPT phase, our algorithm runs a quantum circuit which emulates RG flow: an arbitrary input ground state wavefunction in the phase is mapped to a unique minimally-entangled reference state, thereby allowing for efficient phase identification. This construction is enabled by viewing a generic input state in the phase as a collection of coherent `errors' applied to the reference state, and engineering a quantum circuit to efficiently detect and correct such errors. Importantly, the error correction threshold is proven to coincide exactly with the phase boundary. We discuss the implications of our results in the context of condensed matter physics, machine learning, 
		and near-term quantum algorithms.
	\end{abstract}
	
	\maketitle
	
	\begin{figure}
		\centering
		\includegraphics[width=.48\tw]{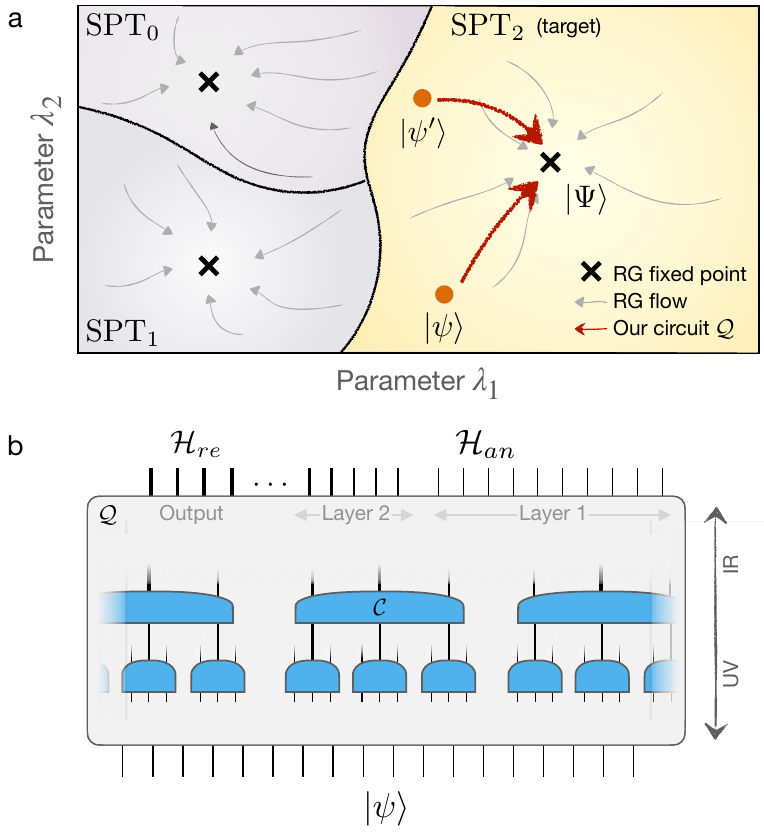} 
		\caption{\label{fig:qpr_and_rg}
			{\bf Phase recognition and RG flow}. {\bf a}.
			A schematic of a phase diagram drawn as a function of two arbitrary parameters $\l_{1,2}$. 
			Our goal is to determine whether or not a state $\k{\psi}$ drawn from the phase diagram belongs to a given target phase, e.g. $\spt_2$ above.
			Each $\k{\psi}$ can be written as a finite-depth symmetric quantum circuit $\mcu$ acting on a universal state $\k{\Psi}$, which serves as the representative wavefunction for the target phase. 
            We design a quantum circuit $\mcq$ which `undoes' the action of $\mcu$ for all $\k{\psi}$ in the target phase, by inducing an RG flow from $\k{\psi}$ to $\k{\Psi} \tp \k{\phi_{an}}$, where $\k{\Psi}$ is an RG fixed-point reference state, and $\k{\phi_{an}}$ encodes the non-universal correlations in $\k{\psi}$. {\bf b}. $\mcq$ is built from a hierarchical array of unitary operators $\mcc$, chained together in a way reminiscent of MERA. Each $\mcc$ acts on three sites: the center output leg of $\mcc$ is a `renormalized' degree of freedom that is passed on to the next layer of the circuit, while the remaining `ancilla' legs store the non-universal information in $\k{\psi}$, and do not participate in subsequent layers. This divides the output of each layer of $\mcq$ into two subspaces: renormalized legs $\mch_{re}$ on which $\mcq$ imprints the result of the phase identification algorithm, and ancilla legs $\mch_{an}$, which contain the data that distinguishes different wavefunctions in the target phase.  } 
	\end{figure}
	
	\section{Introduction} 
	
	The renormalization group (RG) is one of the most important concepts to arise in theoretical physics over the past century \cite{kadanoff1966scaling,wilson1972critical,wilson1975renormalization,gell1954quantum}.  Given a microscopic description of a many-body system, the framework of RG provides a systematic way of extracting the universal features capturing the system's low-energy physics. %
	RG is an indispensable tool in condensed matter physics, where it is used to classify phases of matter. 

	It is interesting to note that there is an important concept in information theory which shares many parallels with RG, namely that of error correction (EC). Writ large, EC is a class of techniques developed to protect information (both classical and quantum) from the effects of noise. The key idea of EC is to hide a small amount of important logical information redundantly in a larger number of physical degrees of freedom, and it shares with RG the ability to separate universal information from unimportant details. Due to this conceptual similarity, there have been several works drawing connections between RG and EC from a variety of perspectives, ranging from the theory of quantum channels to the physics of black holes~\cite{furuya2022real,furuya2021renormalization,gomez2021wilsonian,almheiri2015bulk,pastawski2015holographic,kim2017entanglement,yang2016bidirectional,chen2022exact,cong2019quantum,cong2022enhancing}. 
	However, these connections are  often heuristic, based on empirical findings, or else appear in highly formalized settings. 

	In this work, we establish an explicit connection between RG and EC by presenting a class of quantum circuit algorithms that solve quantum phase recognition (QPR) problems. Given a ground state wavefunction $\ket{\psi}$ and a phase of matter $\sfP$, the QPR problem asks whether or not $\ket{\psi}$ belongs to $\sfP$.
	Our quantum circuits efficiently solve this problem for a large class of gapped 1D symmetry-protected topological (SPT) and symmetry-breaking phases. We do this by implementing a simple form of error correction in a hierarchical multi-scale way, in a manner which emulates RG flow~(Fig.~\ref{fig:qpr_and_rg}a). We show that the convergence of this RG flow coincides {\it exactly} with the phase boundary of $\sfP$, 
	implying that our algorithm faithfully recognizes input wavefunctions even very close to phase boundaries. 
	Our circuit is built from 3-local gates with a depth of at most $O(\log(L))$, and has a substantially smaller sample complexity compared to more conventional approaches that directly measure order parameters, making it a practical solution for near-term quantum hardware.

	Our RG quantum circuits are inspired by the quantum convolutional neural network (QCNN) architecture introduced in \cite{cong2019quantum} and experimentally demonstrated in \cite{herrmann2021realizing}, and have some conceptual similarities with the measurement-based approach of Refs. \cite{bartlett2010quantum,miller2015resource}. For our purposes, a QCNN is a unitary circuit $\mcq$ built from layers of local unitaries chained together in a hierarchical fashion, similar to the multiscale entanglement renormalization ansatz (MERA) \cite{vidal2008class,aguado2008entanglement} and tree tensor networks.
	The hierarchical nature of $\mcq$ originates from an iterative filtering of the Hilbert space which occurs at each layer of the circuit. At a given layer, $\mcq$
	effectively separates the system's Hilbert space into two subspaces, $\mch_{re}$ and $\mch_{an}$. The {\it renormalized subspace} $\mch_{re}$ contains error-corrected degrees of freedom that are further renormalized in subsequent layers, while the {\it ancilla subspace} $\mch_{an}$ stores detailed microscopic information about the detected errors, and the degrees of freedom therein do not participate in further layers (Fig.~\ref{fig:qpr_and_rg}b).
	
	The nontrivial task we address in this paper is to design a single universal circuit $\mcq$ that performs phase recognition on {\it all} input states $\k\psi$.
	We accomplish this by addressing three challenges:
	(i) constructing a local EC protocol that isolates non-universal information present in $\k\psi$ and deposits it in $\mch_{an}$,
	(ii) ensuring any input state in $\sfP$ flows under our circuit to a unique fixed point $\k{\Psi_\sfP}$, while simultaneously ensuring that any state outside of $\sfP$ does {\it not} flow to $\k{\Psi_\sfP}$, and 
	(iii) efficiently performing phase recognition given the converged fixed point wavefunction.
	As we show, these desiderata can be all satisfied by a simple analytic construction that does not require any numerical optimization or training. 
	
	We emphasize that in the scenario under consideration, the input to the QPR problem is a quantum wavefunction $\k\psi$ --- which may be handed to us from a physical experiment --- rather than a tensor network or other efficient classical description thereof. 
	This means that we are not allowed to perform any nonlinear data processing that requires a classical description of an input state, such as probing the state's entanglement structure, or performing tensor-network renormalization~\cite{pollmann2010entanglement,pollmann2012detection,gu2009tensor,levin2007tensor,evenbly2015tensor}.  
	While a classical description can always be obtained using quantum state tomography, such an approach will generically require a large sample complexity, and is conceptually distinct from the RG-inspired treatment we propose.
	
	\section{Recognizing 1D Symmetry-Protected Topological phases}
	
	\ss{Overall strategy}
	
	Before presenting technical results, we briefly review the definition of gapped phases in 1D, establish notation, and summarize the overall strategy of our approach.
	
	Two gapped local Hamiltonians are said to be in the same phase if and only if their ground states $\k{\psi}$ and $\k{\psi'}$ are related by a constant-depth local circuit $\mcu$ satisfying $\k{\psi'} = \mcu\k{\psi}$ up to negligible corrections~\cite{albash2018adiabatic}.
	Each gapped phase $\sfP$ can be thus defined as an equivalence class with a representative element $\k{\Psi_\sfP}$, where all other ground states in $\sfP$ can be written as $\k{\psi} = \mcu_\psi \k{\Psi_\sfP}$ for some constant-depth circuit $\mcu_\psi$ that depends on the microscopic details contained in $\ket{\psi}$. 
	Here we focus on a class of Hamiltonians possessing an internal unitary Abelian symmetry $G$, which is unbroken in their ground states~\footnote{A similar approach is possible for cases in which $G$ is spontaneously broken; see App. \ref{app:symm_breaking} for the details.}.
	In this restricted setting, {\it symmetric} local circuits similarly define distinct gapped phases of matter, called symmetry-protected topological (SPT) phases~\cite{chen2013symmetry,huang2015quantum,schuch2011classifying,senthil2015symmetry}.
	
	The classification of SPT phases in 1D has by now been fully worked out, and can be understood from the entanglement structure of the ground state wavefunction. 
	Consider acting with a symmetry transformation on some subregion $A$. 
	In non-trivial SPT phases, the wavefunction remains unchanged on the interior of $A$, but the quantum correlations between $A$ and its complement become `twisted' near the boundaries of $A$, creating a pair of excitations that cannot be eliminated locally in a symmetry-preserving way. 
	Each of these excitations individually transform under a {\it projective} representation $V^{(\o)}_g$ of $G$. While $V^{(\o)}_g$ is in general not a linear representation of $G$, the combined symmetry action on both excitations, i.e. $V^{(\o)*}_g \tp V^{(\o)}_g$, is always a linear representation.
	Thus, the excitation associated with $V^{(\o)}_g$ can be loosely thought of as carrying a `fractionalized' symmetry charge.
	The distinct choices of $V^{(\o)}_g$ are determined by algebraic objects $\o$ known as second cohomology classes of $G$, and these in turn enumerate the distinct SPT phases, which we denote by $\spt_\o$.
	
	Given a symmetry $G$ and a target phase $\spt_\o$, we construct our phase recognition algorithm in three steps:
	
	{\it Step 1}: we choose a specific reference state $\k{\Psi_\o}$ for $\spt_\o$. This reference state can be constructed directly using the matrices $V^{(\o)}_g$ and satisfies a number of desired conditions, to be elaborated on below.
	
	{\it Step 2}: we identify a set of \textit{indicator observables} $\mcs^{(\o)}_g$ that exactly recognize the desired reference state $\ket{\Psi_\o}$, while rejecting reference states for other phases:
	\be \label{string_stab} \bra{\Psi_{\o}} \mcs^{(\o')}_g \k{\Psi_{\o}} = \delta_{\o,\o'}.\ee 
	In the present setting, the indicator observables can be chosen to be an appropriate set of string order parameters~ \cite{pollmann2012detection}.
	
	{\it Step 3}: we construct a hierarchical unitary circuit $\mcq_\o$ 
	(Fig.~1b) such that after a sufficiently large number of layers, applying $\mcq_\o$ to any input state in the target phase $\ket{\psi} \in \spt_\o$ yields $\ket{\Psi_\o} \tp \k{\phi_{an}} \in \mch_{re} \tp \mch_{an},$ where $\k{\phi_{an}}$ stores the non-universal information present in $\k{\psi}$. In essence, the role of $\mcq_\o$ is to approximately `undo' the action of the circuit $\mcu_\psi$ appearing in $\ket{\psi} = \mcu_\psi \ket{\Psi_\o}$, and to do so without explicit knowledge of $\mcu_\psi$ (Fig.~1a, red arrows).
	Note that under our circuit, all such $\k{\psi} \in \spt_\o$ flow to the fixed point $\k{\Psi_\o}$ without violating unitarity: this is made possible by the presence of $\mch_{an}$, which serves as a bath that stores all non-universal information about the input state.
	At the same time, we design $\mcq_\o$ such that for all input states in distinct phases $\k{\psi_{\o'}}\in \spt_{\o'}$ with $\o'\neq\o$, the output state on $\mch_{re}$ never flows to $\k{\Psi_{\o}}$, thus preventing `false positives' from occurring. 
	
	Combined together, our goal is to ensure that $\lan\psi | \mcq_\o^\da (\mcs^{(\o)}_g \tp \unit_{an}) \mcq_\o |\psi\ran$ is asymptotically close to $1$ if $\k\psi\in \spt_\o$, and asymptotically close to zero otherwise.
	If this is achieved, the sharp difference in the value of the indicator observable in the two cases enables us to efficiently answer the QPR problem with high confidence, using only a small number of measurements.
	In other words, running the circuit $\mcq_\o$ and subsequently measuring $\mcs^{(\o)}_g$ constitutes an exact quantum algorithm to recognize the phase $\spt_\o$, and it does so with low sample complexity.

	We should note that one could always simply directly measure the indicator observables in the input state, without first acting with the circuit $\mcq_\o$.
	While in principle such an approach will generically work, the number of measurements required to get sufficiently good statistics can be very large, especially when input states lie close to phase boundaries (where the expectation values of $\mcs^{(\o)}_g$ often become vanishingly small). 
	Our algorithm (namely the application of the RG circuit $\mcq_\o$) eliminates this problem.

	The success of our approach hinges on the construction of 
	$\mcq_\o$, which we perform using ideas inspired by error correction.
	Specifically, we take the viewpoint that any input state in $\spt_\o$ 
	should be thought of as a collection of coherent `errors' applied to the reference state.
	The precise types of errors that arise are identified in Theorem \ref{maintext_theorem} below, which provides the foundation needed to design $\mcq_\o$.
	Intuitively, our Theorem states that for ground states of SPT phases, the errors can be represented as collections of operators that create pairs of (virtual) fractionalized excitations (this perspective is closely related to the operational viewpoint on topological phases adopted in Refs. \cite{jamadagni2022operational,jamadagni2022error}).
	As a single local error results in a pair of such excitations, 
	one can leverage this redundancy as a primitive for error correction.
	Specifically, by detecting fractionalized excitations supported on parts of a system, one can identify occurrences of local errors, and correct for their effects in the remaining system.
	Our RG circuit $\mcq_\o$ is constructed based on this inuition while still evolving the wavefunction in a unitary way. 
	
	{}

	\subsection{Reference states and indicator observables }
	
	Our main technical idea is to utilize tensor network methods to construct both $\k{\Psi_\o}$ and $\mcq_\o$.
	For concreteness, let us consider a periodic 1D chain consisting of $L$ qudits.
	As inputs to the QPR problem, we are given the symmetry group $G$, its unitary representation $R_g$, $g\in G$, a target phase $\spt_\o$, and the phase's associated projective representation $V^{(\o)}_g$.
	To simplify our presentation, we will focus on a subclass of phases whose projective representations are {\it maximally non-commutative}, meaning that for all $g \neq e \in G$, $V_g^{(\o)}$ fails to commute with at least one $V_h^{(\o)}$~\cite{else2012symmetry,stephen2017computational,stephen2017computational2}.
	Most familiar SPT phases satisfy this property (such as the Haldane phase); 
	an example of one which does not is the nontrivial SPT phase protected by $G = \zz_4 \times \zt$ \cite{stephen2017computational}. 
	The generalization of our construction to arbitrary Abelian $G$-SPT phases requires adding an extra step prior to running our circuit but does not involve any qualitatively new ideas; the details are given in App. \ref{app:non_mnc}.
	Finally, to simplify notation we will also assume that the local Hilbert space $\mch$ has dimension $\dim(\mch) = |G|$ and is enumerated by a set of basis states $\k{g}, g\in G$.
	Under these assumptions, $R_g$ is always isomorphic to the tensor product $V^{(\o)*}_g \tp V^{(\o)}_g$, which makes the pattern of symmetry fractionalization manifest. 
	
	\begin{figure} 
		\includegraphics[width=.49\tw]{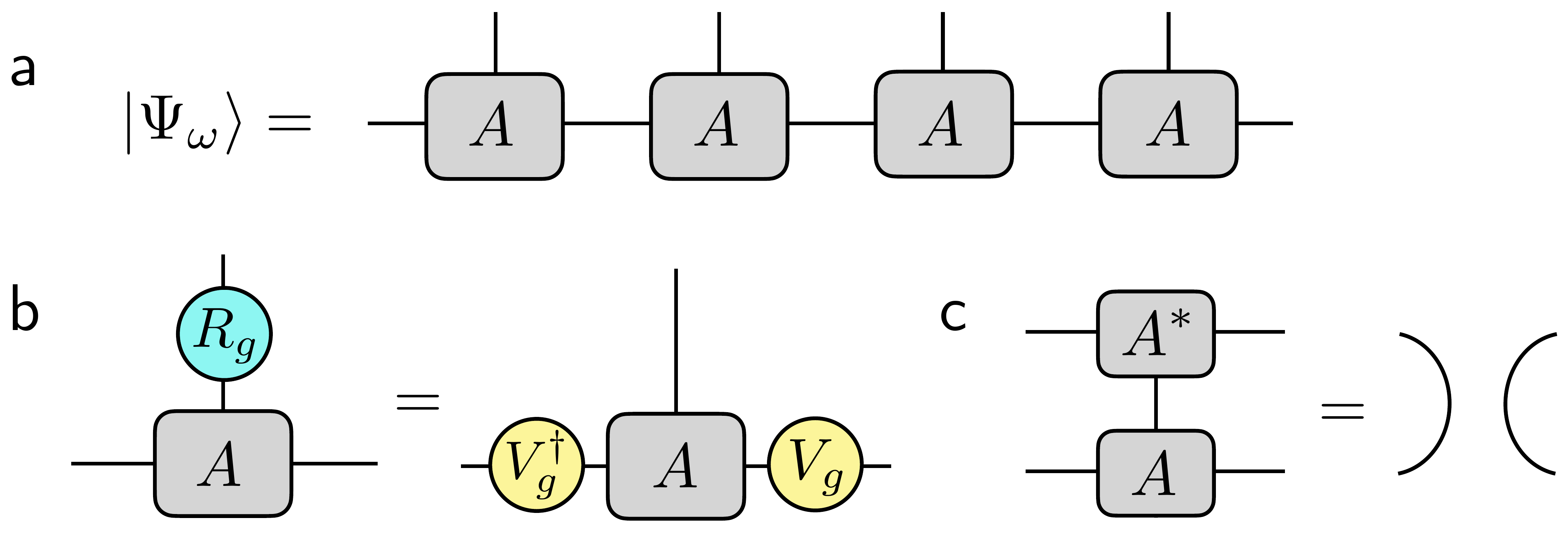}
		\caption{\label{fig:symfrac_and_iso}
			{\bf Diagrammatic representations of the reference state and tensor identities.} 
			{\bf a}. The reference state $|\Psi_\omega\ran$ can be written as an infinite MPS formed by contracting tensors $A$ (gray boxes). 
			{\bf b}. Symmetry fractionalization: a local action of the symmetry $R_g$ `pushes through' the MPS tensor to become two projective actions at the virtual level. 
			{\bf c}. Zero correlation length: the transfer matrix of the reference MPS is a trivial rank-1 projector, i.e. the MPS is both left- and right-canonical.} 
	\end{figure} 
	
	A particularly nice choice of reference states $\k{\Psi_\o}$ can be written down in a matrix product state (MPS) form as  
	\be
	\label{psi_canonical}
	\k{\Psi_\o} =  \sum_{\{g_i\}}  \Tr[A^{[g_1]}\cdots  A^{[g_L]}] \k{g_1, \cdots, g_L},
	\ee 
	where each $g_i$ runs over all $|G|$ basis states, and $A^{[g]}$ is a 
	$\sqrt{|G|}\times \sqrt{|G|}$ matrix which implements the isomorphism between $R_g$ and $V^{(\o)*}_g \tp V^{(\o)}_g$. $A^{[g]}$ is fully 
	determined by $R_g$ and $V_g^{(\omega)}$:
	\be A^{[g]} = V^\da_{\s(g)},\ee 
	where $\s$ is a permutation on the group elements of $G$ defined through the relation $\lan h| R_{\s(g)} |h \ran  =   \Tr[V_h V_g V_h^\da V_g^\da]/\sqrt{|G|}$ for all $g,h\in G$.

	This choice of $A$ has the merit of satisfying two important properties: (i) zero correlation length: $\sum_g A^{[g]}_{\alpha \beta} (A^{[g]}_{\alpha'\beta'})^* = \delta_{\alpha \alpha'} \delta_{\beta \beta'}$; and 
	(ii) symmetry fractionalization: $\sum_{k} (R_g)_{hk} A^{[k]} = (V_g^{(\omega)})^\dagger A^{[h]} V_g^{(\omega)}$ as proven in App~\ref{app:mps_technology}.
	These two identities are illustrated graphically in Fig. \ref{fig:symfrac_and_iso}.

	$\k{\Psi_\o}$ satisfies a number of additional properties that make it a natural choice for serving as a reference wavefunction.
	First, the conditions (i) and (ii) imply that $\ket{\Psi_\omega}$
	can be understood as a minimally-entangled fixed-point wavefunction in the sense of Ref.~\cite{verstraete2005renormalization}.
	Second, one can find a frustration-free parent Hamiltonian of $\ket{\Psi_\omega}$ consisting only of mutually commuting nearest neighboring interactions~\cite{schuch2011classifying}.
	Finally, $\k{\Psi_\o}$ is a simultaneous eigenstate of a collection of the aforementioned {\it string operators} $\mcs^{(\o)}_{g;i\ra j}$, which implement the action of $R_g$ between the sites $i,j$ and perform a projective action of $g$ at sites $i,j$:
	\be \label{string_def} \mcs^{(\o)}_{g;i\ra j} = S^R_{g,i} \tp \bot_{l=i+1}^{j-1} R_{g,l} \tp S^L_{g,j},\ee
	where the explicit form of the matrices $S^{L/R}_{g,i}$ --- which decorate the string ends with projective symmetry actions --- depends on $\o$, and can be found in App. \ref{app:mps_technology}.
	This allows us draw a connection with stabilizer quantum error correcting codes: $\ket{\Psi_\omega}$ is the unique logical state stabilized by the string operators, which therefore serve as a natural choice of indicator observables. 
	With these choices, the reference states for distinct SPT phases satisfy the desired selection rule Eq.~\eqref{string_stab}~\cite{pollmann2012detection,else2013hidden,de2021symmetry}.

	\begin{figure}
		\centering
		\includegraphics[width=.49\tw]{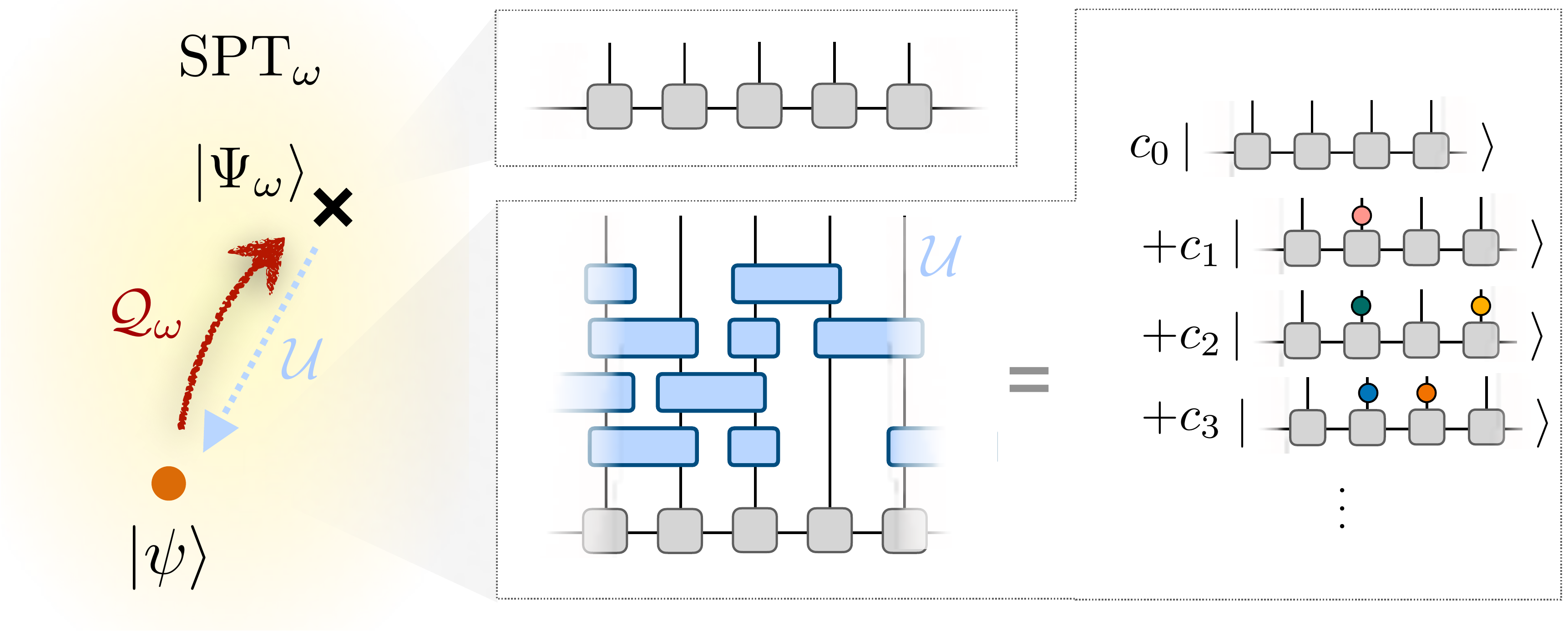}	
		\caption{\label{fig:wfs_and_ec} 		
			{\bf SPT ground states: }
			Any input state $\k\psi \in \spt_\o$ can be written as $\mcu \k{\Psi_\o}$ for a constant-depth symmetric local circuit $\mcu$. The action of $\mcu$ on $\k{\Psi_\o}$ can be decomposed as a linear combination of local symmetry actions (colored circles; inset), which we view as `errors' acting on $\k{\Psi_\o}$. }
	\end{figure}
	
	\subsection{Constructing the RG circuits} 
	
	Having determined the representative wavefunction and indicator observables, we turn to the most non-trivial step in our construction: finding the RG circuit $\mcq_\o$.
	To do so, we first show that any $\ket{\psi}\in \spt_\o$ can be understood as $\ket{\Psi_\o}$ locally deformed by a collection of `errors', which take the form of products of local symmetry actions applied at each site (see Fig.~\ref{fig:wfs_and_ec}):
	\begin{theorem} \label{maintext_theorem} 
		A generic ground state wavefunction $\k{\psi}\in \spt_\o$ may be written as 
		\be \label{psio} \k{\psi} = \sum_{\bfg \in G^L} C_\bfg R_\bfg \k{\Psi_\o},\qq R_\bfg \equiv \bot_i R_{g_i},\ee 
		where the complex coefficients $C_\bfg$ characterize the  amplitude for a collection of errors labeled by the string of group elements $\bfg = (g_1,\dots,g_L)$ to occur,
		with the errors being correlated over a finite distance set by the correlation length $\xi$  of $\k{\psi}$.  
	\end{theorem}	
	This Theorem strongly constrains the types of errors that our circuits need to correct, namely only those errors caused by local symmetry operators $R_g$. 
	Since all of the $R_g$ commute in our case, the error correction process we will design will be essentially classical in nature: this is made possible by the presence of the symmetry $G$, which prevents `bit flip' errors ($i.e.$ those off-diagonal in our choice of basis) from arising.  
	
	The proof of this Theorem proceeds by first studying the general structure of symmetric finite-depth quantum circuits  $\mcu_\psi$ and finding an appropriate basis of operators in which to decompose the action of $\mcu_\psi$ on $\ket{\Psi_\o}$; the symmetric nature of $\mcu_\psi$ and the special fixed-point properties of $\k{\Psi_\o}$ guarantee that only $R_\bfg$ operators appear in the resulting decomposition. The precise statement of our Theorem and details of the full proof can be found in App.~\ref{app:general_ground_states}.

	Since quantum circuits are linear, we can focus our attention on correcting each term in the sum Eq.~\eqref{psio} separately. We will also restrict our attention to distributions of errors which are generic, meaning that there are no fine-tuned degeneracies among the coefficients $C_\bfg$. Adversarially chosen wavefunctions for which errors associated with distinct group elements occur with {\it exactly} equal probabilities will be able to fool our RG circuits, but since such wavefunctions constitute a set of measure zero in the space of all SPT ground states, we will not be concerned with this issue. 
	
	\begin{figure} 
		\includegraphics[width=.49\tw]{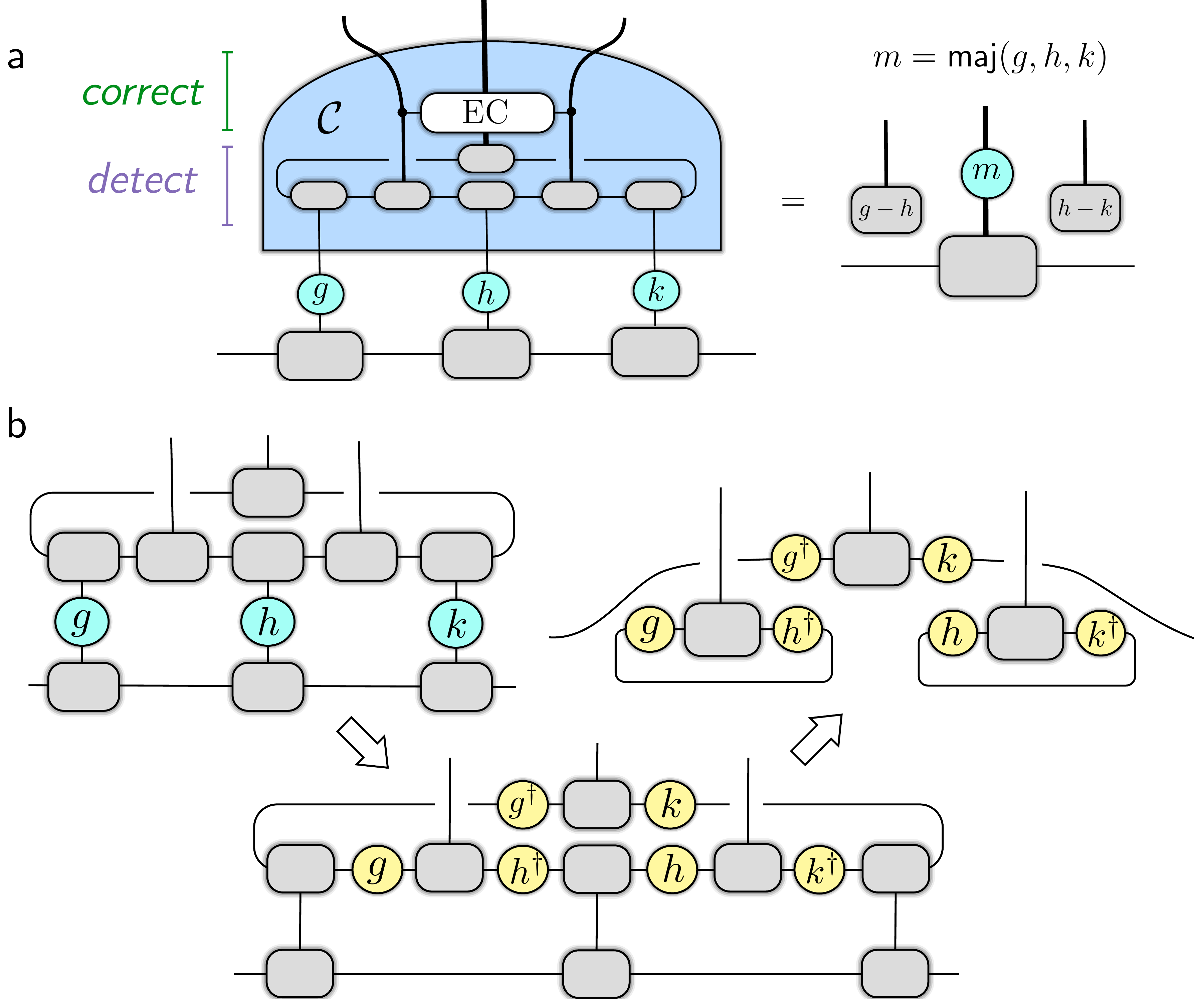} 
		\caption{\label{fig:mcc}  {\bf a.} The elementary unitary $\mcc$ is designed to detect and correct errors (cyan circles), which it filters out by way of a majority-vote scheme. $\mcc$ consists of six MPS tensors chained together in a ring-like configuration (which detects the errors) together with a control gate EC (which corrects them). 
			$\mcc$ is designed such that when it acts on a region containing three symmetry actions $R_{g,h,k}$, the renormalized output (middle leg) retains its original MPS form together with a symmetry action $R_{\maj(g,h,k)}$ of the majority vote, while the ancilla outputs (left and right legs) become decoupled from the rest of the system.
			{ \bf b.} The error detection stage is accomplished in two steps. First, the errors $R_{g,h,k}$ are fractionalized into projective representations (yellow circles) using Fig. \ref{fig:symfrac_and_iso}b. In the next step, the zero-correlation length property is used (Fig. \ref{fig:symfrac_and_iso}c) to decouple the renormalized leg from the ancilla legs. The renormalized leg incurs an error determined by $V_g^{(\o)\da} \tp V_k^{(\o)}$, while the ancilla legs store the `domain wall' variables $g-h, h-k$ that are used to diagnose the error. } 
	\end{figure} 
	
	A circuit $\mcq_\o$ which corrects generic errors can in fact easily be constructed; see Fig.~\ref{fig:qpr_and_rg}b for a sketch of the overall architecture.
	Each layer of $\mcq_\o$ consists of parallel applications of a certain 3-qudit unitary gate $\mcc$ (Fig.~\ref{fig:qpr_and_rg}b, blue boxes).
	$\mcc$ consists of two components, which `detect' and `correct' errors respectively, as shown in Fig.~\ref{fig:mcc}a.
	The `detect' portion of $\mcc$ is constructed by chaining together a ring of six reference-state MPS tensors $A^{[g]}$. 
	Fig.~\ref{fig:mcc}b illustrates diagrammatically how the ring of six $A^{[g]}$ tensors performs error detection. By using the symmetry fractionalization property of $A^{[g]}$ (Fig.~\ref{fig:symfrac_and_iso}b), a trio of errors $R_g\tp R_h \tp R_k$ (cyan circles) can be `pushed up' into the ring, where they split into pairs of projective representations (yellow circles).
	Using the zero correlation length property (Fig.\ref{fig:symfrac_and_iso}c), one finds that the left and right output qudits then become isolated from the central (renormalized) qudit, which remains entangled with its neighbors. 
	By measuring only the left and right qudits, one can then determine how to eliminate the errors remaining on the center leg. 
	Specifically, under the error $R_g\tp R_h \tp R_k$, the action by $\mcc$ leaves the ancilla legs in a product state $\approx\k{g-h}\tp \k{h-k}$ determined by the `domain wall' variables $g-h,h-k$, while the state of the renormalized leg incurs an action by the operator $V^{(\o)\da}_g \tp V_k^{(\o)}$.
	As a simple example, consider the case where $g=h=k$ are all identical.
	In this case the action of $R_g \tp R_h \tp R_k$ is locally indistinguishable from that of a global symmetry operation, which should accordingly not be treated as a local error. Indeed, the action of $\mcc$ leaves the ancilla legs in the state $\k{0}\tp \k{0}$, indicating that no EC needs to be performed. The renormalized leg on the other hand incurs an action of $V_g^{(\o)\da} \tp V^{(\o)}_g$, which by the symmetry fractionalization property is equivalent to an action by $R_g$. 
	This property in turn implies that our circuit is symmetric, in the sense that a global symmetry action passes through $\mcc$ to become a symmetry action on only $\mch_{re}$: $\mcc R_g^{\tp 3} = (R_g \tp \unit_{an}) \mcc$.
	
	Now consider what happens when an error occurs, $i.e.$ when not all of $g,h,k$ are equal. In this case, we correct the error by way of a majority voting scheme. This is performed by a control gate denoted by the gate marked `EC' in Fig.~\ref{fig:mcc}a. This gate, whose explicit form is given in App. \ref{app:qcnn_details}, performs an action controlled by the ancilla state $\k{g-h}\tp\k{h-k}$, and results in the renormalized leg being left in a state acted on by $R_{\maj(g,h,k)}$, where $\maj(g,h,k)$ is the majority vote of $g,h,k$ (which we define to equal $h$ in the case of ties). 

	Taking the system size $L = 3^l$, a single layer of $\mcc$ unitaries consequently act as 
	\be \mcc^{\tp 3^{(l-1)}} R_\bfg \k{\Psi_\o}_l = R_{\maj(\bfg)} \k{\Psi_\o}_{l-1} \tp \k{\phi_{an}(\bfg)},\ee 
	where $\k{\Psi_\o}_l$ denotes the reference state on a system of size $3^l$, $\maj(\bfg) = (\maj(g_1,g_2,g_3),\maj(g_3,g_4,g_5),\dots)$, and $\k{\phi_{an}(\bfg)}$ stores information about the errors detected by $\mcc^{\tp 3^{l-1}}$.

	\ss{Diagnosing convergence of the RG flow}
	
	For a given depth $d$, our RG circuit $\mcq_\o^{(d)}$ is constructed by applying $d$ successive layers of $\mcc$ unitaries to the renormalized degrees of freedom in a hierarchical fashion, as indicated in Fig.~\ref{fig:qpr_and_rg}b.
	Subsequent rounds of majority voting dilute the concentration of errors present in the renormalized wavefunction, so that the density of errors decays exponentially when $d \gtrsim \log_3(\xi)$, with $\xi$ the correlation length of $\k{\psi_\o}$ (more detailed arguments establishing this will be given below). 
	For sufficiently large $d$ the renormalized wavefunction will thus converge to one in which no errors occur. In this limit a single layer of the RG circuit simply disentangles the ancilla qubits from the renormalized ones, thereby mapping $\k{\Psi_\o}$ to a smaller copy of itself, tensored with the trivial product state $\k{0\cdots0}$ --- precisely the behavior desired from a fixed-point reference state. On the other hand, if the input wavefunction $\k{\psi_{\o'}}$ is {\it not} drawn from $\spt_\o$, the expectation value of the indicator observable in the state $\mcq_\o \k{\psi_{\o'}}$ will exactly vanish. This is a simple consequence of the selection rule in Eq.~\eqref{string_stab} and the fact that 
	our circuits are symmetric; see App. \ref{app:convergence} for the proof. Thus, no `false positives' occur, and our circuit faithfully performs phase recognition for this class of SPT phases.
	 
	The convergence of the RG flow can be summarized by saying that $\mcq_\o^{(d)}$ flows any $\k{\psi_\o}\in \spt_\o$ to a copy of the reference state $\k{\Psi_\o}$ supported on the renormalized Hilbert space $\mch_{re}^{(d)}$, tensored with a non-universal state on the ancilla space:
	\be \label{ref_tp_an} \mcq_\o^{(d) } \proj{\psi_\o} \mcq_\o^{(d)\da} \approx \proj{\Psi_\o} \tp \r_{an}\ee 
	with the approximation becoming an equality in the limit of large depth $d$. We will find it helpful to define a renormalized density matrix obtained by tracing out $\mch_{an}$ as 
	\be \label{rhoren} \r_{re,\o}^{(d)} \equiv \Tr_{\mch_{an}}[\mcq^{(d)}_\o \r^{(0)} \mcq^{(d)\da}_\o],\ee
	where $\r^{(0)}$ is the initial density matrix. The convergence of the RG flow means that $\r_{re,\o}^{(d\ra\infty)} = \proj{\Psi_\o}$ if $\r^{(0)}$ is a ground state in the target SPT phase. The fact that $\mcq_{\o}^{(d)}$ can be written as a tree tensor network (i.e. as a tensor product of unitaries supported on blocks of size $3^d$) implies that $\r_{re,\o}^{(d)}$ is obtained from $\r_{re,\o}^{(d-1)}$ through LOCC, where the `local' in LOCC is defined with respect to the sites of $\mch_{re}^{(d)}$. This means that any mixed-state entanglement measure (negativity, distillable entanglement, etc.) is monotonically decreasing with $d$, as entanglement measures monotonically decrease under LOCC by definition (note that this holds regardless of the input state $\r^{(0)}$). This concretely establishes the sense in which our RG flow successively eliminates non-universal entanglement. 
	
	We note in passing that while our focus until now has been on recognizing ground states, the initial state $\r^{(0)}$ need not be pure, and our circuit also recognizes mixed states with well-defined patterns of symmetry fractionalization. More precisely, $\r^{(d)}_{re,\o}$ flows to the fixed point $\proj{\Psi_\o}$ for any input states of the form
	\be \label{mixedrho} \r^{(0)} = \sum_{a,b} c_{a,b} \kb{\psi_{\o,a}}{\psi_{\o,b}},\ee 
	where the $c_{a,b}$ are complex coefficients, and for all $a$, $\k{\psi_{\o,a}} \in \spt_\o$. This follows as a simple consequence of Eq.~\eqref{ref_tp_an}.

 A more direct approach to diagnose the convergence of the RG flow is to measure the distribution of errors detected at each layer of the RG circuit.
	Let $\mch^{(d)}_{an}$ denote the degrees of freedom which become part of the ancilla space at layer $d$ of the circuit. As in Eq.~\eqref{rhoren}, we define the density matrix of these degrees of freedom as 	
	\be \label{rhoan} \r_{an,\o}^{(d)} = \Tr_{\mch \setminus \mch^{(d)}_{an}}[\mcq^{(d)}_\o \r^{(0)} \mcq^{(d)\da}_\o].\ee 
	If $\r^{(0)}$ is a ground state in $\spt_\o$ (or a mixed state of the form Eq.~\eqref{mixedrho}), then $\r_{an,\o}^{(d)}$ will approach the pure state $\proj{0\cdots 0}$ at large $d$, after the errors present in $\r^{(0)}$ have been corrected away. For all other $\r^{(0)}$, $\r_{an,\o}^{(d)}$ at sufficiently large $d$ will instead be an equal weight mixture of all $\k{g}$ at each site. 
	We can quantify this by defining 
	\be\label{ancillapurity} \mcf^{(d)} \equiv \frac{1}{\log_{|G|} \dim \mch_{an}^{(d)}}\sum_i \frac{\Tr[\r_{an}^{(d)} \Pi_{0,i}] - 1/|G|}{1-1/|G|},\ee 
	where $\Pi_{0,i}$ is the projector onto the state $\k 0$ at site $i$. For large $d$, $\mcf^{(d)}$ will approach 1 in the target phase and vanish outside of it (thus faithfully performing phase recognition),and in a subsequent section we will demonstrate numerically that this indeed holds. Note that since $\mcf^{(d)}$ is obtained from the expectation value of a single-site operator, it has the advantage of being simpler to measure than the string operators, which are necessarily supported on multiple sites.  However, unlike for the string operators, there is no selection rule enforcing $\mcf^{(d)}$ to exactly vanish outside of the target phase for all $d$. 
	
	\ss{Convergence of iterated majority vote}

	We now give an intuitive explanation for why errors are generically quickly diluted under successive rounds of majority voting. For the purposes of this discussion we will assume that the errors have finite correlation length and are distributed in a way which is statistically translation-invariant on long scales, an assumption which may fail for distributions of errors that break translation only in thermodynamically large scales ($e.g.$ errors obtained by applying a symmetry operator on one half of the system). 

	With this assumption we are able to rigorously prove that the density of errors is monotonically decreasing with $d$, and provide strong evidence 
	that the convergence is exponentially fast when $d > \log_3(\xi)$.
	Let $p_d(\bfg)$ be the probability for an error $R_\bfg$ to occur at depth $d$ of the RG flow. We claim that $p_{d\ra\infty}(\bfg) = \prod_i \d_{g_i,g_*}$ for some fixed group element $g_*\in G$, which is equivalent to the statement that the iterative majority vote process converges to a definite `winner' $g_*$. 

To argue this, we first note that the treelike structure of $\mcq^{(d)}$ means that as long as $3^d \gg \xi$, the probability distribution at depth $d$ approximately factorizes, $p_d(\bfg) \approx \prod_i p_d(g_i)$, where $p_d(g)$ are the single-site marginals of $p_d(\bfg)$. Intuitively, this can be seen by a simple geometric argument.  After $d$ rounds of majority voting on a given length-$3^d$ string of group elements $\bfg$, the majority $\maj^d(\bfg)$ will on average not be sensitive to changing a small ($\ll 3^d$) number of the $g_i$. Consider then the two-site marginal distribution $p_d(g_i,g_{i+1})$. If $\xi \ll 3^d$, this distribution can be computed approximately by ignoring the group elements located within a distance $\xi$ of the joining point between the two depth-$d$ majority-vote trees used to compute $g_i, g_{i+1}$. Doing so then leads to the distribution factorizing as $p_d(g_i,g_{i+1}) \approx p_d(g_i) p_d(g_{i+1})$, and other marginals can be shown to factorize by similar arguments. 
	To prove convergence we can therefore focus on distributions in which errors occur independently on each site. The proof from this starting point is straightforward, as it is easy to show that the difference between the largest and second-largest values of $p_d(g)$ is monotonically increasing with $d$, reaching a value of $1-\ep$ in a number of layers going roughly as $\log(1/\ep)$. 

    A more direct proof is possible if one makes an additional (reasonable) assumption about the three-site marginals $p_d(g,h,k)$. Let $g_1$ ($g_2$) be the element with the largest (second-largest) single-site marginal probability $p_d(g)$. Define the function 
    \be f_d(g) \equiv \sum_{h\in G} (p_d(g,h,g) - p_d(h,g,h)).\ee 
    It is then straightforward to show that the difference $\d_d \equiv p_d(g_1) - p_d(g_2)$ satisfies $\d_{d+1} - \d_d = f_d(g_1) - f_d(g_2).$
    Thus if we assume that $f_d(g_1) > f_d(g_2)$, then $\d_d$ is monotonically increasing with $d$, implying convergence of the majority vote. This essentially amounts to assuming that for two configurations of errors with the same number of domain walls, the one with more instances of $g_1$ has higher probability of occurring (this can be viewed as an assumption that the majority vote is not `Gerrymandered'). The full details are given in App.~\ref{app:convergence}. 
    
	\subsection{Numerical demonstration and sample complexity} 
	
	\begin{figure}
		\includegraphics[width=.5\tw]{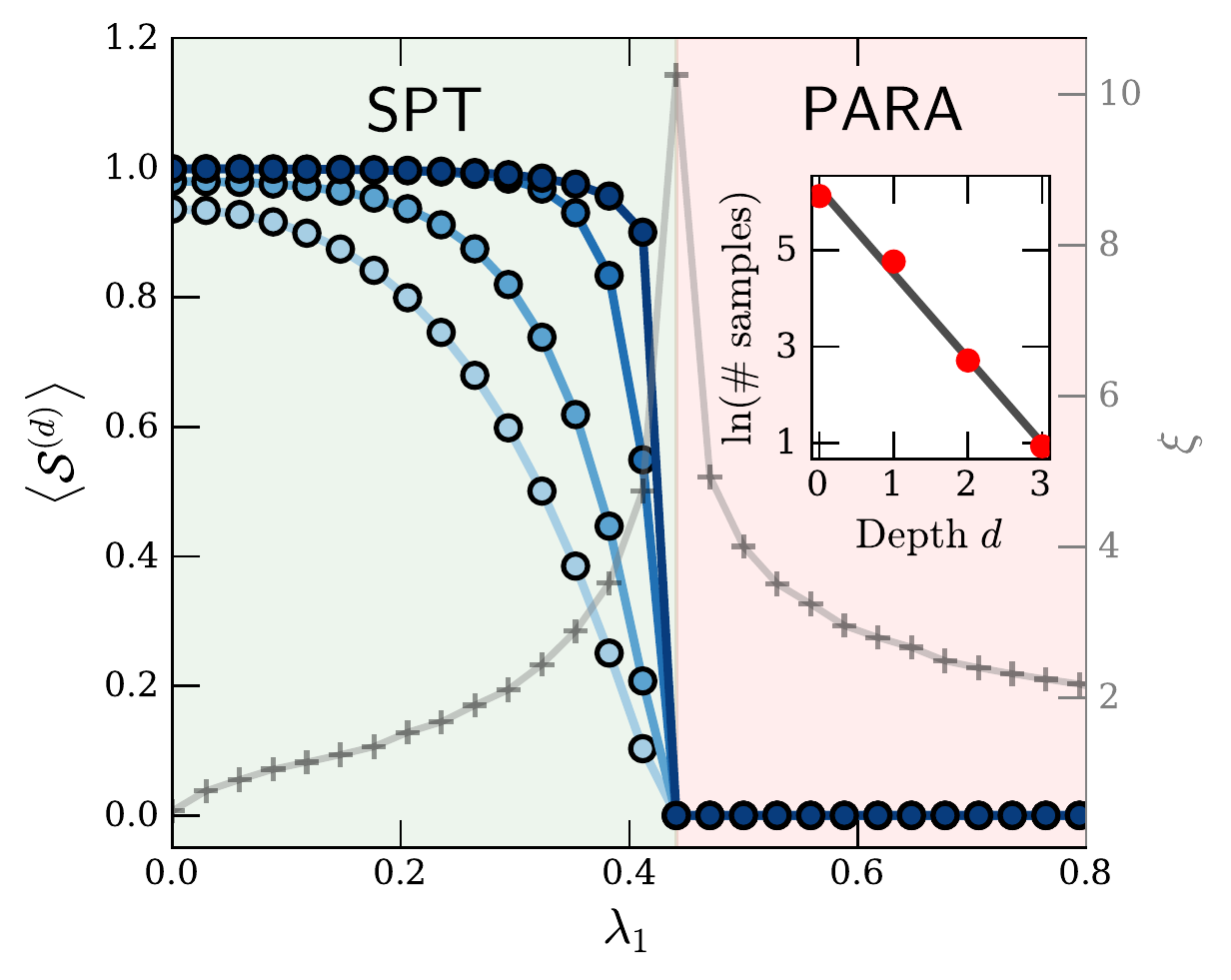}
		\caption{\label{fig:dmrg}
			Expectation value $\langle \mcs^{(d)}\rangle$ of a length $L/2$ string operator for the group element $g=(1,0)$, taken in the ground state of the $\zz_3^2$ cluster state model Eq.~\eqref{znham} on an $L=162$ length chain and with $\l_2 = 0.5$. Different curves indicate different layer depths $d$, increasing from $d=0$ (lightest) to $d=3$ (darkest). 
			Grey crosses indicate the correlation length, demonstrating that $\langle \mcs^{(d)}\rangle$ becomes a sharper diagnostic of the true phase boundary as $d$ is increased. 
			Inset: number of samples needed for phase recognition for $\lambda_1 = 0.4$ taken near the transition, plotted on a semi-log scale. The solid line is a fit to $\#\, {\rm samples} \propto e^{-1.76 d}$.}
	\end{figure}
	
	To assess the performance of our algorithm, we numerically implement our RG circuits by using them to perform QPR on ground states of a $\zn$ cluster state Hamiltonian \cite{tsui2015quantum}, perturbed by onsite and nearest-neighbor longitudinal fields with strengths $\l_1,\l_2$:
	\be \label{znham} H_\o =-\sum_{g\in \zn;i} \( Z_i^{-s_i \o g} X^g_{i+1} Z_{i+2}^{s_i\o g} + \l_1 X_i^g + \l_2 X^g_iX^g_{i+1}\), \ee 
	where $Z = \sum_g e^{\twp  i g /N} \proj g,X = \sum_g \kb{g+1}{g}$ are the usual $\zn$ clock and shift matrices, $s_i \equiv (-1)^i$, and $\o \in \zn$ labels the different SPT phases protected by the symmetry group $G = \zn^2$.
	For large $\l_{1,2}$ the ground state is in the trivial phase, while for small field strengths it is in a nontrivial SPT phase. 
	
	\begin{figure}
		\includegraphics[width=.48\tw]{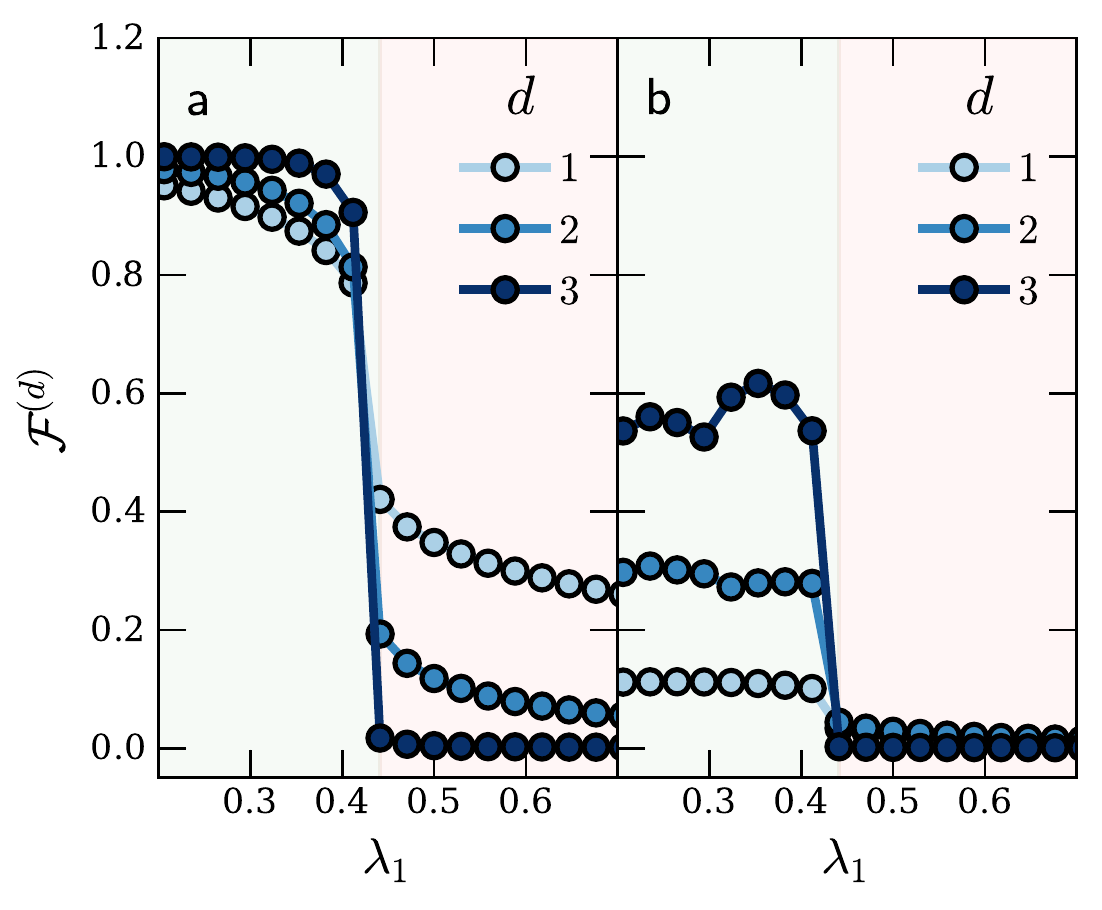}
		\caption{\label{fig:purif} {\bf a.} Error densities for the same ground states and parameter values as considered in Fig. \ref{fig:dmrg}. $ \mcf^{(d)}$ denotes the expectation value Eq.~\eqref{ancillapurity} of the projector onto the error-free subspace of the ancilla qudits at depth $d$ of the circuit, which equals unity at the fixed point $\k{\Psi_\o}$ (where no errors occur), and vanishes in the maximally mixed state (where all errors occur with equal probability). {\bf b.} As in the left panel, but with a single disorder realization of the depth-1 circuit Eq.~\eqref{confusing_circuit} applied to the ground state.  }
	\end{figure}
		
	Taking $N=3,\o=1$ as an example, 	
	we 
	numerically solve for the ground state using DMRG with bond dimension $\chi=200$, 
	to which we directly apply our RG circuit. 
	At each depth $d$ we measure the string operator $\mcs^{(d)}_{(1,0);i\ra j}$, where $|i-j|$ is a third of the coarse-grained system size. The results are shown in Fig.~\ref{fig:dmrg}, where we see that our circuit does a good job of correctly identifying the phase boundary as $d$ is increased. In the inset we plot the {\it sample complexity}, viz. the number of measurements of the string operator needed to determine the phase with high confidence (see App. \ref{app:sample_complexity} for the precise statement). As a function of $d$ we see an exponential improvement in the sample complexity as compared to directly measuring the string operators in the ground state ($d=0$), which is consistent with analytical predictions. 
	
	As discussed above near Eq.~\eqref{rhoan}, an alternate way of diagnosing convergence of our RG flow is to measure the distribution of errors detected at each step. In our numerics we measure the error density by directly calculating the quantity $\mcf^{(d)}$ defined in Eq.~\ref{ancillapurity}, which we show in Fig.~\ref{fig:purif}a for the same ground states as in Fig.~\ref{fig:dmrg}. $\mcf^{(d)}$ is seen to approach a step function with increasing $d$, confirming that $\mcf^{(d)}$ provides a good indicator function for the target phase.

    Our approach also does well at handling error distributions with larger error densities and spatial disorder. More specifically, we note that the ground states of Eq.~\eqref{znham} have rather simple error distributions, with the error density being low everywhere except near the phase boundary. In contrast, we can also consider states with a larger density of errors, with one simple way of producing such states being to act on ground states of Eq.~\eqref{znham} with a random short-depth symmetric circuit. To analyze this numerically, we choose the simple circuit
	\be \label{confusing_circuit} \mcu =  \bot_{i=1}^L {\rm diag}(e^{i\phi_i}, \, e^{i\phi_i},\, e^{-2i\phi_i} ),\ee
	where the angles $\phi_i$ at each site are drawn uniformly from $[0,\twp)$ (recall that we are working in a basis which diagonalizes the onsite symmetry action). The result for a single disorder realization is shown in Fig.~\ref{fig:purif}b. We see that the additional random errors significantly reduce $\mcf^{(d)}$ in the SPT phase, but that the error density decreases quickly with increasing layer depth, demonstrating the robustness of our phase recognition algorithm.

	\subsection{Universal multiscale string operators} \label{sec:msos}
	
	As discussed above, the string operators serve as a set of `indicator observables' that distinguish the different reference states $\k{\Psi_\o}$. The convergence of our RG circuits can thus be summarized in the equation 
	\be \label{string_convergence} \lan \psi_\o | \mcq_\o^{(d)\da} (\mcs^{(\o')}_{g;i\ra j} \tp \unit_{an})\mcq_\o^{(d)}|\psi_\o\ran  \xra{d\gg\log_3(\xi)} \d_{\o,\o'}\ee
	for any $\k{\psi_\o}\in\spt_\o$ (which we saw borne out in the numerics of Fig. \ref{fig:dmrg}).  
	In the `Schr\"odinger' picture, Eq.~\eqref{string_convergence} measures the expectation value of a bare string operator $\mcs^{(\o')}_{g;i\ra j}$ in the renormalized state $\mcq_\o^{(d)}\k{\psi_\o}$.
	It is however also helpful to change to the `Heisenberg' picture, where  
	Eq.~\eqref{string_convergence} can be viewed as the expectation value of a renormalized %
	string operator with respect in the original `UV' state $\k{\psi_\o}$~\cite{cong2022enhancing}. This is formalized by defining 
	the {\it mutliscale string operator} \cite{cong2019quantum} (MSO) as 
	\be \label{mso_def} \wt \mcs^{(\o,d)}_{g;i\ra j} \equiv \mcq_\o^{(d)\da} (\mcs_{g;i\ra j}^{(\o)} \tp \unit_{an} ) \mcq_\o^{(d)} .\ee 
Owing to the hierarchical structure of $\mcq^{(d)}$, the MSO is a linear combination of a number of terms that grows doubly-exponentially in $d$, with each term given by a product of elementary string operators~\cite{cong2019quantum,cong2022enhancing} (see App. \ref{app:qcnn_details} for a recursion relation that explicitly constructs the MSOs).

	Since we have shown that our RG circuit accurately performs phase identification, expectation values of the MSOs serve as a sharp diagnostic of the target SPT phase. That this is possible is rather remarkable, as the MSOs are constructed {\it solely} from the algebraic data that characterize the target SPT phase (a priori, one might expect that the renormalized string operators would depend on microscopic details of the specific state to be recognized). 
	This construction provides a physically transparent way of revealing the universal order present in SPT ground states, since the MSOs implement (arbitrarily accurate) projections onto the space of {\it all} ground state wavefunctions in a target phase. 
	This approach is complementary to other protocols which disentangle the SPT \cite{else2013hidden}, directly measure the projective symmetry action using swap gates \cite{haegeman2012order}, or use symmetry twists to create MSO-like operators in theories with $U(1)$ symmetry \cite{tasaki2018topological,huang2021provably}.

	\section{Discussion and implications}
	
	The goal of this work has been to advocate that renormalization group (RG) flow and error correction (EC) should really be viewed as two sides of the same coin.\footnote{Note that we are not trying to protect any logical information beyond that required to recognize a given phase. Whether or not we can $e.g.$ guarantee that the logical information contained in the degenerate SPT edge modes can always be preserved throughout the RG flow is a question we leave to the future.} 
	We arrived at this perspective by implementing RG flow in a unitary way, via the action of a certain class of quantum circuits. This is rather distinct from conventional approaches to RG, which is usually viewed as an irreversible flow performed on Hamiltonians and Lagrangians. Nevertheless, the mechanics of our construction is quite similar in spirit to the original block spin approach of Kadanoff~\cite{kadanoff1966scaling}, which in some sense we simply embed into a larger operational framework that allows unitarity to be preserved and enables us to identify phases that lack local order parameters. 

	Since we are interested in making rigorous statements, we have restricted our attention to 1D gapped phases of matter with Abelian internal symmetries, which can be analyzed using powerful group-theoretic and tensor network methods. In the main text we have focused on SPT phases, but as shown in App. \ref{app:symm_breaking} a similar treatment is also possible for the (simpler) case of symmetry-breaking orders. Our belief is however that the general connection between RG and EC extends well beyond these examples, bringing with it a diverse array of implications across multiple fields of physics.  We describe some of these implications and discuss future directions below.
	
	{\it Quantum matter:} From the quantum matter perspective, the most important implication of our work is the existence of the 
multiscale string operators (MSOs): operators that recognize all wavefunctions belonging to the same phase (see Sec.~\ref{sec:msos}).
    Another consequence of our work stems from our circuit's ability to function as a {\it universal variational ansatz} for ground state wavefunctions. The basic idea is to turn EC on its head, by running our RG circuit $\mcq$ `in reverse'. We do this by starting from a variational family of errors---parametrized by a wavefunction $\k{\phi_{an}}$---which we then encode into a variational wavefunction by acting with $\mcq^\da$ on $\k{\Psi} \tp \k{\phi_{an}}$. This performs a {\it reversed} RG flow from the IR to the UV, and generates a state in the target SPT phase whose pattern of local entanglement is controlled by the choice of $\k{\phi_{an}}$. 
 Crucially, our results imply that choices of $\ket{\phi_{an}}$ are in one-to-one correspondence with ground state wavefunctions in the target phase (up to a set of measure zero). 
	The hierarchical construction of this encoding map shares many parallels with that of the Multiscale Entanglement Renormalizaton Ansatz (MERA) \cite{vidal2008class}, but is distinguished by having features native to the target phase `hard-coded' into the ansatz. The EC machinery has the added benefit of explicitly separating out the universal degrees of freedom $(\k{\Psi})$ from the non-universal ones ($\k{\phi_{an}}$). 

 	{\it Quantum simulation:}
	From a practical point of view, our circuits provide a way of performing phase recognition even very close to critical points, where signals from conventional order parameters are often smeared out due to strong quantum fluctuations (or long correlation lengths). This problem is particularly acute in near-term quantum simulators, which typically have rather modest system sizes \cite{herrmann2021realizing}. 
	Furthermore, our protocol for performing RG involves applying a circuit whose depth is at most {\it logarithmic} in the system size, making it suitable for implementation on near-term quantum devices.
 
    As discussed in Sec.~\ref{sec:msos}, our phase recognition can be effected in practice in two different ways: a) first acting on the input state $\k\psi$ with a $d$-layer RG circuit $\mcq_\o^{(d)}$ and then measuring the expectation value of the elementary string operator $\mcs_{g}^{(\o)}$, or b) directly measuring the expectation value $\lan \psi | \wt \mcs^{(\o,d)}_{g} | \psi\ran $ of the MSO in the input state. 
    The advantage of b) is that no unitary gates are needed: the form of the MSOs implies that their expectation values can be reconstructed from simultaneous measurements of all of the exponentially many (in $d$) mutually-commuting elementary string operators (each supported only on two sites).
    Together with classical post-processing, this allows estimating the MSOs at arbitrary depth $d$ without applying $\mcq_\o^{(d)}$.
    This feature enables the immediate implementation of our approach in experiments, as already done in two different quantum simulation platforms~\cite{herrmann2021realizing,cong2022enhancing}, prior to our rigorous theoretical understanding.
    On the other hand, an advantage of approach a) is its robustness against experimental imperfections 
    that could arise in future digital fault-tolerant quantum computers. 
    This is due to the exponential decrease with $d$ in the size of operators that need to be measured, meaning that approach a) is exponentially less sensitive to final readout errors than approach b). In essence, our circuit $\mcq$ concentrates important and universal information about the input state onto a small number of output qubits.
	
	{\it Complexity of QPR:} 
	A key performance metric for an algorithm solving the QPR problem is the {\it sample complexity} \cite{haah2017sample}, namely the number of copies of a given state $\k\psi$ needed to confidently decide if $\k{\psi}$ belongs to the target phase. 
	Our multiscale string operators (MSOs) take on a value asymptotically close to unity {\it everywhere} in the target phase (and vanish outside of it), implying that measurements of the string operators performed after acting with our RG circuits lead to nearly deterministic outcomes, and correspondingly to low sample complexities for a fixed system size.
    This low sample complexity is particularly desirable from the perspective of near-term quantum simulators, where finite-size effects are often quite large.\footnote{Note that on a system of size $L$, the number of well-separated string operators with size $\gg \xi$ diverges as $L \ra \infty$, giving a trivial sample complexity in the thermodynamic limit. Thus this reduction in sample complexity is only well-defined for finite systems.}
It would be interesting to see if one could provide a fundamental lower bound on the sample complexity needed in this context, and if so, to check whether or not our algorithm saturates it.

	{\it Machine learning:} 
  There has recently been a concerted effort to tackle the phase recognition problem using machine learning (ML) methods \cite{aaronson2007learnability,huang2021provably,carleo2017solving,wang2016discovering,carrasquilla2017machine,carleo2019machine}. 
	In this context, our MSOs 
 suggest that our circuits may serve as a good reference point from which to benchmark such heuristic approaches.
 At the same time, our circuit provides a good starting point for numerical optimization schemes in QCNN-based ML approaches~\cite{cong2019quantum}.
Furthermore, the existence of MSOs has implications on the rigorous performance guarantee of classical shadow-based ML approaches to QPR problems.
In previous work~\cite{huang2021provably}, the performance guarantee for recognizing SPT phases was limited to 
the case of a 1D $\zt^2$ SPT phase with enhanced $O(2)$ symmetry. Our work extends these results to a broader class of phases: all 1D SPT phases with Abelian 
 internal symmetries. 
Relatedly, we also note that while we have focused on performing phase recognition by measuring expectation values of observables, our RG circuits may also allow for the efficient extraction of other universal (potentially non-linear) properties of the quantum phase in question.
	
	{\it Intrinsic topological phases:} 
	A clear next step is to rigorously characterize circuit architectures that recognize phases with intrinsic topological order (perhaps in a way which builds on string-net tensor network methods \cite{levinwen,levin2007tensor,gu2009tensor,evenbly2015tensor} or RG-inspired decoders \cite{bravyi2011analytic}), with first steps recently being taken in \cite{cong2022enhancing,cian2022extracting}. 
	Given the recent realization of topologically ordered states in quantum simulators \cite{satzinger2021realizing, semeghini2021probing} and measurement-based protocols for preparing topological orders \cite{verresen2021efficiently, tantivasadakarn2022hierarchy}, addressing this question is of immediate practical importance. Since the local unitary circuits (in the language of EC, the `errors') which relate different wavefunctions in a topological phase can be completely generic --- as opposed to the {\it symmetric} local circuits relevant for SPT phases --- we expect that the EC procedure in this case must be one of {\it quantum} error correction. Whether or not the EC threshold in these phases can be made to coincide with the phase boundary is an outstanding question for the future. 

	More broadly, connections between EC and RG may have implications for the design of useful RG schemes in contexts very different from those considered in this paper. While the particular RG scheme we have employed incorporates spatial coarse-graining, the connection between RG and EC may be deeper than this, and approaching RG from the perspective of EC may give us fresh insights into systems for which there is no useful real-space RG scheme, such as (non)-Fermi liquids \cite{shankar1994renormalization} and fracton-inspired models with UV/IR mixing \cite{you2021fractonic,gorantla2021low,lake2022renormalization}. 
	
	In a different context, it is natural to ask whether the integration of EC into the MERA framework provides any new lessons for tensor network approaches to holography \cite{swingle2012entanglement,jahn2021holographic,hayden2016holographic}, or for the role that EC plays in quantum gravity more generally \cite{almheiri2015bulk,pastawski2015holographic,kim2017entanglement,yang2016bidirectional}. This direction seems especially promising given that the layout of our RG quantum circuits resembles the geometry of holographic tensor networks and that our elementary MPS tensor $A$ is perfect, potentially allowing for more general constructions of holographic error correcting codes~\cite{pastawski2015holographic}.
	
	\section*{Acknowledgements}
	
	The authors thank Iris Cong, Jeongwan Haah, Iman Marvian, Hannes Pichler, David Stephen, Ruben Verresen, and Ashvin Vishwanath for helpful discussions, and are particularly grateful to Iris Cong, Sarang Gopalakrishnan, and Nishad Maskara for detailed feedback and a careful reading of a draft version. E.L. was supported by the Fannie and John Hertz Foundation and by NSF grant DMR-2206305. S.B. was supported by the National Science Foundation Graduate Research Fellowship under Grant No. 1745302. 
	
	\begin{widetext}

		\bibliography{circuit_rg}
				
		\bs 
		\section*{\underline{Guide to appendices}}
		
		For the reader's conveniece, we list below a brief description of the technical appendicies to follow: 
		\begin{itemize}
			\item Appendix \ref{app:math}: Group theory details, introducing notation. 
			\item Appendix \ref{app:zn_details}: Representation theory details for finite Abelian groups $G$. 
			\item Appendix \ref{app:mps_technology}: MPS techniuqes and construction of the canonical reference states. 
			\item Appendix \ref{app:general_ground_states}: General representation of SPT ground states in terms of `errors' acting on the canonical reference states. 
			\item Appendix \ref{app:qcnn_details}: The details of our circuit architecture. 
			\item Appendix \ref{app:convergence}: Proof that the RG implemented by our RG circuit converges. 
			\item Appendix \ref{app:symm_breaking}: How to construct RG circuits for spontaneous symmetry breaking. 
			\item Appendix \ref{app:measurements}: Error correction with measurements and classical post-processing. 
			\item Appendix \ref{app:blocking}: How to deal with generic $G$-representations.
			\item Appendix \ref{app:sample_complexity}: Detailed analysis of sample complexity. 
			\item Appendix \ref{app:non_mnc}: Modified quantum circuits for cohomology classes which are not maximally non-commutative.
		\end{itemize}
		
		\bs 
		
		\hrule 
		
		\appendix 
		
		\section{Mathematical preliminaries } \label{app:math}
		
		In this appendix, we provide some of the background group theory details used in this work. 
		
		Let $G$ be a finite Abelian group. We will let $R_g$ denote the local symmetry action of $g\in G$, which until App. \ref{app:blocking} we take to be in the regular representation, acting on a Hilbert space of dimension $|G|$. Explicitly, 
		\be R_g = \sum_{h\in G} \c_h(g) \proj h.  \ee 
		
		Let $V^{(\o)}_g$ be a projective representation of $G$ with factor set $\o : G \times G \ra U(1)$, so that 
		\be V^{(\o)}_g V^{(\o)}_h = \o_{g,h} V^{(\o)}_{g+h}.\ee 
		For ease of notation, we will omit the dependence of $V_g$ on $\o$ when there is no cause for confusion. 
		
		Since the $V_g$ are only defined up to $g$-dependent phases, $\o_{g,h}$ is equivalent to $\o_{g,h} \a_g \a_h / \a_{g+h}$ for some set of phases $\a_g$. For all of the cases we will be interested in, $\o$ can always be chosen to satisfy 
		\be \o_{g,h} = \o_{-g,h}^* = \o_{g,-h}^*.\ee 
		We define the function
		\be  \l_\o(g,h) \equiv \frac{\o_{g,h}}{\o_{h,g}},\ee 
		which determines the group commutators of the $V_g$ via 
		\be V_gV_h = \l_\o(g,h) V_hV_g.\ee 
		Note that $\l_\o$ is invariant under transforming $\o_{g,h}$ by the phases $\a_g$, and due to the cocycle condition on $\o$, it is a homomorphism in both of its arguments:
		\be \l_\o(g+k,h) = \l_\o(g,h) \l_\o(k,h),\ee 
		and similarly for $\l_\o(g,h+k)$. 
		$\l_\o$ furthermore obeys the identities 
		\be \l_\o(g,e) = \l_\o(g,-g) = 1 ,\qq \l_\o(g,h) = \l_\o^*(h,g) = \l_\o^*(g,-h) = \l_\o(-g,-h),\ee
		from which one sees that Hermitian conjugates commute in the same way: $V_g^\da V_h^\da = \l_\o(g,h) V_h^\da V_g^\da$.

		We will also make use of the function 
		\be \vs_g \equiv \o(g,g) \ee 
		which can be used to convert between $V_g$ and $V_{-g}^\da$ via 
		\be \vs_g V_{-g} = V_{g}^\da.\ee 
		
		\begin{definition}
			The projective center \cite{de2021symmetry} $C_\o$ is defined to contain all those elements which are projectively represented by scalars: 
			\be C_\o \equiv \{ g\in G \, : \, \l_\o(g,h) = 1 \, \forall\, h \in G\}.\ee 
		\end{definition}

		We furthermore define a {\it maximally non-commutative} factor set \cite{else2013hidden} as one in which no nontrivial elements are projectively represented by scalars: 
		
		\begin{definition}
			A maximally non-commutative factor set is one with trivial projective center, viz. one for which $C_\o = \{ e\}$. 
		\end{definition}

		For a given factor set, we may construct a homomorphism $\L$ as 
		\bea \L_\o & : G \ra G^* \\ 
		g &\mt \l_\o(\cdot,g) \eea 
		where $G^*$ denotes the group of characters.  Then by definition, 
		\be C_\o = \ker(\L_\o).\ee 
		Since $\L_\o(g) : G \ra U(1)$, it can always be identified with some nontrivial character of $G$. 
		We will write the map which associates $g$ to the appropriate character as $\G_\o$, with $\G_\o(g)$ thus satisfying 
		\be \label{gammadef} \c_h(\G_\o(g)) = \l_\o(h,g) \, \, \forall \,\, h \in G.\ee 
		The element $\G_\o(g)$ can be identified explicitly using character orthogonality as 
		\be \k{\G_\o(g)} = \int_k \c_h^*(k) \l_\o(h,g) \k k, \ee 
		where $\int_k \equiv  \frac1{|G|} \sum_{k\in G}$.

		If $\o$ is a maximally non-commutative factor set, then $\L_\o$ and $\G_\o$ are injective. Thus since $|G^*| = |G|$, $\L_\o$ and $\G_\o$ are isomorphisms iff $\o$ is maximally non-commutative. This means that the knowledge of $\l_\o(g,h)$ for all $h$ is sufficient to uniquely determine the group element $g$ iff $\o$ is maximally non-commutative. This property greatly simplifies the error correcting technology used in the construction of our RG circuits (and also simplifies the analysis of SPT phases as resource states for measurement-based quantum computation \cite{else2012symmetry,raussendorf2017symmetry}), and is the reason why we choose to restrict our attention to maximally non-commutative cohomology classes throughout most of this work.

		We now state some useful math facts that we will not go through the trouble of proving; proofs can be extracted from \cite{berkovich1998characters}, \cite{karpilovsky1994group}:
		
		\begin{prop}
			The following facts hold: 
			\begin{itemize} 
				\item Maximally non-Abelian factor sets exist iff $G = G'\times G'$ is a square. 
				\item All projective representations with factor set $\o$ have dimension $ \sqrt{|G|/|C_\o|}$.
				\item Distinct projective irreps with factor set $\o$ are enumerated by the linear characters of the group $C_\o$; there are thus $|C_\o|$ distinct projective irreps. In particular, if $\o$ is maximally non-commutative, there is only a single projective irrep up to isomorphism.
			\end{itemize} 
		\end{prop}
		
		One result that we could not find in the literature is the grand orthogonality theorem for projective representations:
		
		\begin{theorem}[Schur orthogonality for projective representations]
			For $V^\mu$ and $V^\nu$ any two projective irreps of $G$, 
			\be \label{shur_ortho} \int_g [(V^\mu_g)^\da]_{ij}[V^\nu_g]_{kl} =\frac1{\dim(V^\mu)} \d_{il}\d_{jk} \d_{\mu\nu}.\ee 
		\end{theorem}
		\begin{proof}
			The proof is straightforward and mimicks the linear case: we construct an intertwiner with the map 
			\be 
			F_{ij}^{\mu\nu} = \int_g (V^\mu_g)^\da E_{ij} V_g^\nu,
			\ee 
			
			which satisfies 
			
			\bea
			V^\mu_h F_{ij}^{\mu\nu} &= \int_g \vs_g \o(h,-g) V^\mu_{-g+h} E_{ij} V_g^\nu  \\
			&= \int_g \vs_{g}\vs_{-g+h}^*\o_{h,-g} (V^\mu_{g-h})^\da E_{ij} V_{g-h+h}^\nu \\ &= \int_g \vs_{g}\vs_{-g+h}^* \o_{h,-g} \o_{g-h,h}^* ((V^\mu_{g-h})^\da E_{ij} V_{g-h}^\nu) V_h^\nu.
			\eea 
			
			Crucially, the factor 
			\be \vs_{g}\vs_{-g+h}^* \o_{h,-g} \o_{g-h,h}^* =1\ee 
			due to the cocycle condition $\d\o=1$, which expanded out is simply
			\be
			\omega(g_1, g_2) \omega(g_1+g_2, g_3) = \omega(g_2,g_3) \omega(g_1, g_2+g_3).
			\ee
			Therefore $[F^{\mu\nu}_{ij}]_{kl} = \d_{\mu,\nu} \d_{k,l} \l_{i,j}$ for some constant factor $\l_{i,j}$. This can be identified by computing 
			\be \l_{i,j} = \frac{\Tr[F_{ij}^{\mu\mu}]}{\dim(V^\mu)} = \frac{\d_{i,j}}{\dim(V^\mu)},\ee 
			thereby proving the claim. 
			
		\end{proof}
		
		\begin{corollary}
			If $V_g$ is a projective irrep and $\o$ is maximally non-commutative, then 
			\be {\rm Tr}[V_g] = \d_{g,e} \sqrt{|G|}.\ee 
			In particular, the only nonzero character of $V_g$ is the trivial character. 
		\end{corollary}
		\begin{proof}
			We use Shur orthogonality to write 
			\bea \int_g |\Tr[V_g]|^2 & = \sum_{i,j} \int_g [V^\da_g]_{ii} [V_g]_{jj} = 1.\eea 
			Since $\int_g = \frac1{|G|} \sum_{g\in G}$ we must have 
			\be \sum_{g\in G} |\Tr[V_g]|^2 = |G|,\ee 
			analogously to the case of linear representations.  Now if $\o$ is maximally non-commutative, then $\dim(V) = \sqrt{|G|}$. Since $V_e = \unit_{\sqrt{|G|}}$, the above sum is then fully saturated by $g=e$, proving the claim. 		
		\end{proof}
		
		\section{Representation theory data for finite Abelian $G$ }\label{app:zn_details}
		
		In this appendix we provide explicit expressions for some of the representation theory data employed in the main text. We first focus on the case of $G = \zn \times \zn$, and later show how this can be generalized. 
		
		\ss{$G=\zn \times \zn$}
		
		The second cohomology classes of $G$ form the group $H^2(\zn\times \zn;U(1)) = \zn$. 
		Writing group elements $g\in \zn\times \zn$ as $\bm{g}=(g_1,g_2)$, we may always choose a gauge in which the factor sets are given by \cite{berkovich1998characters}
		\be \o_{g,h} = \z^{\o g_2h_1},\ee 
		where $\z \equiv e^{2\pi i /N}$ and where we have abused notation slightly by letting $\o$ in the exponent denote the element of $\zn$ corresponding to the cohomology class $\o$. 
		The commutator function $\l_\o$ is correspondingly 
		\be \l_\o(g,h) = \z^{-\o \bfg \times \bfh} = \z^{-\o (g_1h_2-g_2h_1)}.\ee
		With the standard choice of character $\c_g(h) = \z^{\bfg\cdot\bfh}$, we then have 
		\be \label{zngamma} \G_\o((g_1,g_2)) = \o \cdot (-g_2,g_1),\qq \vs((g_1,g_2)) = \z^{\o g_1g_2}.\ee 
		Explicit linear and projective representations consistent with the above choice of factor set are afforded by 
		\be R_g = Z^{g_1} \tp Z^{g_2},\qq V_g^\o = X^{g_1} Z^{\o g_2},\ee 
		where 
		\be X \equiv \sum_a \kb{a+1}a,\qq Z\equiv \sum_a \z^a \proj a\ee 
		are the usual $\zn$ clock and shift matrices, satisfying $ZX = \z XZ$. 
		
		By attempting to invert $\G_\o$, we see that $\o$ is maximally non-commutative iff $\o$ is invertible in $\zn$, i.e. iff $\o$ and $N$ are relatively prime \cite{else2013hidden}, in which case $\G_\o\inv((g_1,g_2)) = \frac1\o(g_2,-g_1)$ (note that $\o$ is always maximally non-commutative if $N$ is prime). By the results of App. \ref{app:math}, all possible projective representations are then gauge-equivalent to the $V^\o_g$ written above.
		
		\ss{General Abelian groups} 
		
		Let $G = \prod_{i=1}^M \zz_{N_i}$,
		and write elements in $G$ as $M$-component vectors $(g_1,\dots,g_M)$.  	
		The calculation of $H^2(G;U(1))$ can be accomplished by repeated use of the Kunneth formula \cite{brown2012cohomology}
		\be 0 \ra \bigoplus_{k=0}^n H^k(G;\zz) \tp_\zz H^{n-k}(K;\zz) \ra H^n(G\times K;\zz) \ra \bigoplus_{k=0}^{n+1} {\rm Tor}[H^{n+1-k}(G;\zz), H^k(K;\zz)]\ra0,\ee 
		together with use of $H^k(G) \equiv H^k(G;U(1)) = H^{k+1}(G;\zz)$, $\zn \tp_\zz \zz_M = \zz_{(N,M)}$ (with $(N,M)\equiv \gcd(N,M)$), and the fact that 
		\be H^k(\zn;\zz) = \begin{dcases}
			\zn \qq & k \in 2\nn \\ 0\qq & k \in 2\nn+1 \\ \zz & k = 0 
		\end{dcases}.\ee 
		Putting these together, we obtain 
		\be H^2(G;U(1)) = \prod_{j<k}^M \zz_{(N_j,N_k)}.\ee 
		As a straightforward generalization of the $\zn^2$ case, representatives for the distinct factor sets are provided by 
		\be \o_{g,h} = \prod_{1\leq i < j \leq M} \z_{(N_i,N_j)}^{\hat\o_{ij} g_ih_j},\ee 
		where $\hat \o$ is now a symmetric matrix with entries $\hat\o_{ij} \in \zz_{(N_i,N_j)}$. The commutator function is consequently 
		\be \l_\o(g,h) = \prod_{1\leq i < j \leq M} \z_{(N_i,N_j)}^{\hat\o_{ij}(g_ih_j - h_i g_j)},\ee 
		from which $\o$ is seen to be maximally non-commutative if $\hat \o_{ij}$ is relatively prime to $(N_i,N_j)$ for all $i,j$. 
		Explicit forms for the projective representations with maximally non-commutative $\o$ can be constructed by taking appropriate products of the expressions in the $\zn\times \zn$ case.

		\section{MPS techniques and the construction of $\k{\Psi_\o}$ } \label{app:mps_technology} 
		
		In this section we describe how to construct the canonical reference states $\k{\Psi_\o}$ for each SPT phase, and prove some relations involving the shift operators $S^{L/R}_g$ introduced in the main text. 
		
		\ss{MPS tensors} 
		
		In all of what follows, we take the physical indices of our MPS tensors to be indexed by elements of $G$, with the physical onsite Hilbert space $\mch$ having dimension $|G|$.  
		We will always choose a basis in which the onsite symmetry action $R_g$ is diagonal, with 
		\be \label{correctiblerh} R_g = \sum_{h\in G} \chi_g(h) \proj h,\ee  
		where $\chi_g(\cdot)$ is the linear character on $G$ dual to the element $g$. 
		
		Since the physical indices of our MPS states are indexed by $G$, an arbitrary MPS tensor $A$ may be written 
		\be A = \sum_{g\in G} \k{g} A^g,\ee 
		where $A^g$ is a $\chi\times \chi$ matrix in ${\rm End}(\mch_{\text{virt}})$, with $\mch_{\text{virt}}$ the $\c$-dimensional virtual space of the MPS. All MPS tensors descrbing MPS wavefunctions in $\spt_\o$ must be such that the symmetry pushes from a linear action on the physical leg to a projective action on the virtual legs \cite{cirac2021matrix}, with 
		\be \label{explicit_pushthrough} \sum_{h\in G} R_g \k{h} A^h = \sum_{h\in G} \k h (V^\da_g A^h V_g)\ee 
		for some $\c$-dimensional projective representation $V_g$ of $G$ with factor set $\o$ (which neccessitates that $\c\geq \sqrt{|G|}$ if $\o$ is nontrivial). We will find it helpful to use notation where matrices acting on the left of $A$ denote an action on $A$'s physical index, while matrices of the form $(X\tp Y)$ which act on the right of $A$ denote an action on $A$'s left virtual index by $X$ and right virtual index by $Y$. In this notation Eq.~\eqref{explicit_pushthrough} reads 
		\be R_g A = A (V_g^\da \tp V_g).\ee 
		Graphically, 
		\be \label{pushthrough} \igptfc{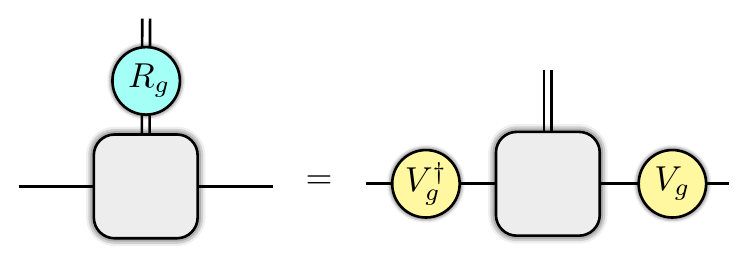}\ee 
		where the grey box denotes an MPS tensor whose physical leg is indicated by a doubled line (since in the present MNC case, $\dim \mch = (\dim \mch_{virt})^2$). In our notation MPS tensors with downward-pointing physical legs will tacitly assumed to be complex conjugated; thus $R_g$ pushes `up' through MPS tensors as  
		\be \igptfc{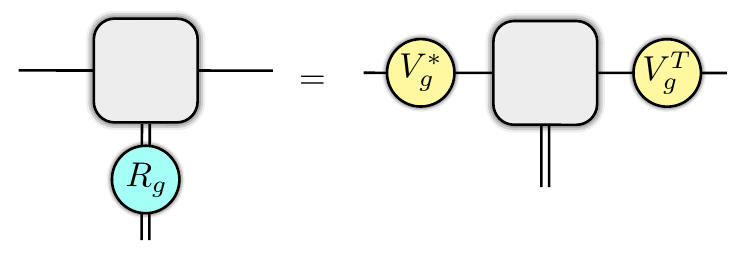}\ee

		Since there is no spontaneous symmetry breaking (SSB) in the models we are considering, all MPS tensors will be required to be injective \cite{cirac2021matrix}, meaning that some integer number $n$ of $A$ tensors can be blocked together so that the MPS tensor $\sum_{\bfg \in G^n} \k{\bfg} (A^{g_1} \cdots A^{g_n})$ defines an isometry on $\mch_{\text{virt}}^{\tp 2}$. 
		
		We now discuss the construction of the canonical reference states $\k{\Psi_\o}$. 
		The MPS tensors that appear in our definition of $\k{\Psi_\o}$ will in fact satisfy a much stronger property than injectivity: 
		\begin{definition}[perfect MPS]
			An MPS tensor $A$ is said to be perfect if it is a perfect tensor when viewed as a map $\mch_{virt}^{\tp 2} \ra \mch$. That is, $A$ is perfect if both 
			\be \label{virtis} {\rm Tr}[(A^{g})^\da A^h] \propto \d_{g,h} \ee 
			(isometry when virtual legs are contracted), and 
			\be\label{inj} \sum_{g\in G} [A^g]_{ij} [A^g]^*_{kl} \propto \d_{ik}\d_{jl} \ee 
			(isometry when physical legs are contracted). 
		\end{definition}
		Graphically, Eq.~\eqref{virtis} means (from now on all grey boxes drawn will indicate canonical MPS tensors $A_\o$, unless indicated otherwise)
		\be \includegraphics[width=.2\tw]{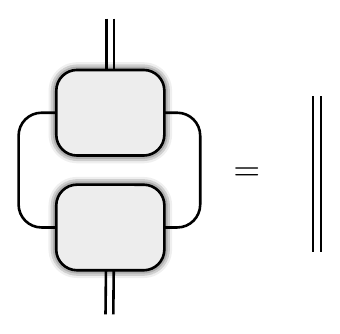} \ee 
		while Eq.~\eqref{inj} means 
		\be \label{graphical_inj} \includegraphics[width=.25\tw]{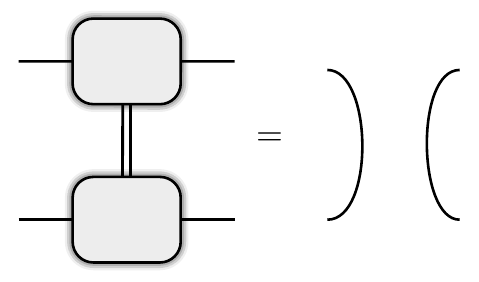} \ee 
		We will refer to a translation-invariant MPS wavefunction constructed from perfect MPS tensors as a `perfect MPS'. Note that by Eq.~\eqref{inj}, a perfect MPS tensor is automatically injective without the need for any blocking. 
		
		It is very easy to write down an explicit expression for a perfect MPS tensor in the maximally non-commutative case:	
		\begin{prop}[maximally non-commutative MPS tensors]
			For maximally non-commutative $\o$, one may always construct a perfect MPS with bond dimension $\sqrt{|G|}$ and MPS tensors 
			\be\label{canonical_mps_def} A_\o = \sum_{g\in G} \k{g}  V_{\G_\o\inv(g)}^\da = \sum_{g\in G} \vs_g \k g V_{-\G_\o\inv(g)}.\ee
		\end{prop}
		Note that $A_\o$ has the smallest allowable bond dimension, and thus defines a state in $\spt_\o$ with the smallest possible entanglement. The fact that the bond dimension of $A_\o$ is the square root of the physical dimension of $A_\o$ is the reason for indicating physical legs with doubled lines in our graphical notation. 
		
		That the above choice of $A_\o$ works is proven by noticing that the phases generated when commuting $V_g$s past one another can always be transformed into the linear character $\c_g(h)$ appearing in the action of $R_g$:
		\begin{proof}
			Injectivity of $A_\o$ can be checked using Schur orthogonality \eqref{shur_ortho}: 
			\be\label{correctible_inject} \sum_{g\in G} (A^g_\o)^\da \tp A_\o^g = |G|\int_g  \sum_{ijkl} \kb{ij}{kl} [V_g]_{ij} [V_g]^*_{lk} = \sqrt{|G|}\sum_{ij} \kb{ij}{ij} = \sqrt{|G|}\unit_{\text{virt}}^{\tp 2}.\ee  
			That $A_\o$ is perfect then follows from the fact that 
			\be \Tr[(A_\o^g)^\da A_\o^h] \kb{g}{h} =  \kb{g}{h} \Tr[V^\da_{\G\inv_\o(g)}V_{\G\inv_\o(h)}] = \sqrt{|G|}\d_{g,h}\ee 
			
			We then verify the pull-through property \eqref{explicit_pushthrough}. Conjugating the $V_{\G\inv_\o(g)}$ by $V_h$s at the cost of a $\l_\o$ term, we have 
			\bea R_h A_\o = \sum_{g\in G} \k g \frac{\c_h(g)}{\l_\o(h,\G\inv_\o(g))} V_h^\da V_{\G_\o\inv(g)}^\da V_h . \eea
			But by definition of the map $\G_\o$ \eqref{gammadef}, we have 
			\be \l_\o(h,\G_\o\inv(g)) = \c_h(g),\ee 
			and thus we have $R_h A_\o = A_\o(V_h^\da\tp V_h)$ as desired. 
		\end{proof}
		
		We define the canonical representative of the phase $\spt_\o$ using these MPS tensors: 
		
		\begin{definition}[canonical reference state of an SPT phase]
			For a given maximally non-commutative cohomology class $\o$, the canonical reference state of the phase $\spt_\o$ is defined as
			\be \label{canonical_def} \k{\Psi_\o} \equiv \sum_{\bfg\in G^L} {\rm Tr}[A^{g_1}_\o \cdots A_\o^{g_L}]\k{\bfg},\ee 
			where $A_\o$ are the MPS tensors of \eqref{canonical_mps_def}. 
		\end{definition}
		
		A useful fact we will employ later, which follows easily from the injectivity of $A_\o$ and $\Tr[V_g] \propto \d_{g,e}$, is 
		\be \label{rgoverlap} \lan \Psi_\o | R_\bfg | \Psi_\o\ran = \prod_i \d_{g_i,e}.\ee

		\ss{Shift operators and parent Hamiltonians} 
		
		We now define the right and left shift operators employed in the main text, and show how they can be used to construct a Hamiltonian whose ground state is $\k{\Psi_\o}$. 
		\begin{definition}[right shift operator]
			We define the right shift operator as 
			\be S^R_g = \sum_{h\in G}  \o(g,\G\inv_\o(h)+g) \vs_{-g}^* \kb{h+\G_\o(g)}{h} = \sum_{h\in G} \o(g,\G\inv_\o(h)) \vs_{-g}^* \kb{h}{h-\G_\o(g)}.\ee
			
		\end{definition}
		This operator is designed to satisfy 
		\bea \label{sright} A_\o(\unit \tp V_g) & = \sum_{h\in G} \k h V_{\G\inv_\o(h)}^\da V_{-g}^\da \vs^*_{-g} \\ 
		& = \sum_{h\in G} \k{h+\G_\o(g)} V_{\G\inv_\o(h)}^\da\o(-g,\G\inv_\o(h)+g)^* \vs^*_{-g} \\ 
		& = S^R_g A_\o.\eea
		Graphically, 
		\be \igptfc{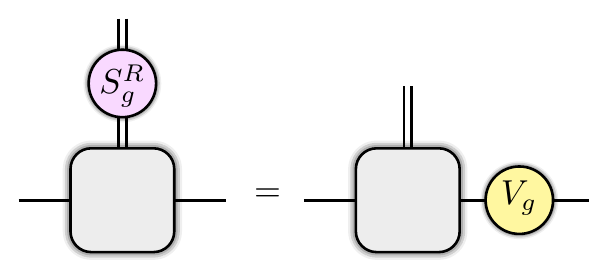}.\ee 
		Given that $S_R$ is unitary, $(S^R_g)^\da$ can be used to correct for an `error' of $V_g$ occurring on the right virtual leg of $A_\o$, via 
		\be (S_g^R)^\da A_\o (1\tp V_g) = A_\o.\ee 
		Note that the $S^R_g$ obey the same fusion rule as the projective representations themselves, 
		\be S^R_g S^R_h = \o_{g,h}S^R_{g+h}, \ee 
		so that 
		\be \label{srcomm} S^R_g S^R_h = \l_\o(g,h) S^R_h S^R_g .\ee 
		Furthermore, $S^R_g$ picks up the same phase as $V_g$ under Hermitian conjugation:
		\be (S^R_g)^\da = \vs_g S^R_{-g},\ee
		and it commutes with the linear symmetry operators as 
		\be \label{srwithr} S^R_g R_h = \l_\o(g,h) R_h S^R_g.\ee 
		These equations can all be proved by evaluating the action of both sides on $A_\o$; that this is sufficient follows from the fact that $A_\o$ is perfect.

		Similarly,
		\begin{definition}[left shift operator] We define the left shift operator as 
			\bea S^L_g =  \vs_gR_g S^R_{-g} = \sum_{h\in G} \o(g-\G\inv_\o(h),g) \kb{h-\G_\o(g)}{h} = \sum_{h\in G} \o(\G\inv_\o(h),g)^* \kb{h}{h+\G_\o(g)}.\eea 
		\end{definition}
		
		$S^L_g$ is designed to satisfy 
		\bea \label{sleft} A_\o(V_g^\da \tp \unit) & = R_g A_\o (\unit \tp V_g^\da) \\ 
		& = \vs_g R_g A_\o(\unit \tp V_{-g}) \\ 
		& = \vs_gR_g S^R_{-g} A_\o \\ 
		&= S^L_gA_\o.\eea
		Graphically, 
		\be \igptfc{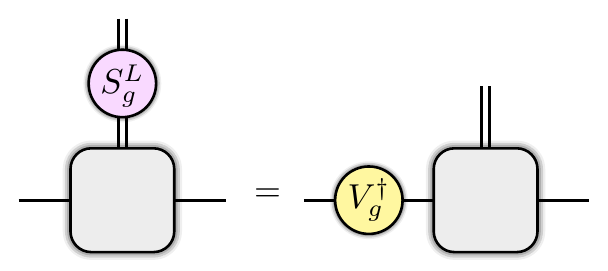}. \ee  
		Therefore an `error' of $V_g^\da$ occurring on the left virtual leg of $A_\o$ can be corrected by way of 
		\be (S^L_g)^\da  A_\o(V_g^\da \tp \unit) = A_\o.\ee 
		The $S^L$ operators combine in the opposite way as the projective representations, since they put the projective representations in `reversed' order on the left legs: 
		\be S^L_g S^L_h = \o_{h,g}^* S^L_{g+h}.\ee 
		Therefore (c.f. \eqref{srcomm})
		\be S^L_g S^L_h =\l_\o(h,g) S^L_h S^L_g,\ee 
		and 
		\be \label{slwithr} S^L_g R_h = \l_{h,g} R_h S^L_g.\ee

		Note that the $S^{L/R}_g$ are designed to `fractionalize' the linear symmetry action, in that 
		\be \label{fracd_rg} R_g = S^L_g S^R_g.\ee  
		We also see that $\G_\o(g)$ has the meaning of the linear representation of $G$ that the shift operators transform under the adjoint action of $G$: indeed, by \eqref{srwithr} and \eqref{slwithr}, 
		\be \label{shift_irreps} R_h^\da S^L_g R_h =  \c_{\G_\o(g)}(h) S^L_g ,\qq R_h^\da S^R_g R_h = \c^*_{\G_\o(g)}(h) S^R_g.\ee 
		Thus $S^L_g, S^R_g$ transform in conjugate representations, as befitting the fact that their product (viz. $R_g$) transforms trivially under the adjoint action, by virtue of $G$ being Abelian. 
		
		The left and right shift operators can be used to construct string operators, as well as a Hamiltonian whose ground state is $\k{\Psi_\o}$: 
		
		\begin{prop}[Parent Hamiltonian] 
			A commuting projector Hamiltonian with ground state $\k{\Psi_\o}$ is 
			\be H = -\sum_i \sum_{g\in G} S_{g,i-1}^R S_{g,i}^L.\ee
		\end{prop}
		
		\begin{proof} 
			The standard way of constructing $H$ is to write \cite{cirac2021matrix}
			\be H = -\sum_i \mcp_{{\rm span}(AA)},\ee 
			where $\mcp_{{\rm span}(AA)}$ is the projector onto the subspace $\mch_{AA}$ spanned by the tensor product of two contracted MPS tensors: 
			\be \mch_{AA} =  {\rm span}\{\k{g,h} \Tr[A_\o^g A_\o^h X]\, | \, X \in {\rm Mat}_{\sqrt{|G|}\times \sqrt{|G|}},\,g,h\in G\}.\ee 
			Since our MPS tensors are injective, $\dim \mch_{AA} = \dim \mch = |G|$. We therefore need to find $|G|$ independent operators $\mcp_{i,i+1}^g$ which stabilize the ground state. Constructing these operators is actually quite straightforward to do if we make use of the shift operators defined previously; a little reflection reveals that the expression above works, as it is easily checked to stabilize the MPS ground state by way of \eqref{sright} and \eqref{sleft}.
		\end{proof} 
		
		The following string operators, which are formed by taking a symmetry action on a finite interval and decorating the ends with shift operators, play a fundamental role in our ability to distinguish different SPT phases: 
		\begin{definition}[string operators]
			For any two sites $j>i$ and any $g\in G$, we define the string operator $\mcs_{g,i\ra j}$ as 
			\be \label{string_def_app} \mcs_{g;i\ra j} \equiv  S^R_{g,i} \tp \bot_{l=i+1}^{j-1} R_{g,l} \tp S^L_{g,j}.\ee
		\end{definition}
		Using the pull-through condition \eqref{pushthrough}, these string operators are easily checked to stabilize $\k{\Psi_\o}$. Note also that the operators appearing in the above parent Hamiltonian are simply $\mcs_{g,i\ra i+1}$.
		
		Intuitively, the string operators are constructed by acting with the symmetry over a finite interval, and then decorating the ends of the interval with a fractionalized action of the symmetry (this terminology being appropriate since $R_g$ fractionalizes into $S^L_g S^R_g$). If one were to decorate the ends of the interval with operators transforming under $G$ according to any irrep {\it other} than $\c_{\G_\o(g)}$ (see \eqref{shift_irreps}), it is easy to see that the resulting string operator will have vanishing expectation value throughout the SPT phase \cite{pollmann2012detection,de2021symmetry}.

		\ss{Explicit expressions for $G = \zn\times \zn$} \label{app:explicit_zn_tech}

		We now give expressions for the various tensors defined above in case considered in the main text, where $G = \zn\times \zn$ and $\o$ is maximally non-commutative. We will often decompose the physical Hilbert space into two factors of dimension $N$, writing $\mch = \mch_1 \tp \mch_2$ with $\dim\mch_{1} = \dim\mch_{2} = N$. 
		
		Since $\o$ is invertible in $\zn$ by assumption, we may use \eqref{canonical_mps_def} to write down the canonical MPS tensors as
		\be A_\o = \sum_{g\in G} \k g Z^{g_1} X^{-g_2/\o}.\ee 
		That $A_\o$ is perfect follows from the fact that the matrices $X^aZ^b, a,b,\in \zn$ generate $GL(N)$, together with the invertibility of $\o$ in $\zn$.  Note that the canonical MPS tensors for two maximally non-commutative factor sets $\nu$ and $\o$ are linearly related by
		\be A_\nu = \sum_{g\in G} (\proj{g_1} \tp \kb{\nu g_2/\o}{g_2}) A_\o,\ee

		One can further verify that the left and right shift operators are 
		\be S^L_g = Z^{g_1} X^{\o g_2} \tp X^{-\o g_1},\qq S^R_g = X^{-\o g_2} \tp X^{g_1 \o }Z^{g_2} .\ee 
		The Hamiltonian $H = - \sum_{g,i} S^R_{g,i} S^L_{g,i+1} \equiv -\sum_{g,i} H_{g,i}$ is consequently 
		\be H = -\sum_{g,i} \(X^{-\o g_2} \tp X^{\o g_1} Z^{g_2} \tp Z^{g_1} X^{\o g_2} \tp X^{-\o g_1} \)_{i\ra i+3} ,\ee 
		where now we are indexing the sites by the individual tensor factors of dimension $N$. 
		
		While the Hamiltonian as written above contains terms acting on 4 consecutive sites, it can be simplified somewhat as follows. First, for the $\zt$ case, we may write 
		\be H_{\zt^2} = -\sum_i (IIII + IXZX + XZXI + XYYX)_{i\ra i+3}.\ee 
		The first term in parenthesis is a constant and can be dropped. The last term is a product of the first two, and as all terms commute it is thus redundant and can also be dropped, leaving us with the familiar $XZX$ form of the Hamiltonian.  
		
		A similar type of simplification is possible for general $N$ \cite{tsui2015quantum}. 
		Note that $H_{g,i} = H_{(g_1,0),i} H_{(0,g_2),i}$, and that all terms in the Hamiltonian mutually commute. 
		Therefore terms where both $g_1,g_2$ are nontrivial will automatically be minimized if the terms where one of $g_1,g_2$ is trivial are minimized. This means we may get away with only summing over elements of the form $\{(g_1,0), (0,g_2)\}$. This gives 
		\be H_{\zn^2} = - \sum_{g\in \zn} \( \sum_{i\in 2\zz}  X^{-\o g}_i \tp Z^g_{i+1} \tp X^{\o g}_{i+2} + \sum_{i\in 2\zz+1} X^{\o g}_i \tp Z^g_{i+1} \tp X^{-\o g}_{i+2} \).\ee
		Note that the alternating placement of daggers on the two sublattices is actually important to keep track of, and is needed to ensure that all of the terms in $H_{\zn^2}$ commute.

		With this notation the string operators are generated by 
		\be \mcs_{(1,0);i\ra i+2l} = (I_i\tp X^\o_{i+1})\bot_{n=2}^{l-2} (Z_{2n} \tp I_{2n+1}) \tp (Z_{i+2l-1} \tp X^{-\o}_{i+2l})\ee 
		and 
		\be \mcs_{(0,1);i\ra i+2l} = (X^{-\o}_{i}\tp Z_{i+1})\bot_{n=2}^{l-2} (I_{2n}\tp Z_{2n+1} )\tp (X^\o_{i+2l-1} \tp I_{i+2l}).\ee

		\section{General expression for ground states of maximally non-commutative SPT phases } \label{app:general_ground_states}
		
		This section is devoted to proving the following theorem: 
		
		\begin{theorem}[relation between SPT states in the same phase] \label{thm:simpleU} 
			Let $\k{\psi_\o} \in \spt_\o$ be a ground state in the SPT phase identified with the maximally non-commutative cohomology class $\o$. Then for $\k{\Psi_\o}$ the canonical reference state defined in \eqref{canonical_def}, we may always write
			\be\label{psicirc} \k{\psi_\o} = \sum_{\bfg\in G^L} C_\bfg R_\bfg  \k{\Psi_\o},\ee 
			for $\bm{g}$ an $L$-component vector of group elements, and where the $C_\bfg$ are complex coefficients satisfying $\sum_\bfg |C_\bfg|^2 = 1$ and $R_\bfg \equiv  \bot_{i=1}^L R_{g_i,i}.$
		\end{theorem}
		
		The proof proceeds by finding a convenient operator basis in which to decompose the action of the unitary circuit relating $\k{\psi_\o}$ to $\k{\Psi_\o}$. The overall logic is essentially the same as a similar result pertaining to the 2D cluster state proven in Ref. \cite{raussendorf2019computationally}. 
		
		To begin, we need the following lemma: 
		
		\begin{lemma}
			For $\o$ maximally non-commutative, any $n$-site operator $\mco_n$ can be decomposed as 
			\be\label{opdecomp} \mco_n = \sum_{\bfg,\bfh\in G^n} C_{\bfg,\bfh} \bigotimes_{i=1}^n S^R_{g_i,i} R_{h_i,i}\ee 
			for some constants $C_{\bfg,\bfh}$. Furthermore, $\mco_n$ commutes with the symmetry action iff $\sum_i g_i = e$. 
		\end{lemma}
		
		\begin{proof}
			The proof proceeds simply by noting that the operators in the above sum form an orthogonal set which spans $\mch^{\tp n}$. Consider first the case where $n=1$. Then we have  
			\bea \Tr[S^R_g R_h (S^R_{g'} R_{h'})^\da] & = \Tr[S^R_g (S^R_{g'})^\da R_{h-h'}]  \\ 
			& = |G| \d_{g,g'} \d_{h,h'}.\eea 
			Since there are $|G|^2= \dim \End(\mch)$ such operators, they form a basis for all single-site unitaries. Tensoring $n$ copies of them together then gives a basis for the unitaries on $\mch^{\tp n}$. 
			
			Requiring that a given operator is symmetric means that it must commute with $R_k^{\tp n}$ for all $k\in G$. Now 
			\bea \mco_n R_k^{\tp n} & = R_k^{\tp n} \sum_{\bfg,\bfh\in G^n}  C_{\bfg,\bfh} \bigotimes_{i=1}^n  \l_\o(k,g_i) S^R_{g_i,i} R_{h_i,i}   \\ & = 
			R_k^{\tp n} \sum_{\bfg,\bfh \in G^n} C_{\bfg,\bfh} \c_k(\G\inv_\o(\sum_jg_j)) \bigotimes_{i=1}^nS^R_{g_i,i} R_{h_i,i} .
			\eea 
			The RHS is equal to $R_k^{\tp n}\mco_n$ iff $\sum_j g_j=e$, and so any symmetric operator must indeed be of the form claimed. 
		\end{proof}
		
		Since $S^R_{-g} = S^L_{g} R_{-g}\vs_{-g}$, we can equivalently replace $S^R_{-g_i,i}$ operators appearing in the decomposition of a given $\mco_n$ with $S^L_{g_i,i}$ operators, at the cost of modifying the values of the $C_{\bfg,\bfh}$. The combination of $S^{R/L}_{g_i,i}$ operators appearing in the decomposition of any symmetric operator can thus always be manipulated so that it consists only of tensor products of pairs $S^R_{g,i} \tp  S^L_{g,j>i}$. These give the endpoints of a $g$ string operator, and hence by shifting the $R_{h_i,i}$, any symmetric operator may be written in terms of single-site symmetry actions and string operators as 
		\be \label{mcosym} \mco^{\text{sym}} = \sum_{\bfg,\bfh\in G^n}  \bot_{i=1}^n R_{h_i,i} \sum_{j>i}C'_{\bfg,\bfh;j}\mcs_{g_i,i\ra j}\ee 
		where the $C'_{\bfg,\bfh;j}$ are complex numbers determinable in principle from the $C_{\bfg,\bfh}$.

		We can use the above fact to prove theorem \ref{thm:simpleU}. By definition, all states $\k{\psi_\o} \in \spt_\o$ can be related to one another via the action of a constant-depth local\footnote{Strictly speaking the correct adjective is `quasilocal', allowing for the gates appearing in $\mcu$ to have arbitrarily large support, so long as operators with large support are given exponentially small weights. Keeping track of the distinction between local and quasilocal will not be important for the issues that concern us here.} unitary circuit $\mcu_\psi$, where each gate in $\mcu_\psi$ commutes with the symmetry action $R_g^{\tp L}$. 
		By \eqref{mcosym}, all of the gates in $\mcu_\psi$ can be decomposed in terms of the onsite symmetry generators $R_{g_i,i}$ and string operators $\mcs_{g_i,i\ra j}$ --- this leads to the picture shown in Fig~\ref{fig:wfs_and_ec}. We can then move all $\mcs_{g_i,i\ra j}$ operators appearing in the circuit to the `bottom' of the circuit, so that they act directly on $\k{\Psi_\o}$; this can be done merely at the cost of picking up various unimportant complex phases. However, {\it all} string operators act trivially on (i.e. stabilize) $\k{\Psi_\o}$ and hence may be eliminated from $\mcu_\psi$, leaving behind only the $R_\bfh$ operators. The $R_\bfh$ then form of depth-1 circuit, and $\k{\psi_\o}$ can indeed be written in the form of \eqref{psicirc}.
		\\ 
		\qed

		\section{The circuit architecture  } \label{app:qcnn_details} 
		
		\begin{figure}
			\includegraphics{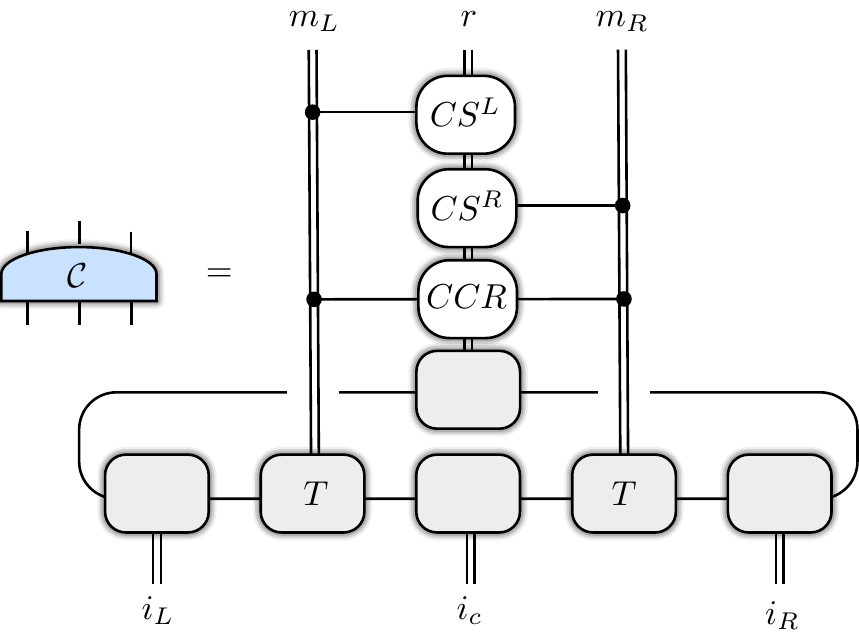} 
			\caption{\label{fig:just_mcc} Definition of the coarse-graining unitary $\mcc$. Grey boxes denote canonical MPS tensors $A_\o$, the $T$ denotes transpose, and the controlled gates $CS^{L/R}$, $CCR$ are defined in the text. }
		\end{figure}
		In this appendix we elaborate on the on the properties of the tensors used to build the coarse-graining convolution unitary $\mcc$, and demonstrate that $\mcc$ performs error correction by implementing the recursive majority function. The explicit form of $\mcc$ is repeated in Fig.~\ref{fig:just_mcc} for convenience. The ring of six MPS tensors serves to coarse-grain the input wavefunction; when acting on $\k{\Psi_\o}$ this part of $\mcc$ simply produces a coarse-grained (by a factor of 3) version of $\k{\Psi_\o}$, tensored with a product state of $|0\ran$s on the remaining sites. The controlled gates activate when nontrivial `errors' act on $\k{\Psi_\o}$, and these gates serve to correct the errors by re-directing the action of the errors away from the renormalized leg $r$, and onto the measurement legs $m_L,m_R$. 
		
		Specifically, we introduce the controlled left-shift operator as
		\bea CS^L & \equiv \sum_{g\in G} \proj{\G_\o(g)} \tp  (S^L_{g})^\da \tp \unit \eea 
		where we have written the operator in the basis $\mch_{m_L}\tp \mch_{r}\tp \mch_{m_R}$ (see Fig.~\ref{fig:just_mcc}). 
		Similarly, we define the controlled right-shift operator as 
		\be CS^R \equiv  \sum_{g\in G} \vs_g^* \unit \tp (S^R_g)^\da \tp   \proj{\G_\o(-g)} .\ee 
		Finally, the double-controlled gate $CCR$ is defined to perform a linear action of the symmetry if the left and right legs are in opposite states, and to act as the identity otherwise: 
		\bea CCR & \equiv \sum_{g}  \proj{\G_\o(-g)} \tp R_g^\da \tp \proj{\G_\o(g)} + \sum_{g\neq g'}    \proj{\G_\o(-g)} \tp \unit \tp \proj{\G_\o(g')} \\ & 
		= \sum_{g,g'} \proj{\G_\o(-g)} \tp ((R_g^\da-\unit) \d_{g,g'} + \unit) \tp \proj{\G_\o({g'})}. \eea 

		To understand the utility of these gates, consider first how a single-site error passes through $\mcc$:
		\begin{prop}
			$\mcc$ corrects all single-site $R_g$ errors acting on $\k{\Psi_\o}$. That is, for an error of the form $R_g \tp R_h \tp R_k$ where only one of $g,h,k\neq e$, 
			\be \mcc (R_g\tp R_h \tp R_k) A_\o^{\tp 3} = \k{\G_\o(h-g)} \tp A_\o \tp \k{\G_\o(k-h)}.\ee
		\end{prop}
		In the above, we have used notation where tensor products are in physical space, with virtual indices being implicitly contracted as appropriate. Thus e.g. $A^{\tp 3}_\o = \sum_{\bfg \in G^3} \k{\bfg} A_\o^{g_1} A_\o^{g_2} A_\o^{g_3}$ is a tensor valued in $\mch_{\text{virt}} \tp \mch^{\tp 3} \tp \mch_{\text{virt}}$. 
		\begin{proof}

			We demonstrate this simply by explicitly computing the action of $\mcc$ in all situations where a single error is present. Suppose an error of $g$ occurs on the leftmost input site $i_l$. Understanding how this error travels through $\mcc$ can be done graphically as follows: 
			\be \label{single_error} \includegraphics{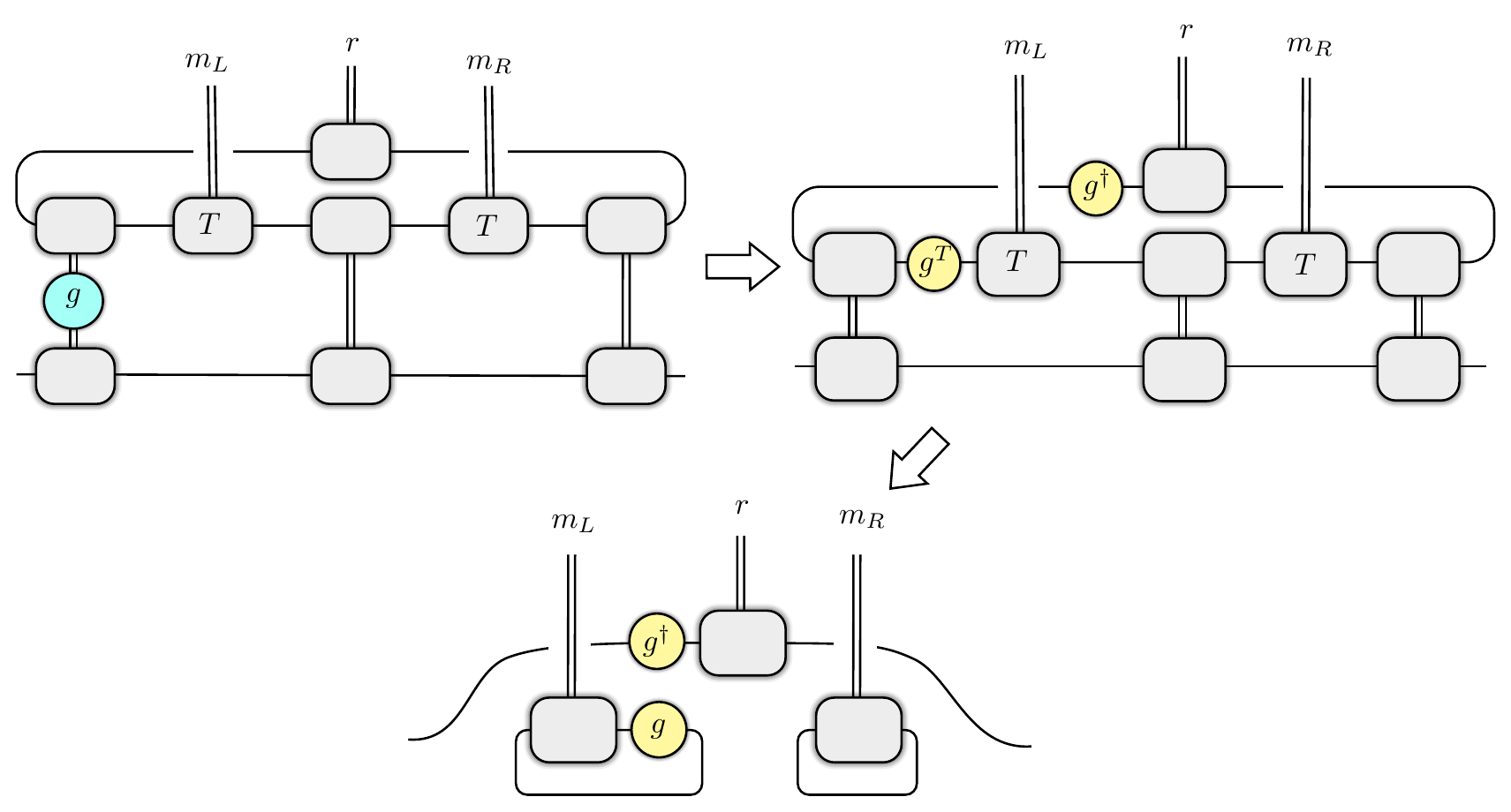} \ee 
			where the cyan and yellow circles stand for linear and projective representations respectively, and where in the second step we have used the injectivity condition \eqref{graphical_inj}. The states of the measurement legs $m_{L/R}$ can be worked out using 
			\be \label{traced_vg_app}  \sum_{h\in G} \k{h} \Tr[A^h_\o V_g] = \igptwoc{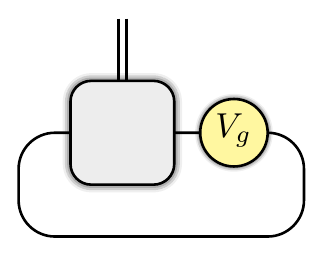} = \sqrt{|G|} \k{\G_\o(g)}\ee
			since $\Tr[V_g^\da V_h] = \sqrt{|G|} \d_{g,h}$. Thus in this case, the left leg $m_L$ is in the state $\k{\G_\o(g)}$, and the right leg $m_R$ is in the state $\k{e}$. The job of the control gates $CS^{L/R}, CCR$ is to use the information contained in these states to correct for the action of $S^L_g$ that has been applied to the renormalized leg $r$ (due to the $V_g^\da$ projective rep appearing in the last panel of \eqref{single_error}). In this case, since $m_R$ is in the state $\k{e}$, only the $CS^L$ gate activates. This gate applies $(S^L_g)^\da$ to the renormalized leg, which cancels the $S^L_g$ error. Thus indeed, 
			\be \mcc (R_g \tp \unit \tp \unit) A_\o^3= \k{\G_\o(g)} \tp A_\o \tp \k{e}.\ee 
			
			Similarly, consider what happens when an error $R_k$ occurs on the rightmost site $c$. Similar arguments show that the right measurement module is left in the state $\vs_k \k{\G_\o(-k)}$, the left module in the sate $\k{\G_\o(k)}$, and that renormalized leg incurrs an error of $S^R_k$. Only the $CS^R$ gate activates, which applies a corrective action of $(S^R_k)^\da$ to the renormalized leg, cancelling the error. Therefore 
			\be \mcc (\unit \tp \unit \tp R_k) A_\o^{\tp 3}= \k{e} \tp A_\o \tp \k{\G_\o(-k)}.\ee 
			
			Finally, consider an error of $R_h$ applied to the central leg. This leaves the left module in the state $\vs_h\k{\G_\o(-h)}$, the right module in the state $\k{\G_\o(h)}$, and leads to no error on the renormalized leg. Both the $CS^{L/R}$ gates activate, and conseuqnetly the renormalized leg is acted on by $\vs_h\vs_{-h}^* (S^L_{-h})^\da  (S^R_{-h})^\da = \l_\o(-h,h) (R_{-h})^\da = R_h$. We see that the $CCR$ gate has been chosen to exactly compensate for this action; thus 
			\be \mcc (\unit \tp R_h \tp \unit) A_\o^{\tp 3} = \k{\G_\o(-h)} \tp A_\o \tp \k{\G_\o(h)}.\ee 
			
		\end{proof}
		
		We also have the necessary property that $\mcc$ preserves the global symmetry action: 
		\begin{prop}
			The global symmetry commutes with $\mcq_\o$, in the sense that the global symmetry operator $R_g^{\tp L}$ is mapped to $R_g^{\tp L'} \tp \unit_M$ upon commuting with $\mcq_\o$, where $L' = L/3$ and $\unit_M$ denotes the identity on the subspace spanned by the measured legs.
		\end{prop}
		
		\begin{proof}
			The proof is straightforward, and does not require us to project onto states generated by $A_\o$. Without loss of generality, we can focus on the commutation relation between $R_g^{\otimes 3}$ and a single convolution module $\mcc$. First, we see that $R_g^{\tp 3}$ passes through the layer below the control gates in $\mcc$ unchanged. Since the controls occur in the $R_g$ basis, we just need to compute how $R_g$ on the renormalized leg passes through the $CS^{L/R}$ gates. For the $CS^L$ gate, 
			\bea CS^L (\unit \tp R_g \tp \unit) & = \sum_{h\in G}  \proj{\G_\o(h)} \tp (S^L_h)^\da R_g \tp \unit \\ & = 
			\sum_{h\in G} \c_g(\G_\o(h)) \proj{\G_\o(h)} \tp R_g (S^L_h)^\da \tp \unit   \\
			& = (R_g^\da \tp R_g \tp \unit ) CS^L,\eea 
			where we used $(S^L_h)^\da R_g = \l_\o(-g,h) R_g (S^L_h)^\da = \c_{-g}(\G_\o(h)) R_g (S^L_h)^\da$. One can similarly verify that 
			\be CS^R (\unit \tp R_g \tp \unit ) = (\unit \tp R_g \tp R_g^\da) CS^R.\ee 
			Therefore when the $R_g$ on the renormalized leg passes through the control gates, it cancels the two applications of $R_g$ on the measurement legs. Thus 
			\be \mcc R_g^{\tp 3} = (\unit \tp R_g \tp \unit )\mcc\ee
			at the level of operators (i.e. not just when acting on states generated by the MPS tensors $A_\o$). 
			
		\end{proof}
		
		To define the action of $\mcc$ on general errors of the form $R_g \tp R_h \tp R_k$, we define a majority vote function as follows: 
		\begin{definition}[majority vote rule]
			The function $\maj : G^3 \ra G$ is defined to implement majority vote on $G$. In the case of ties, $\maj$ selects out the center element of the tuple $(g,h,k) \in G^3$. Explicitly, 
			\be \maj(g,h,k) \equiv \begin{dcases} g \quad & {\rm if} \, \, (h=g)\vee( k=g) \\ 
				h\quad & {\rm if} \, \, (g=h) \vee (k=h) \\ 
				k \quad & {\rm if}\, \, (g=k)\vee (h=k) \\ 
				h \quad & {\rm else}\end{dcases}\ee 
		\end{definition}
		The previous two propositions then lead to 
		\begin{corollary}
			$\mcc$ implements the majority vote rule on $R_\bfg$ operator strings which act on the canonical ground state $\k{\Psi_\o}$. That is, 
			\be \mcc (R_g \tp R_h \tp R_k)  A_\o^{\tp 3} = \k{\G_\o(h-g)} \tp R_{\maj(g,h,k)} A_\o \tp \k{\G_\o(k-h)}.\ee  
			Graphically, 
			\be \label{mcc_maj} \igpfoc{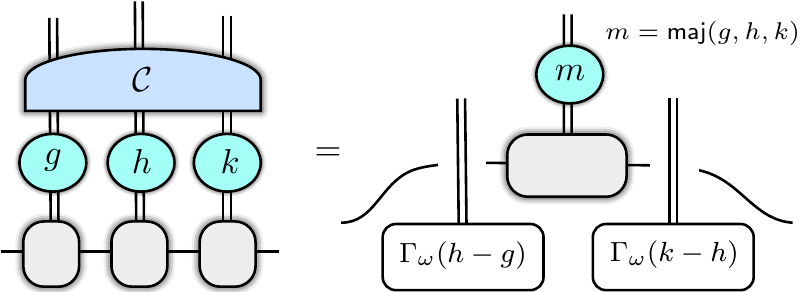}\ee 
		\end{corollary}
		
		Note that this property is not satisfied by the QCNN used in \cite{cong2019quantum} to study the $\zt\times \zt$ cluster state. 
		
		\begin{proof}
			Since the global symmetry pushes through $\mcc$ and since single errors are corrected, this property obviously holds in the 3:0 situation when $g=h=k$. When we are in a 2:1 vote situation, we have e.g. 
			\be \mcc (R_g \tp R_g \tp R_h) A_\o^{\tp 3} = (\unit \tp R_g \tp \unit) \mcc (\unit\tp \unit \tp R_{h-g}) A_\o^{\tp 3}=  \k{e} \tp R_gA_\o \tp \k{\G_\o(h-g)},\ee 
			and likewise for the two other 2:1 vote configurations. Since $R_g$ acts on the renormalized leg in all cases, the 2:1 vote works in the way claimed.
			
			Consider then the case of a tie, where $g\neq h \neq k$. In this case the $R_g^\da$ part of the $CCR$ gate will be unactivated as $g\neq k$. On the other hand, both the $CS^{L/R}$ gates will activate. In this case, the left measurement module is left in the state $\vs_h \o_{g,-h} \k{\G_\o(g-h)}$, the right module is left in the state $\vs_k\o_{h,-k}\k{\G_\o(h-k)}$, and the renormalized leg incurs an error of $V_g^\da \tp V_k$. After the $CS^{L/R}$ gates activate, we get a combined action on the renormalized leg of 
			\bea \vs_h \vs_k \vs^*_{k-h}\o_{g,-h}\o_{h,-k}  (S^L_{g-h})^\da (S^R_{k-h})^\da  & = R_h (S^L_g)^\da (S^R_k)^\da, \eea 
			where we used $(S^R_{k-h})^\da = \o_{k,-h} (S^R_{-h})^\da (S^R_k)^\da$, $(S^L_{g-h})^\da = \o^*_{g,-h}(S^L_{-h})^\da (S^L_g)^\da$, and 
			where we made use of the fact that $\o$ is a bicharacter to write $\vs_h \vs_k \vs_{k-h}^* = \o_{k,h}\o_{h,k}$. The action of the suriving shift operators cancels the error on the renormalized leg, and so in this case, 
			\be \mcc (R_g\tp R_h \tp R_k) A_\o^{\tp 3} =  \k{\G_\o(h-g)} \tp R_{h}A_\o \tp \k{\G_\o(k-h)} .\ee

		\end{proof}

		Note that the above majority-vote property is not true at the level of operators; it relies projecting onto the space generated by the MPS tensors $A_\o$.  
		
		From the above, we know that $R_g$ operators on renormalized legs push `down' through $\mcc$ to become $R_g^{\tp 3}$. To determine how an arbitrary operator on a renormalized leg pushes down through $\mcc$, we need only determine what happens to $S^L_g$. We do this with the following proposition: 
		
		\begin{prop}
			The left shift operator $S^L_g$ pushes `down' through the circuit as  (recall $\int_h \equiv \frac1{|G|} \sum_{h\in G}$)
			\bea \label{renormed_s} 
			\mcc^\da (\unit \tp S^L_g \tp \unit)\mcc & =  R_g \tp S^L_g \tp \unit +  \( S^L_g \tp \unit \tp \unit  - R_g \tp S^L_g \tp \unit \) \int_k S^R_k \tp R_k \tp S^L_k. \eea

		\end{prop}
		
		\begin{proof}
			$S^L_g$ moves through the $CS^R$ gate for free (as $S^L_g, S^R_h$ commute), and pushes through the $CS^L$ gate as
			\bea (\unit \tp S^L_g\tp \unit ) CS^L  & = 	\sum_{h\in G}  \proj{\G_\o(h)} \tp S^L_g (S^L_h)^\da \tp \unit \\ 
			& = \sum_{h\in G}  \c_g(\G_\o(h)) \proj{\G_\o(h)} \tp (S^L_h)^\da S^L_g \tp \unit \\ 
			& = CS^L(R_g \tp S^L_g \tp \unit)\eea 
			where in the first line we used $S^L_g (S^L_h)^\da = \l_\o(-h,g) (S^L_h)^\da S^L_g = \c_g(\G_\o(h)) (S^L_h)^\da S^L_g$. 
			Now we pull $R_g \tp S^L_g\tp \unit$ through the $CCR$ gate: 
			\bea (R_g \tp S^L_g \tp \unit) CCR & = CCR \Big( \sum_{h\in G} R_g^\da \proj{\G_\o(-h)} \tp \unit \tp \proj{\G_\o(h)} \\ & \qq + \sum_{h\neq h' \in G}  \proj{\G_\o(-h)} \tp \unit \tp \proj{\G_\o(h')}\Big)(R_g \tp S^L_g \tp \unit)\eea
			where we similarly used $S^L_g R_h^\da = \c_g(\G_\o(h)) R^\da_h S^L_g$. 
			The second gate on the RHS has the structure of a controlled-$R_g^\da$ gate acting on the first $\tp$ factor, with the third $\tp$ factor acting as the control. To simplify this, we use
			\be \proj{g} = \int_{h} R_h \c_h^*(g) \implies \proj{\G_\o(g)} = \int_h R_h \l_\o^*(h,g),\ee 
			which allows us to replace the double control gate with a sum over $R_h$ operators. This then lets us write the expression in parenthesis on the RHS of the above equation as 			
			\bea |G|\int_{h,h'} & (\d_{h,h'}( R_g^\da-\unit) + \unit) \proj{\G_\o(-h)} \tp \unit \tp \proj{\G_\o(h')}   \\ & = |G|^2 \int_{h,h',k,k'} \l_\o(k,h) \l_\o(k',-h') ((R_g^\da -\unit)\d_{h,h'} + \unit) (R_k \tp \unit \tp R_{k'}) \\
			& = \int_k R_k(R_g^\da -\unit) \tp \unit \tp R_k + \unit^{\tp 3}.\eea 
			While it is not completely obvious in this presentation, the RHS is indeed unitary as required. 
			
			Pushing $\unit \tp S^L_g \tp \unit$ through the control gates in $\mcc$ thus generates the operator $(\int_k (R_g^\da -\unit)R_k \tp \unit \tp R_k + \unit^{\tp 3})(R_g\tp S_g^L\tp \unit)$. 
			Now we push these terms through the ring of six MPS tensors in $\mcc$. This can be done by noting that $R_g \tp S^L_g \tp \unit$ pushes through the ring unchanged, $R_g^\da \tp \unit \tp \unit$ pushes through to $S^R_{-g} \tp S^L_{-g} \tp \unit$, while $R_k \tp \unit \tp R_k$ pushes through to a 3-site $k$ string operator: 
			\be R_k \tp \unit \tp R_k \ra S^R_k \tp R_k \tp S^L_k.\ee 
			Therefore as claimed, 
			\bea \mcc^\da (\unit \tp S^L_g \tp \unit)\mcc & =  R_g \tp S^L_g \tp \unit +  \( S^L_g \tp \unit \tp \unit  - R_g \tp S^L_g \tp \unit \) \int_k S^R_k \tp R_k \tp S^L_k, \eea 
			where we used $S^R_{-g} R_g \tp S^L_{-g} S^L_g = S^L_g \vs_g^* \tp \vs_{-g} \unit = S^L_g \tp \unit$. This expression may also be written more explicitly as 
			\bea 
			\mcc^\da (\unit \tp S^L_g \tp \unit)\mcc & =  R_g \tp S^L_g \tp \unit (1-1/|G|) + \frac1{|G|} S^L_g \tp \unit \tp \unit + \frac1{|G|} R_g \tp R_g \tp S^L_g \\ & \qq + \frac1{|G|} (S^L_g \tp \unit \tp \unit - R_g \tp S^L_g \tp \unit) \sum_{G \ni k \neq e,g} S^L_k \tp R_k \tp S^L_k - \frac1{|G|} R_g S^R_g  \tp S^L_g R_g \tp S^L_g \eea

		\end{proof}
		
		The point of this rather messy expression is to show that $S^L_g$ pushes through the RG circuit to become a renormalized operator which still carries the same irrep under the adjoint $G$-action as $S^L_g$ itself (an analogous calculation for $S^R_g$ shows that it pushes through to something where the `net' shift is by $S^R_g$ for all terms). Note furthermore that this symmetry-property means that all of the terms on the RHS of \eqref{renormed_s} mutually commute with each other; this can be easily demonstrated explicitly using the commutation relations \eqref{slwithr} and \eqref{srwithr}. Finally, note that upon iterating this map $d$ times, the number of terms appearing in the resulting linear combination is {\it doubly}-exponential in $d$ --- but due to the symmetry, all such terms transform under the adjoint action in the same way that $S^L_g$ itself does. 
		
		\section{Convergence of the RG circuits } \label{app:convergence} 
		
		In this section we will discuss the convergence of our circuits when run on a generic state $\k{\psi_\o} \in \spt_\o$. We will first prove that the error correction procedure always succeeds in the limit of large circuit layer depth $d$, so that the error correction threshold can be identified exactly with the SPT phase boundary (up to a set of measure zero where fine-tuned degeneracies prevent convergence). A subsequent section contains a numerical study of the convergence for a few types of random error models.  In the final section, we show that the circuit does not `over-correct' and produce false-positives. 
		
		A note on notation: in this section, we will write $\mcq_\o(\r)$ for the action of the RG circuit on the density matrix $\r$ followed by a tracing out of the ancilla degrees of freedom:
		\be \mcq_\o(\r) = \Tr_{an}[\mcq_\o \r \mcq_\o^\da].\ee

		\ss{Analytic arguments: convergence of majority vote} 
		
		Recall that $\k{\psi_\o}$ can always be written as 
		\be \k{\psi_\o} =\sum_\bfg C_\bfg R_\bfg \k{\Psi_\o},\ee 
		where $\sum_\bfg |C_\bfg|^2 = 1$. \footnote{Note that on a closed ring, the normalization condition is actually 
			\be \sum_{\bfg,\bfg'} \sum_{h\in G} \d_{\bfg,\bfg'+h} C_\bfg^*C_{\bfg'} = 1,\ee 
			where $\bfg'+h$ means adding $h$ to each component of $\bfg'$. This is just because the global symmetry operator stabilizes $\k{\Psi_\o}$. However, we can wolog assume that no two $C_\bfg$ differ by a constant factor of $h$ on every site, since $C_\bfg$ and $C_{\bfg+h}$ act identically on $\k{\Psi_\o}$. }. 
		As usual, for convenience we will let $\k{\psi_\o}$ be defined on a chain of length $L = 3^l$. Let $\k{\Psi_\o^{(d)}}$ denote the canonical representative of $\spt_\o$ defined on a chain of length $3^{l-d}$. 
		
		\begin{prop}
			After $d$ applications of  $\mcq_\o$ on the initial state $\proj{\psi_\o}$, the probability $p_d(\bfg)$ that a given error $R_\bfg$ acts on the state $\mcq_\o^d(\proj{\psi_\o})$ is 
			\be p_d(\bfg) =\lan \Psi_\o^{(d)} | R_\bfg^\da \mcq_\o^d(\proj{\psi_\o}) R_\bfg \k{\Psi_\o^{(d)}}  =\sum_{\bfh \in G^L} \d(\maj^d(\bfh)-\bfg) |C_\bfh|^2.\ee 
		\end{prop}

		\begin{proof}
			$d$ applications of $\mcq_\o$ on $\proj{\psi_\o}$ produce the state 
			\be \mcq_\o^d(\proj{\psi_\o}) = \sum_{\bfh,\bfh'\in G^L} C_\bfh C_{\bfh'}^* R_{\maj^d(\bfh)} \proj{\Psi_\o^{(d)}} R^\da_{\maj^d(\bfh')}  \Tr[ \k{\mcm_d(\bfh)} \lan \mcm_d(\bfh')|] \ee 
			where $\k{\mcm_d(\bfh)}$ are the measurement outcomes generated by $d$ layers of the circuit. 
			The probability $p_d(\bfg)$ of getting an error is then found by taking the overlap with the state $R_\bfg\k{\Psi^{(d)}_\o}$, giving by way of \eqref{rgoverlap}
			\bea p_d(\bfg) 
			& =\sum_{\bfh,\bfh'\in G^L} C_\bfh C_{\bfh'}^* \lan \Psi_\o^{(d)} | R_{\maj^d(\bfh)-\bfg}\k{\Psi_\o^{(d)}}\lan \Psi_\o^{(d)}| R_{\maj^d(\bfh')-\bfg}^\da \k{\Psi_\o^{(d)}} \d_{\mcm_d(\bfh)_M,\mcm_d(\bfh')_M} \\ 
			& = \sum_{\bfh,\bfh' \in G^L} C_\bfh C^*_{\bfh'} \d_{\mcm_d(\bfh),\mcm_d(\bfh')} \d_{\bfg,\maj^d(\bfh)} \d_{\maj^d(\bfh),\maj^d(\bfh')} \\  &  
			= \sum_{\bfh\in G^L \, : \, \maj^d(\bfh)  = \bfg} |C_\bfh|^2,\eea 
			which is positive and satisfies $\sum_{\bfg\in G^{L/3^d}} p_d(\bfg) = 1$. 
		\end{proof}
		
		The probability distribution of errors thus flows under application of the RG circuit according to recursive majority voting.
		To understand this flow, we need to know a bit more about the initial distribution $p_0(\bfg) = |C_\bfg|^2$. 
		
		First, we will assume that $p_0(\bfg)$ is translation-invariant, meaning that e.g. the single-site marginal probability $p_0(g_i)$ of finding a $g$ error on site $i$ is independent of $i$. This assumption is made purely for simplicity of presentation, and is not essential. Second, since $C_\bfg$ can be identified with the MPO of $R_\bfg$ operators that relates $\k{\Psi_\o}$ to $\k{\psi_\o}$, $p_0(\bfg)$ will have a finite correlation length $\xi$. Letting 
		\be p_d(g_{i_1},\dots,g_{i_n}) = \sum_{\{g_j\} \, : j \, \not \in \{i_1,\dots,i_n\}} p_d(\bfg) \ee 
		denote the $n$-site marginals of $p_d(\bfg)$, this means that $p_0(g_i,g_j) \approx p_0(g_i)p_0(g_j) $ when $|i-j|>\xi$. In fact in the approximation where $\k{\psi_\o}$ is obtained from $\k{\Psi_\o}$ from a constant-depth local (as opposed to quasilocal) circuit --- which is good enough for our purposes --- $p_0(g_i,g_j)$ {\it exactly} factorizes for all $|i-j|>\xi$, by causality. 
		
		Finally, we will assume that $p_0(\bfg)$ is generic, meaning that there are no accidental degeneracies among the single-site marginals $p_0(g)$. More specifically, we will assume that for all $g\neq h$, $|p_0(g) - p_0(h)|$ is nonzero and independent of system size (e.g. does not vanish as $1/L$). Situations where this assumption fails are measure zero in the space of possible probability distributions generated by the circuits relating $\k{\psi_\o}$ to $\k{\Psi_\o}$, and hence will be ignored.  
		
		Our goal in the following is to prove the following claim, which ensures $\mcq_\o^{d\ra \infty}(\proj{\psi_\o}) = \proj{\Psi_\o},\,\,\, \forall\,\, \k{\psi_\o}\in \spt_\o$, so that the RG flow converges on the correct reference state: 
		
		\begin{claim}[convergence of majority vote]\label{claim:maj_conv}
			For a generic translation-invariant starting probability distribution $p_0(\bfg)$ with finite correlation length, the sequence of distributions $\{ p_d(\bfg)\}$ converges as $d\ra\infty$ to the product distribution $p_{d\ra\infty}(\bfg) = \prod_i \d_{g_i,g_*}$, for some fixed $g_*\in G$. 
		\end{claim}

		Claim \ref{claim:maj_conv} follows easily if one makes an additional (reasonable) assumption about the initial distribution of errors. 
		
		\begin{definition}[dominance property]
			For a probability distribution $p(\bfg)$, let $g_1$ be the element for which the single-site marginal probability $p(g)$ is maximized; $p(g) \leq p(g_1) \, \, \forall \, g \in G$. Let $p(g,h,k) = p(g_i,h_{i+1},k_{i+2})$ denote the 3-site marginal probability, for arbitrary $i$. The distribution $p(g)$ is said to be dominant if for all $G\ni g \neq g_1$,
            \be \label{dominance}\sum_{h\in G} (p(g_1,h,g_1) - p(h,g_1,h)) > \sum_{h\in G} (p(g,h,g) - p(h,g,h)).\ee 
		\end{definition}
		
		This property essentially means that if the single-site marginals $p(g)$ are maximized on $g_1$, then for two strings $\bfg$ with the same number of domain walls (occurrences of adjacent non-identical group elements in $\bfg$), the one with more occurrences of $g_1$ is the one with the higher probability. To see this, we note that the difference between the LHS and RHS of \eqref{dominance} can be written as 
        \bea \label{elaborated_dominance}\sum_{h\in G} [(p(g_1,h,g_1)-p(g,h,g)) - (p(h,g_1,h) & - p(h,g,h))]  = 2 (p(g_1,g,g_1) - p(g,g_1,g)) \\ & + \sum_{G \ni h \neq g_1,g} [(p(g_1,h,g_1) - p(g,h,g)) - (p(h,g_1,h) - p(h,g,h))].\eea 
  We expect the RHS of \eqref{elaborated_dominance} to be positive for distributions associated with ground states of all `reasonable' Hamiltonians. Indeed, the first term proportional to $p(g_1,g,g_1) - p(g,g_1,g)$ will be positive if for two configurations with the same number of domain walls, the one with more instances of $g_1$ has higher probability. If this property holds, the differences $d_1(h) = p(g_1,h,g_1) - p(g,h,g),\, d_2(h) = p(h,g_1,h) - p(h,g,h)$ for $h \neq g_1,g$ will also be positive. The RHS of \eqref{elaborated_dominance} will then be positive if $d_1(h) > d_2(h)$, which is reasonable since more instances of $g_1$s appear in $d_1(h)$. 
		
	To see why distributions which are dominant converge as in claim \ref{claim:maj_conv}, we note that under majority vote, the single-site marginals at layer $d+1$ are obtained from the single- and three-site marginals at layer $d$ as 
		\bea \label{pd_flow} p_{d+1}(g) & = \sum_{h\in G} p_d(g,h,g) + \sum_{h,k\in G} p(h,g,k) - \sum_{h\in G}p(h,g,h) \\ & = p_d(g) + \sum_{h\in G} (p_d(h,g,h) - p_d(g,h,g)).\eea 
		Let $\{g_a\},\,  a=1,\dots,|G|$ be an ordering of the elements of $G$ sorted by decreasing probability, so that $a\leq b \implies p_d(g_a) \geq p_d(g_b)$. Define the gap 
		\be\label{deltadef_app} \d_d \equiv p_d(g_1) - p_d(g_2)\ee 
		to be the difference in probabilities between the most likely and second-most likely group elements. The gap then evolves under $\maj$ as 
		\be \label{delta_flow} \d_{d+1} - \d_{d} = \sum_{h\in G} [(p_d(h,g_1,h) - p_d(g_1,h,g_1)) - (p_d(h,g_2,h) - p_d(g_2,h,g_2))].\ee 
		If $p_d(\bfg)$ obeys the dominance property, the gap clearly monotonically increases with $d$. A generic initial distribution satisfying the dominance property thus flows until the gap reaches the maximal size of $\d_{d\ra \infty} = 1$, implying that all of the probability is concentrated on a single group element, and that the RG flow converges. 
		
		We have to work a bit harder when the initial distribution is not dominant, because in this case $\d_d$ needn't be monotonically increasing with $d$. Indeed, it is possible to construct examples in which $\d_d$ changes nonmonotonically. These examples are rather contrived, but do not appear to be measure-zero in the space of distributions. We refer to this phenomenon as Gerrymandering, for obvious reasons: 
		
		\begin{definition}
			A series of distributions $p_d(\bfg),d\in \nn$ related by recursive majority voting is said to Gerrymander if there is some $d_*$ for which the single-site marginal $p_{d>d_*}(g)$ attains its maximum on $g_>$, while $p_{d\leq d_*}(g)$ attains its maximum on $g_< \neq g_>$. 
		\end{definition}
		
		Any Gerrymandering distribution cannot be dominant, and as we will see shortly, can also not be iid. It is however possible to construct translation-invariant Gerrymandered distributions using inter-site correlations. For example, we may construct a distribution whose initial 3-site marginals are 
		\be \label{gerry} p_0(g,h,k) = \(1-(\d_{h,g_*}\d_{k,g_*} + \d_{k,g_*}\d_{g,g_*} + \d_{h,g_*} \d_{g,g_*}) + 2\d_{h,g_*}\d_{k,g_*}\d_{g,g_*}\) \( \a \frac{\d_{g_* \in \{g,h,k\}}}{3(|G|-1)^2} + (1-\a) \frac{1-\d_{g\in \{g,h,k\}}}{(|G|-1)^3} \),\ee 
		where $0<\a<1$ and $\d_{a\in \{g,h,k\}}$ is unity if $a \in \{g,h,k\}$, and zero else. 
		The first factor in \eqref{gerry} is to ensure that no two $g_*$s occur within a distance of 3 from one another in a way consistent with translation invariance, implying $p_1(g_*) = 0$. However, by increasing the value of $\a$ one can arrange for $p_0(g)$ to be maximized at $g_*$ as long as $|G|>2$ ($\a<1$ is needed for translation invariance, which is spontaneously broken at $\a=1$). This phenomenon is robust with respect to relaxing the hard-core constraint on the distribution of $g_*$s, and thus we expect Gerrymandering to occupy a set of nonzero measure in parameter space.

		We now argue that a general distribution (even if it Gerrymanders) will still converge in the manner of claim \ref{claim:maj_conv}. We argue in two stages: first, we prove that if $p_d(\bfg)$ is iid for some $d$, meaning that $p_d(\bfg) = \prod_i p_d(g_i)$, then $p_{d\ra\infty}(\bfg)$ converges in the way claimed. Secondly, we argue that for a generic initial distribution $p_0(\bfg)$, $p_d(\bfg)$ can be well-approximated as iid for all $d$ sufficiently greater than $\log_3(\xi)$. 
		
		\begin{prop}
			Let $p_0(\bfg)$ be iid, with gap $\d_0$. For any $\ep>0$, there exists a $d_\ep$ such that 
			\be \label{maxlcond} \d_{d_\ep} > 1-\ep,\ee 
			so that $d_\ep$ sets the scale when $p_d(g)$ is $\ep$ close to its $d\ra\infty$ limiting value. 
			The scale $d_\ep$ is approximately upper bounded by 
			\be d_\ep \leq \ln(1/\ep) + |G|\ln(1/\d_0).\ee
		\end{prop}
		
		\begin{proof}
			
			We give a rather lowbrow proof that doesn't make use of any information-theoretic properties of the flow of $p_d$. 
			From \eqref{delta_flow}, in the iid case we have 
			\be \label{iid_delta_flow} \d_{d+1} - \d_d = p_d(g_1)^2 - p_d(g_2)^2-\d_d \sum_{h\in G} p_d(h)^2 = \d_d\( 2p_d(g_1)(1+\d_d -p_d(g_1)) - \d_d(1+\d_d) - \sum_{G \ni h\neq g_1,g_2} p_d(h)^2 \).\ee 
			Note that the RHS is bounded as  
			\be  p_d(g_1)^2 - p_d(g_2)^2 -\d_d \sum_{h\in G} p_d(h)^2 \geq  \d_d p_d(g_2) \geq 0,\ee 
			so that the gap grows monotonically with $d$ (with $\d_{d+1}-\d_d >0$ as long as $p_d(g_2) >0$ and $\d_d>0$); the iid case thus satisfies the above dominance condition, and is guaranteed to converge when $d\ra\infty$. 
			
			Now we estimate the convergence time. 
			For a fixed $\d_d$, the difference $\d_{d+1} - \d_d$ is readily checked to be minimized on the distribution where all $p_d(h\neq g_1)$ are identical: 
			\be \label{almost_uniform} p_d(g_1) = \frac{1-\d_d}{|G|} + \d_d,\qq p_d(h\neq g_1) = \frac{1-\d_d}{|G|},\ee 
			with the form of this distribution being preserved under $\maj$. This is also the distribution which minimizes $p_d(g_1)$ for fixed $\d_d$.
			While \eqref{almost_uniform} contains non-generic degeneracies among the $p_d(h\neq g_1)$, we can consequently use it to upper-bound the time needed for convergence to be achieved for a fixed starting $\d_0$. Inserting this distribution on the RHS of \eqref{iid_delta_flow}, 
			\be \frac{\d_{d+1} - \d_d}{\d_d} = \frac{(1-\d_d)(1+\d_d(|G|-1))}{|G|}.\ee 
			
			Given an initially small gap of $\d_0$, we may write the above for small $d$ as the differential equation $\dot \d_d \approx \d_d / |G|$. Therefore from an initial small value, $\d_d$ will grow to be of order 1 by the scale 
			\be d_{1} \approx  |G| \ln \frac1{\d_0}.\ee 
			After the flow reaches $d_1$, we may expand in small $\eta_d \equiv 1 - \d_d$. In terms of $\eta_d$, we have 
			\be \eta_{d+1}-\eta_d = -\eta_d(1-\eta_d) \( (1-\eta_d) - \frac{\eta_d}{|G|-1}\) = -\eta_d + O(\eta_d^2), \ee
			and so after reaching $d_1$, we will reach the scale $d_2$ at which $\d_d > 1-\ep$ for $d_2 \approx d_1 + \ln 1/\ep$. Therefore up to constants, we approach the limiting distribution within a distance $\ep$ after $d_2$ steps, with 
			\be d_2 = |G| \ln \frac1{\d_0} + \ln \frac1\ep. \ee 
			Since the above analysis was performed in the worst case, where the flow to $\d_d > 1-\ep$ was as slow as possible, the scale $d_\ep$ is upper bounded in the way claimed. 
			
		\end{proof}

		We now need to argue the following: 
		\begin{prop}
			Let $p_0(\bfg)$ be a generic distribution with finite correlation length $\xi$. Then for $d \gg \log_3(\xi)$, $p_d(\bfg)$ is approximately iid. That is, for any string $\bfg_n$ of length $n\sim O(1)$ \footnote{We specialize to such strings rather looking at the individual probabilities for $\bfg$ with $|\bfg|=L$, since for the latter case the accumulated error will depend on $L$, and hence not be small.}, the marginal probabilities factor as 
			\be p_d(\bfg_n) = \prod_{g\in \bfg_n} p_d(g) + \cdots,\ee 
			where the $\cdots$ are terms exponentially suppressed in $d/\log_3(\xi)$. 
		\end{prop}

		\begin{figure}
			\includegraphics{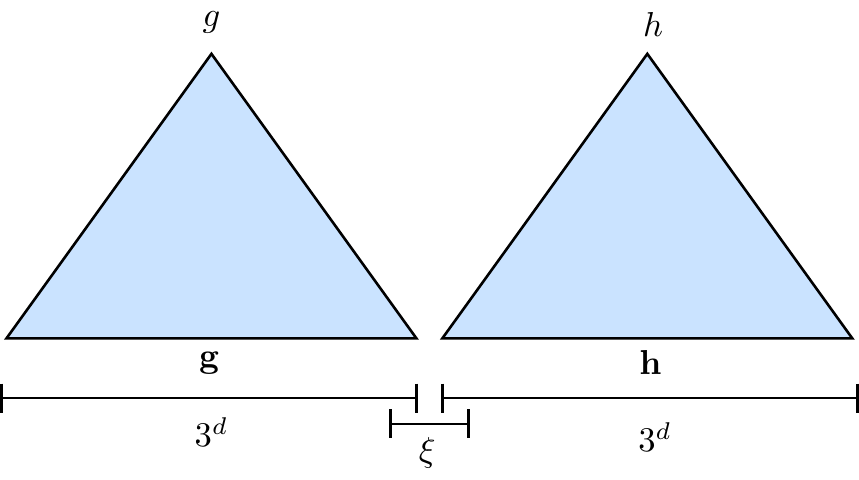}
			\caption{\label{fig:indep_cones} An illustration of why $p_d(g,h)$ approximately factors as $p_d(g)p_d(h)$ when $d\gg \log_3(\xi)$. Each blue triangle represents a $\maj$ tree of depth $d$. The input variables $\bfg,\bfh$ (satisfying $\maj^d(\bfg)=g,\maj^d(\bfh)=h,|\bfg|=|\bfh|=3^d$) are only correlated over the scale $\xi$; thus the only correlations between $g$ and $h$ arise from the variables within the indicated region of width $\xi$. Since the variables wihin this region have a vanishingly small influence on the outcomes of $\maj^d(\bfg),\maj^d(\bfh)$ in the limit $\xi \ll 3^d$, $p_d(g,h)$ approximately factorizes. } 
		\end{figure}
		
		\begin{proof}
			As summarized in Figure~\ref{fig:indep_cones}, this is a very reasonable statement. Consider for simplicity the 2-site marginal probability $p_d(g,h)$ (more complicated marginals are treated in a similar way). For large $d$, $p_d(g)$ is determined by the distribution $p_0(\bfg_d)$, where $\bfg_d$ has length $|\bfg_d| = 3^d$ and runs over variables in the $\maj$ tree of depth $d$ that outputs $g$. Since $p_0(\bfg_d)$ is correlated only over the scale $\xi$, when $3^d \gg \xi$ the variables at the bases of two adjancent depth-$d$ $\maj$ trees contributing to $p_d(g,h)$ will recieve correlations only from a region of size $\xi$ located where the bases of the two trees touch. Since the variables in this region have almost no influence on the outcomes of the two majority trees, $p_d(g,h)$ consequently nearly factorizes (see Fig.~\ref{fig:indep_cones}). 
			
			Now we try to make this inuitive argument more rigorous.  As mentioned above, we can without loss of generality take correlations in the initial distribution $p_0(\bfg)$ to be exactly vanishing on scales greater than $\xi$. This consequently implies that as soon as $3^d > \xi$, the correlations in $p_d(\bfg)$ are purely nearest-neighbor, implying e.g. 
			\be \label{factorizes} d> \log_3(\xi) \implies \sum_{h \in G} p_{d} (g,h,k) = p_d(g) p_d(k).\ee 
			
			Let $d_\xi = \lceil \log_3(\xi) \rceil $ and define $l \equiv d-d_\xi$. Then
			\bea \label{restricted_sum} p_d(g,h) & = \sum_{\bfg,\bfh \in G^{3^{l}}} p_{d_\xi}(\bfg,\bfh) \d_{\maj^{l}(\bfg),g}\d_{\maj^{l}(\bfh),h}.\eea 
			Because of \eqref{factorizes}, $p_d(g,h)$ would factorize exactly if we were able to perform an unrestricted sum on the element $g_{3^l}$, which is located at the boundary between the bases of the two depth-$l$ $\maj$ trees leading to $g$ and $h$. The thing which prevents us from doing so is of course the factor of $\d_{\maj^l(\bfg),g}$. However, when $3^l \gg \xi$, we expect the outcome of the majority vote $\maj^l(\bfg)$ to be largely insensitive to the exact value of $g_{3^l}$: in the vast majority of cases, the value of the element $g_{3^l}$ will not by itself be able to sway the outcome of the majority vote. We thus claim that 
			\be \label{twosite_factorizes} p_{d_\xi + l}(g,h) = p_{d_\xi + l}(g)p_{d_\xi + l}(h) + \cdots,\ee 
			where $\cdots$ is exponentially small in $l$. 
			
			We can quantify this expectation by computing the conditional entropy $H(\maj^l(\bfg) | g_1,\dots,g_{3^l-1})$. This entropy vanishes if $\maj^l(\bfg)$ can be determined from $ g_1,\dots,g_{3^l-1}$ alone, without the need to know the value of the remaining element $g_{3^l}$. In general $H(\maj^l(\bfg) | g_1,\dots,g_{3^l-1})$ is nonzero, but is nevertheless exponentially small in $l$. Indeed, given a scheme for predicting $\maj^l(\bfg)$ from $g_1,\dots,g_{3^l-1}$, we can bound $H(\maj^l(\bfg) | g_1,\dots,g_{3^l-1})$ using Fano's inequality as 
			\be H(\maj^l(\bfg) | g_1,\dots,g_{3^l-1}) \leq H(p_e) + p_e \ln(|G|-1),\ee 
			where $p_e$ is the probability of our prediction scheme giving the wrong answer, and $H(x) = -x \ln x - (1-x) \ln (1-x)$ is the binary Shannon entropy. 
			
			Given $g_1,\dots,g_{3^l-1}$, we can exactly determine the values of all nodes in the ternary tree computing $\maj^l$ except for those on the tree's rightmost flank. For example, when $l=2$ we can determine all of the nodes of the following tree, except those marked in yellow: 
			\be \igpfc{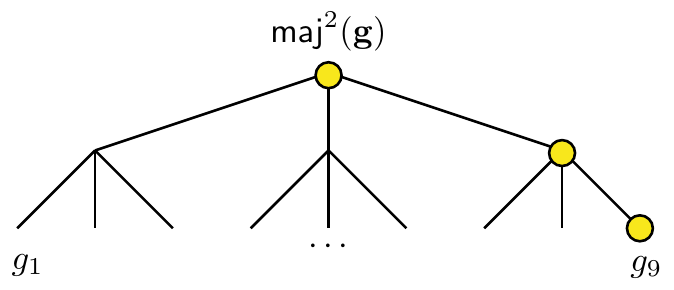} \ee 
			One estimate for $\maj^l(\bfg)$ given $g_1,\dots,g_{3^l-1}$ is to assume that the result of the majority vote on the last three elements $g_{3^l-2},g_{3^l-1},g_{3^l}$ is just the central element $g_{3^l-1}$ (since the central element wins in the case of ties), and to then assign the remaining nodes of the majority tree according to this choice. In order for this estimate to be in error, the choice of $g_{3^l}$ must be able to change the outcomes of every majority vote occuring on the right flank of the $\maj$ tree (the yellow nodes in the above figure). Since the probability of changing each of these majority votes is less than 1, the error probability $p_e$ will accordingly be exponentially small in $l$.
			
			More precisely, in order for this estimate to be incorrect, we first need that $g_{3^l-2} = g_{3^l} \neq g_{3^l-1}$, so that the value of the `deepest' majority vote on the right flank of the $\maj$ tree, viz.  $\maj(g_{3^l-2},g_{3^l-1},g_{3^l})$, disagrees with our guess of $g_{3^l-1}$. The probability of this disagreement happening is less than 1 unless $p_{d_\xi}(g,h,k) \propto \d_{g,k}$ for all $g,h,k$ --- but this is impossible, as summing over $h$ would produce $p_{d_\xi}(g)p_{d_\xi}(k) \propto \d_{g,k}$, which can be true only if $p_{d_\xi}(g)\neq 0$ only for a single $g$ (in which case $\maj$ has already converged). A similar argument applies for each node along the right flank of the majority tree, where the probability of an error occuring given a change in the rightmost member of a $\maj$ vote is strictly less than 1, unless $\maj$ has already converged. Thus $p_e$ is a product of $l$ numbers less than 1, and consequently is exponentially small in $l$ (unfortunately explicitly calculating $p_e$ in terms of $p_{d_\xi}(\bfg)$ is rather complicated). By Fano's inequality this implies that $H(\maj^l(\bfg) | g_1,\dots,g_{3^l-1}) $ is also exponentially small in $l$, meaning that performing an unrestricted sum over $g_{3^l}$ in \eqref{restricted_sum} can be done while incurring only an exponentially small error. This then implies \eqref{twosite_factorizes}. 
		\end{proof}

		\begin{figure}
			\includegraphics[width=.5\textwidth]{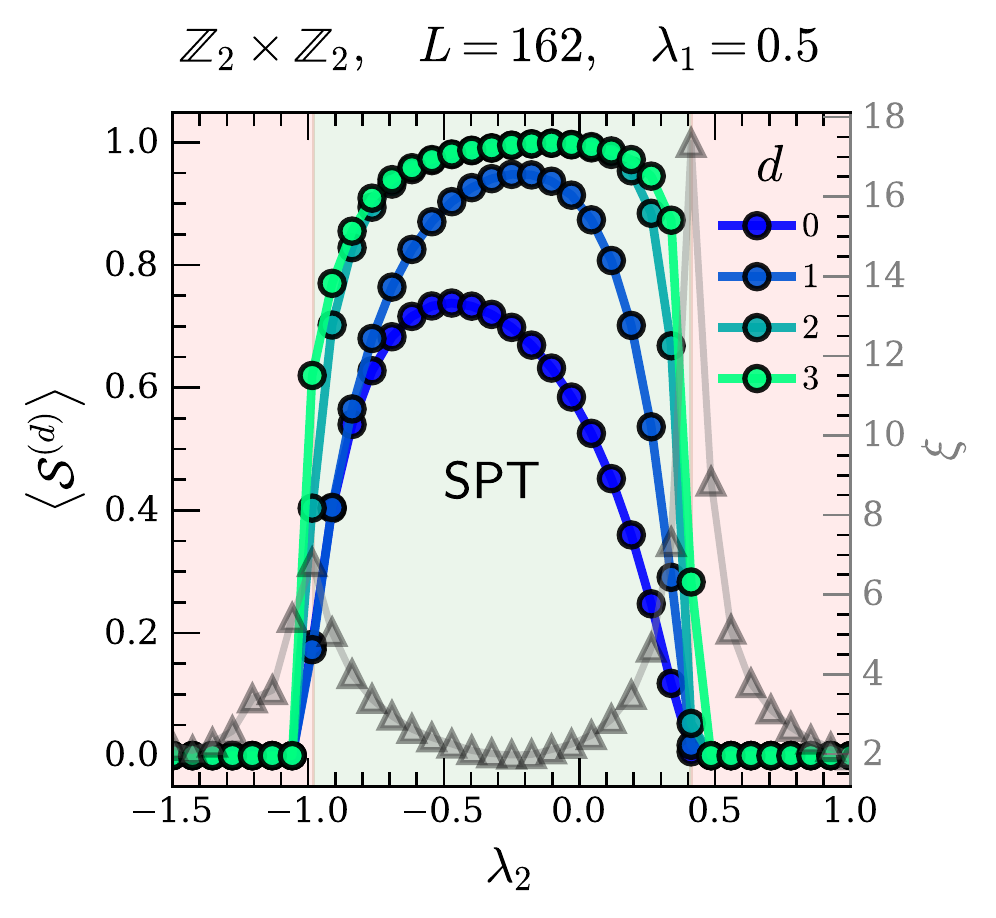}
			\caption{\label{fig:confused} String order parameters measured at different depths $d$ of an RG circuit acting on the ground states of a family of $\zt$ deformed cluster state Hamiltonians \eqref{znham}. Grey triangles denote the correlation length.} 
		\end{figure}

		\begin{figure}
			\includegraphics[width=.5\textwidth]{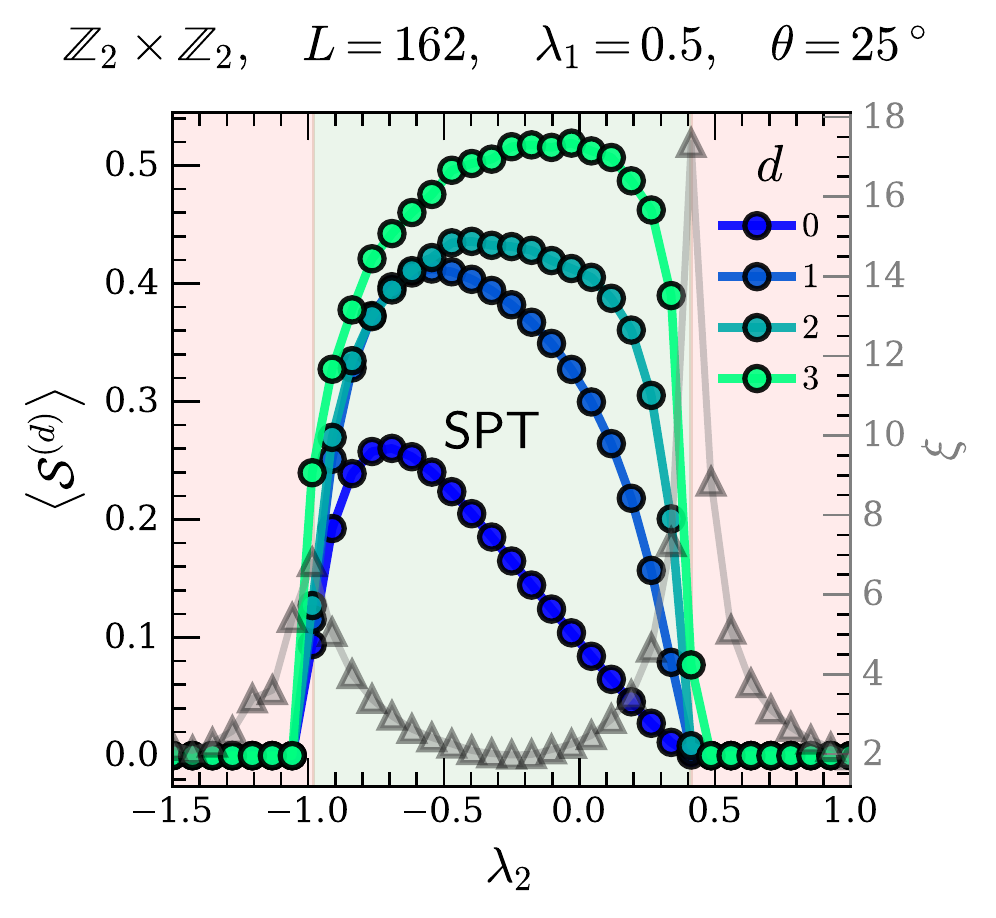}
			\caption{\label{fig:confused_2} As in Fig. \ref{fig:confused}, except with the wavefunction incurring an additional error (symmetry action) on each site with probability $\sin^2(25^\circ)$ (note the different scale on the $y$-axis). }
		\end{figure}
		
		\ss{Numerical tests of convergence}
		
		The numerical results in Figure~\ref{fig:dmrg} were obtained by running our RG circuit on the ground states obtained from the simple family of Hamiltonians given by \eqref{znham}. For this class of Hamiltonians, the initial distribution of errors $p_0(\bfg)$ can be determined straightforwardly within perturbation theory. The density of errors for these models is low except when one is close to the boundary of $\spt_\o$, and thus we expect (and indeed observe) our circuit to converge rather rapidly in most of parameter space. 
		
		One simple way of making the circuit's life more difficult (with the aim of performing a more stringent test of the convergence of our RG flow) is to act on a given input ground state with a depth-1 circuit that performs a symmetric rotation at each site, as was done in the main text near \eqref{confusing_circuit}. For the case of the $\zt$ version of the perturbed cluster state model \eqref{znham}, this can be done simply by acting with the unitary $\bot_i(\cos\t + Z_i \sin \t)$ (where in keeping with the previous appendicies, we are working in a basis where the symmetry is generated by $Z_j$ operators on each site). Fig. \ref{fig:confused} shows how our RG circuit performs on ground states drawn from a particular cut in the parameter space of the Hamiltonian \eqref{znham}, and Fig. \ref{fig:confused_2} shows how the performance is altered when acting on the same ground states with the above unitary. In the latter case the signals from the string operators are suppressed, but the pattern of convergence with increasing circuit depth is similar. 
		
		To test the flow induced by the majority voting process more explicitly, we can examine the flow of error probability distributions $p_0(\bfg)$ drawn from various ensembles not directly tied to the ground states of any simple family of Hamiltonians. Continuing to specify to the case of $\zt^2$ SPT phases, as one example we can consider the error distribution associated with the consistant depth symmetric circuit  
		\be \mcu_\t = \bot_i \prod_{a=1}^{k} \exp\( i \t_a \bot_{j=i}^{i+k} Z_j \),\ee 
		where the $\t_a$ are iid random variables drawn from some distribution on $[0,2\pi)$. This yields the error distribution 
		\be\label{circuiterror} p_e(\bfg) = \underset{\{ \t_a\}}{\EE}\[ \lan +| \mcu_\t^\da \bot_{i=1}^L \(\frac{1+ (-1)^{g_i}X_i}2 \)\mcu_\t | +\ran \],   \ee 
		where $\k{+} = (|0\ran  +|1 \ran )^{\tp L}$. 
		As another example, we can take the error probabilities to be generated by sampling from a random MPS (thereby encorporating a small amount of spatial correlation into the error distribution); in this case 
		\be p_e(\bfg) = \lan \phi_\c | \Pi_\bfg |\phi_\c\ran,\ee 
		where $\k{\phi_\c}$ is a translation-invariant MPS wavefunction constructed from a random MPS tensor of fixed bond dimension $\c$.  
		
		\begin{figure}
			\includegraphics{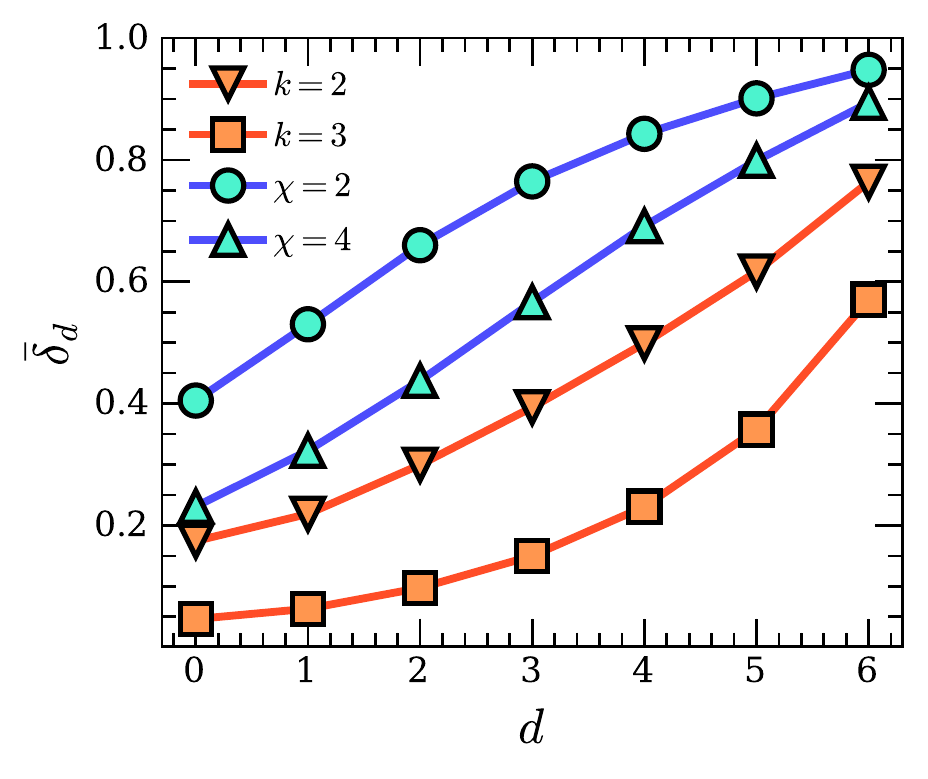}
			\caption{\label{fig:dmat_flow} Convergence of the flow induced by our RG circuit for two different random error models. Here $\ob \d_d$, the average gap between the most likely and second most likely single-site marginals, is plotted as a function of the depth $d$; the flow converges when $\ob \d_d \ra 1$. In the blue curves, the errors are sampled from random MPS wavefunctions with bond dimensions $\c=2,4$, while in the blue curves they are sampled using the error model of \eqref{circuiterror} with $k=2,3$.} 
		\end{figure}
		
		 In Fig. \ref{fig:dmat_flow} we numerically illustrate the convergence of the majority voting process under these two error models for errors applied to a system of size $L = 3^7 = 2187$. Let $\bar \d_d$ be the average gap at layer $d$ between the two largest values of the single-site marginals $p_d(g)$, as in \eqref{deltadef_app}. 
		We consider ranges $k=2,3$ and bond dimensions $\c = 2,4$, obtaining the evolution of $\ob \d_d$ shown in the figure.
		The error model derived from the circuit $\mcu_\t$ with $k=3$ displays the smallest expected gap $\ob\d_d$, but still shows a clear trend towards convergence with increasing $d$. 
		
		\ss{No false positives}
		
		We now show that $\mcq_\o$ does not produce false positives, i.e. that for $\nu\neq \o$, $\mcq_\o$ does not mistakenly converge to $\k{\Psi_\o}$ if given $\k{\psi_\nu}$ as input. That this holds can in fact be argued on general grounds using the results of \cite{de2021symmetry}, wherein it was shown that any quantum channel whose Kraus operators commute with the symmetry action cannot map $\k{\psi_\o}\in \spt_\o$ to $\k{\psi_\nu}\in \spt_\nu$ if $\o\neq \nu$. 	
		$\mcq_\o$ enjoys precisely this symmetry-property: from \eqref{mcc_maj} we see that  $(\unit \tp R_g \tp \unit) \mcc = \mcc R_g^{\tp 3}$,  
		implying that the symmetry pushes through $\mcq_\o$ to operate on just the renormalized legs:
		\be  \label{sym_pushthrough} \mcq_\o[(R_g^\da)^{\tp 3^l} \r\, R_g^{\tp 3^l}] = (R_g^\da)^{\tp 3^{l-1}} \mcq_\o(\r) R_g^{\tp 3^{l-1}}\ee 
		for any density matrix $\r$.
		This then implies that exactly the desired condition is satisfied by the Kraus operators of $\mcq_\o$, so that the claim follows.

		\section{Identifying spontaneous symmetry breaking with RG circuits  \label{app:symm_breaking}} 
		
		In this appendix we briefly discuss how to construct an RG circuit which recognizes phases with and without spontaneous symmetry breaking of an internal finite Abelian symmetry $G$. As with the SPT case, the RG circuit has the merit of being able to reduce the sample complexity of phase recognition near the critical point, and provides a direct link between error correction and RG flow. 
		
		The construction parallels that of the architecture used for recognizing SPT phases, but as there exists a standard local order parameter for diagnosing symmetry breaking the construction is much simpler (see \cite{furuya2021renormalization,friedman2022locality}). As in the rest of the paper, the setting will be on a 1d chain of length $L$.  
		
		Define the linear symmetry and shift operators as 
		\be R_g \equiv \sum_{h\in G} \c_g(h) \proj{h},\qq S_g \equiv \sum_{h\in G} \kb{g+h}{h},\ee 
		which obey $R_g S_h = \c_g(h) S_h R_g$. We will take the global symmetry to be generated by the operators $S_g^{\tp L}$ \footnote{In our discussion of SPT phases we worked in a basis where the symmetry generators were diagonal, since in that case we were only interested in symmetric phases. Here we are focused on phases with SSB, and taking the symmetry generators to be off-diagonal is more natural.}. By character orthogonality, any operator can be expanded in terms of the $R_g$ and $S_h$; the proof is analogous to that of \eqref{opdecomp}. While we will not specify a particular Hamiltonian, for concreteness one could simply take a $G$-generalization of the standard Ising model, e.g.
		\be H = - \sum_i \sum_{g\in G} \( R_{-g,i} \tp R_{g,i+1} + \l S_{g,i}\)\ee 
		where the real constant $\l$ tunes between the symmetric and SSB phases. 
		
		Denote the symmetric phase by $\para$, and let $\ssb$ denote the phase where $G$ is spontaneously broken. We take $\ssb$ to be a direct sum over {\it all} superslection sectors, including states where $R_g$ has all nonzero expectation values. Our goal will be to construct a circuit $\mcq_\o$ for classifying whether a given wavefunction is in $\para$, or in $\ssb$. 
		The reason for defining $\ssb$ as a union over all superselection sectors is that we will not require $\mcq_\o$ to correctly reproduce the superselection sector a given $\k\psi \in \ssb$ belongs to. This is because as in our analysis of SPT phases, we define phases by equivalence classes of wavefunctions under symmetric local constant-depth circuits, which can permute between superselection sectors (consider the depth-1 symmetric local circuit $S_g^{\tp L}$, for example). 
		
		Consider a wavefunction $\k\psi$ drawn from $\ssb$. $\k\psi$ can be obtained from the product state $\k\bfzero \equiv \k{0}^{\tp L}$ by a constant-depth symmetric circuit built from combinations of $S_{h,i}$ and $R_{h,i} \tp R_{-h,j}$ operators, with bounded $|i-j|$. Since the $R_h$ operators act as c-numbers on $\k{\bfzero}$, they can effectively be ignored, and we may write 
		\be \k{\psi} = \sum_\bfg C_\bfg \k{\bfg}.\ee 
		where the complex coefficients $C_\bfg$ satisfy $\sum_\bfg |C_\bfg|^2=1$ and produce a wavefunction with a finite correlation length. 
		
		\begin{figure}
			\includegraphics{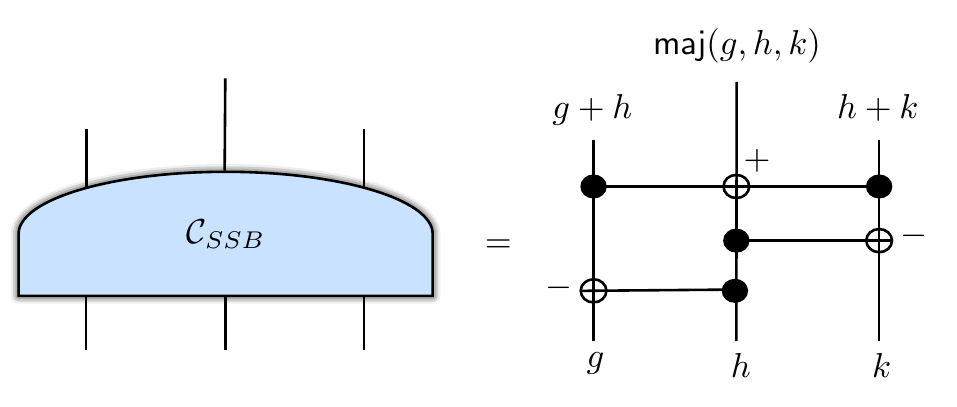}
			\caption{\label{fig:ising} Coarse-graining layer used for performing RG to identify the symmetry breaking phase of a theory with an Abelian symmetry group $G$. $g,h,k \in G$, and the $\pm$ signs denote the type of controlled shift gate applied ($CS_\pm$ or $CCS_\pm$; in the case $G = \zt$ these are just CNOT and CCNOT). The left (right) measurement leg is left in a state which keeps track of the presence of a domain wall between the first and second (second and third) input legs, and the renormalized leg becomes the majority of the inputs. }
		\end{figure}
		
		We now construct a circuit which performs RG on $\k\psi$ simply by implementing majority vote on the strings $\bfg$. The convolution module $\mcc_{SSB}$ that we use for this purpose is built out of controlled shift gates $CS_\pm, CCS_\pm$, defined as the $G$-generalizations of CNOT and Toffoli gates: 
		\bea CS_\pm & \equiv \sum_{g,h\in G} \proj{g} \tp \kb{h\pm g}h , \\
		CCS_\pm & \equiv \sum_{g,h\in G} \proj{g} \tp \proj{g} \tp \kb{h\pm g}h. 
		\eea  
		
		Majority voting on $G$ can be acommplished using the choice of $\mcc_{SSB}$ shown in figure \ref{fig:ising}, which consists of a pair of $CS_-$ gates controlled by the central `renormalized' leg, and a single $CCS_+$ gate controlled by the two measurement legs. It is straightforward to check that 
		\be \mcc_{SSB} (\k g \tp \k h \tp \k k) = \k{g-h} \tp \k{\maj(g,h,k)}  \tp \k{k-h}.\ee 
		The values of the measurement legs thus keep track of domain walls between the left / right and central tensor factors, while the value of the renormalized leg is determined by majority vote. Constructing a circuit $\mcq_\o$ from $\mcc_{SSB}$ in the same way as in the main text, we see that $\mcq_\o$ will converge to a reference product state $\k{g}^{\tp L}$ for some $g\in G$ as long as the recursive majority vote converges, which it will in the absence of fine tuning, as argued previously. Note that the value of $g$ which it converges to may not correpond to the superselection sector in which $\k\psi$ lies, as made clear by the Gerrymandering example of \eqref{gerry}. This however is enough for our purposes. 
		
		Finally, consider the case when $\k\psi$ is drawn from $\para$. $\k\psi$ can then be constructed by applying operators of the form $R_{g,i} \tp R_{-g,j}$ to the reference state $\k{+}^{\tp L}$, where $\k{+} = \frac1{\sqrt{|G|}} \sum_{g\in G} \k g$. We claim that the RG circuit will never incorrectly classify $\k\psi$ as being in $\ssb$, and that the distribution of group elements measured will remain uniformly random at all circuit depths. As in the SPT case, this follows from a symmetry argument. It is easy to check that $\mcc_{SSB}$ commutes with the symmetry group, in that 
		\be \mcc_{SSB} (S_g \tp S_g \tp S_g) = (\unit \tp S_g \tp \unit ) \mcc_{SSB}.\ee 
		Consider then the expectation value of $R_g$ on an arbitrary site at some depth $d$ of the circuit. We have 
		\bea \Tr[R_g \mcq_\o^d(\proj{\psi})] & = \c^*_g(h) \Tr[(S_h^\da)^{\tp L/3^d} R_g S_h^{\tp L/3^d} \mcq_\o^d(\proj{\psi})] \\ 
		& = \c^*_g(h) \Tr[R_g \mcq_\o^d((S_h^\da)^{\tp L} \proj \psi S_h^{\tp L} )] \\ & = \c^*_g(h)  \Tr[R_g \mcq_\o^d(\proj{\psi})] \eea 
		since $\k\psi\in \para$ is symmetric. Since this holds for any $g,h$, the LHS must vanish for all $d$.

		\section{Error correction with measurements and classical post-processing} \label{app:measurements} 
		
		We have seen that the RG circuit essentially performs a classical error correction procedure on an input wavefunction $\k{\psi_\o}\in\spt_\o$. One is naturally then led to ask: to what extent is the full RG circuit architecture actually needed for phase identification? As demonstrated above, the whole implementation of the RG circuit circuit can be effectively replaced by the direct measurement of the multiscale string order parameters, at least as far as phase identification is concerned. However, the multiscale string operator is rather complicated, and contains a very large number of terms. Is there an easier way out, which lets a collection of simple local measurements bypass the work done by the RG circuit? In particular, given that the job of the circuit is to perform a classical error correction of the local symmetry actions, can it simply be replaced by a simpler one-shot measurement of the local errors, together with classical data post-processing? In this appendix we argue that the answer to this question is likely no. 
		
		For the purposes of this discussion it will help to introduce the {\it disentangling operators}
		\be \scd_\o \equiv \sum_{\bfg \in G^{2L}}\Tr\[ \prod_{i=1}^L (A_\o^{g_{2i}})^* (A_\o^{g_{2i+1}})^T\] \bot_{j=1}^L \kb{g_{2j}}{g_{2j+1}} =  \igpfoc{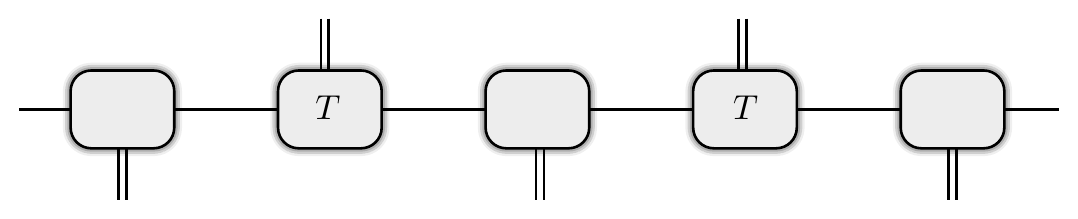}.\ee
		The disentangler $\scd_\o$ gets its name from the fact that it converts $\k{\Psi_\o}$ to the product state $\k{\Psi_0} = \k{0}^{\tp L}$ and vice-versa:
		\be \scd_\o \k{\Psi_\o} = \k{\Psi_0},\ee 
		as can easily be checked using the perfect tensor condition \eqref{graphical_inj}. 

		The utility of $\scd_\o$ is that given $\k{\psi_\o}\in \spt_\o$, it lets us easily sample from the distribution of errors acting on $\k{\psi_\o}$. 
		As before, let the error amplitudes of $\k{\psi_\o}$ be $C_\bfg$. Using \eqref{inj} and \eqref{traced_vg_app}, we obtain  
		\be \scd_\o \k{\psi_\o} = \sum_{\bfg \in G^L} C_\bfg \bot_{i=1}^L \k{\G_\o(g_{i+1}-g_i)}.\ee
		Applying $\scd_\o$ and then measuring in the $R_g$ basis thus lets us directly sample the probability distribution of errors. Indeed, consider the distribution $p^\o_e(\bfg)$, defined by 
		\be p^\o_e(\bfg)  = \lan \psi_\o| \scd_\o^\da \Pi_\bfg \scd_\o \k{\psi_\o},\ee
		with the projector 
		\be \Pi_\bfg \equiv \bot_{i=1}^L  \int_{h\in G} \c^*_h(g_i) R_{h,i} .\ee  
		As argued above, at a generic point in $\spt_\o$, the recursive majority function $\maj$ applied to samples obtained from the distribution $p_e^\o(\bfg)$ will converge to a fixed point where $p_e^\o(g)$ has support on only a single group element $g_*$. Furthermore, the convergence scale obtained in this way will agree with the convergence scale of the actual RG circuit. Thus measurements and classical data processing suffice to determine how `far' $\k{\psi_\o}$ is from $\k{\Psi_\o}$, given the promise that $\k{\psi_\o}\in \spt_\o$.  
		
		However, if one is not promised that $\k{\psi_\o} \in \spt_\o$, the full architecture of the RG circuit cannot simply be replaced by the depth-1 disentangling circuit $\scd_\o$, together with the classical data processing required to compute $\maj^d(p^\o_e(\bfg))$. To see this, define the distributions 
		\be p_e^{\o,\nu}(\bfg) = \lan \psi_\nu| \scd_\o^\da \Pi_\bfg \scd_\o \k{\psi_\nu},\ee 
		where $\psi_\nu \in \spt_\nu$ (with $\nu=0$ being the trivial paramagnetic phase). In order for the phase identification ability of the circuit to be replaced by the action of $\scd_\o$ together with classical data processing, we require that if $\nu\neq\o$, then $\maj^d$ does {\it not} converge when acting on samples obtained from $p_e^{\o,\nu}(\bfg)$. That is, for $\o \neq \nu$, we require that the marginal distribution $p_e^{\o,\nu}(g)$ be uniformly random on $G$. This is however not generically true, as we now argue. 
		
		Consider first the case where the input wavefunction is the product state $\k{\Psi_0}$ (we will see in a moment that this case is special). Since $\scd_\o \k{\Psi_0} = \k{\Psi_\o}$, the marginal probability for a single-site error is 
		\be p_e^{\o,0}(g) = \lan \Psi_0| \scd_\o^\da \Pi_g \scd_\o |\Psi_0\ran = \int_h \chi^*_h(g) \lan \Psi_\o |  R_h | \Psi_\o\ran = \wt f(g), \ee
		where $\wt f(g)$ is the Fourier transform of $f(g) = \lan \Psi_\o | R_g | \Psi_\o\ran$.  
		For $\maj^d(p_e^{\o,0}(\bfg))$ to generically fail to converge, $p_e^{\o,0}(g)$ should be uniformly distributed on $G$. This is true only if $f(g) \propto \d_{g,e}$, i.e. only if $\lan \Psi_\o |R_g| \Psi_\o\ran \propto \d_{g,e}$. In the present case this is indeed true, as can easily be checked using the injectivity of the canonical MPS tensor $A_\o$ and the fact that $\Tr[V_g] = \propto \d_{g,e}$. 
		
		However, while $p^{\o,0}_e(g)$ is uniformly random when $\k{\psi_0} = \k{\Psi_0}$ is an exact product state, we claim that $p^{\o,0}_e(g)$ is {\it not} generically be uniform throughout the paramagnetic phase $\spt_0$. Indeed, since $R_h$ commutes with the symmetry action, there is no selection rule which generically enforces $\lan \psi_0| \scd_\o^\da R_g \scd_\o |\psi_0 \ran \propto \d_{g,e}$, and the expectation value will generically be nonzero. 
		
		More precisely, consider a generic wavefunction $\k{\psi_0} \in \spt_0$. Since $R_{g,i}$ acts trivially on $\k{\psi_0}$ for all $i,g$, $\k{\psi_0}$ may be obtained from the trivial product state $\k{\Psi_0}$ through a circuit of the form 
		\be \k{\psi_0} = \sum_{i>j} \sum_{g\in G} D_{i-j}(g) S^L_{g,i} S^R_{g,j} \k{\Psi_0},\ee 
		where we have without loss of generality chosen a basis of symmetric operators to be generated by the linear symmetry generators (which act trivially on $\k{\Psi_0}$) and the shift operators $S^{L/R}_g$ for the cohomology class $\o$. When we act on this wavefunction with $\scd_\o$, we obtain 
		\be \scd_\o \k{\psi_0} = \sum_\bfg D'_\bfg R_\bfg \k{\Psi_\o},\ee 
		where the coefficients $D'_\bfg$ are generically rather complicated functions of the $D_{i-j}(g)$ \footnote{In the case where the coefficients $D_{i-j}(g)$ are all 2-local, being nonzero only when $i = j+1$, we have simply $D'_\bfg = \prod_i D_1(g_i)$. }.
		The point here is that $\scd_\o \k{\psi_0}$ is simply some generic wavefunction $\k{\psi'_\o} \in \spt_\o$. Such a generic wavefunction will {\it not} be such that $\lan \psi'_\o|R_g|\psi'_\o\ran$ vanishes for all $g\neq e$, since there is no symmetry reason for this to be so (that $\lan \psi'_\o|R_g|\psi'_\o\ran \neq 0$ for all $g\neq e$ generically is also confirmed in our numerics).

		\section{Different $G$-representations  }\label{app:blocking} 
		
		In this paper we have focused on the case in which $G$ is represented in the regular representation (a direct sum of each $G$ irrep, with each irrep entering with multiplicity one) on a Hilbert space $\mch_{reg}$ of dimension $|G|$: 
		\be R_g = \sum_{h\in G} \c_g(h) \proj h.\ee 
		Of course, a given physical system might be in some other representation $r_g$ of $G$ other than the regular representation, with the local Hilbert space $\mch_{r}$ of dimension $\dim \mch_{r} \neq |G|$. A prominent example of such a situation is the Haldane chain. Here a $G = \zt^2$ symmetry is represented on a spin-1 Hilbert space $\dim \mch_{r^{Hald}} = 3$ in a way which does not include the trivial representation: 
		\be r^{Hald}_g = \bpm \c_{(1,0)}(g) & & \\ & \c_{(0,1)}(g) & \\ & & \c_{(1,1)}(g)\epm, \ee 
		with the elements $(1,0)$ and $(0,1)$ corresponding to $\pi$ rotations about the $\uvx$ and $\uvy$ axes, respectively. In general, the symmetry will be represented as 
		\be \label{general_rg} r_g = \sum_{h\in G : N_h > 0} \sum_{\a = 1}^{N_h} \c_h(g) \proj{h,\a},\ee 
		where the degeneracies $\{N_h\}$ are non-negative integers. One choice of (unnormalized) representative wavefunctions $\k{\Psi_\o^r}$ for an SPT in which $G$ acts as \eqref{general_rg} is\footnote{We will continue to specify to maximally non-commutative $\o$ throughout this discussion.} 
		\be \k{\Psi^r_\o} = \sum_{\{g_i\} \in G^L : N_{g_i}  >0\, \forall \, i}\Tr[A^{g_1}_\o \cdots A^{g_L}_\o] \bot_{i=1}^L \sum_{\a_i = 1,\dots ,N_{g_i}} |g_i,\a_i\ran,  \ee
		which is uniform in degeneracy space and which otherwise is obtained simply by restricting the regular representation reference state $\k{\Psi_\o^R}$ to those group elements such that $N_g >0$ (in the case of the Haldane phase, $\k{\Psi_\o^{r^{Hald}}}$ is simply the AKLT state). It is straightforward to check that the symmetry fractionalizes on $\k{\Psi^r_\o}$ according to the cohomology class $\o$. It is also the case that $\k{\Psi^r_\o}$ is injective for all projectively faithful representations of $G$ (the ones we care about when doing quantum mechanics), as we prove below in proposition \ref{prop:injective_blocking}. Thus for each $r,\o$, $\k{\Psi^r_\o}$ indeed constitutes an allowed reference state. 
		
		From this starting point, and with the symmetry representation fixed to be $r_g$, one approach would be to design an RG circuit acting on the onsite Hilbert space $\mch_{r}$, which maps any input wavefunction in $\spt_\o$ to the reference state $\k{\Psi^r_\o}$. While such a construction is likely possible, it would likely require developing an architecture distinct from the architecture we developed for the regular representation, as the MPS tensors of $\k{\Psi^r_\o}$ will not generically be perfect tensors, and the set of $G$-symmetric local unitary operators will not generically be generated purely by the set of string operators and local symmetry actions. 
		
		Because of these complications, and because in this work we are primarily interested in illustrating a proof of principle, we will resort to a different approach. First, it is always possible to embed a model in which $G$ is represented as $r_g$ on a system with local Hilbert space $\mch_{r}$ into a model defined on the local Hilbert space $\mch_{reg}$, on which $G$ acts as the regular representation. This embedding is achieved via a local symmetric isometry $I = \bot_i I_i$, $I_i^\da I_i = \unit_i$, where the isometry at each site maps 
		\be I_i : \mch_{r} \ra \mch_{reg}^{\tp (n_r+1)},\qq n_r = \lceil \log_{|G|} (\max_{g\in G} N_g ) \rceil.\ee 
		Note that for the case where all of the degeneracies $N_g = 0,1$, we have $n_r=0$ and we can simply embed $\mch_{r}$ into $\mch_{reg}$ through the obvious inclusion. The requirement that $I_i$ be symmetric means 
		\be\label{sym_iso} r_{g,i} I_i = I_i R_{g,i}^{\tp (n_r+1)}.\ee 	
		
		Intuitively, the action of a local symmetric isometry such as $I$ preserves SPT phases, as the notion of an SPT phase should be robust with respect to adding unentangled ancillae and coupling them to the system through symmetric local unitaries. Therefore given a wavefunction $\spt_\o \ni \k{\psi_\o^r} \in \mch_{r}^{\tp L}$, the embedded wavefunction $I\k{\psi_\o^r}\in  \mch_{reg}^{\tp L(n_r+1)}$ is also in $\spt_\o$. This can be proved by looking at the projective representations obeyed by split symmetry operators (which are left untouched by the isometry by virtue of \eqref{sym_iso}), or by thinking about edge states: if the $|G|$ different edge states cannot be lifted by a symmetric local perturbation in $\mch_{r}^{\tp L}$, they are also unable to be lifted by a similar perturbation in the larger Hilbert space $\mch_{reg}^{\tp L(n_r+1)}$. 
		
		Since $G$ is represented as the regular representation in the larger Hilbert space, we can then apply the already-developed RG circuit for this case to flow the embedded wavefunction $I \k{\psi_\o^r}$ to the canonical representative $\k{\Psi_\o}$ of the regular representation. String operators can then be measured in $\k{\Psi_\o}$ before using $I^\da$ to project back to the original Hilbert space, and the expectation values of these string operators serve as a diagnostic of whether or not the original wavefunction $\k{\psi_\o^r}$ is indeed in $\spt_\o$. Thus in the present setting, the correct definition of the multiscale string operators is 
		\be \wt\mcs^{(\o)}_g = I^\da \mcq_\o^\da (\mcs^{(\o)}_g \tp \unit_{an}) \mcq_\o I. \ee 
		Note however that in this case, we have only accomplished phase recognition by the direct measurement of the multiscale string operators, and have not actually provided a unitary RG circuit which asymptotically maps any $\k{\psi_\o} \in \spt_\o$ to $\k{\Psi_\o} \tp \k{\phi_{an}}$ (conjugating $\mcq_\o$ with $I$ does not generically work, since $I^\da \mcq_\o I$ needn't be unitary). 
		In actual hardware implementations measuring the multiscale string operators may be impractical, and ideally one would still be able to construct a unitary circuit where each individual gate --- rather than the entire circuit itself --- is able to be implemented on the particular microscopic Hilbert space one is provded with. We leave the best way of getting around this issue as a question for the future. 
		
		\ms 
		
		We close this section with a proof of the proposition referenced above:
		\begin{prop}\label{prop:injective_blocking} 
			One may always block sites together to form composite sites on which $G$ is represented by 
			\be r^{block}_g = \bigoplus_{h\in G} \c_h(g)^{\oplus N_h},\ee 
			where the degeneracies $N_h \geq 1$ for all $h\in G$.
		\end{prop}
		For the maximally non-commutative factor sets we are concerned with in this appendix, the set of matrices $V_g$ generate the full matrix algebra of the virtual space, and thus the above proposition is equivalent to injectivity of the reference states $\k{\Psi^r_\o}$ given above. 
		
		\begin{proof}
			
			Suppose that microscopically $G$ acts on each site as 
			\be r_g = \bigoplus_{\g \in S}\c_\g(g),\ee 
			where the $\g$ run over some subset $S$ of $G$. Since in quantum mechanics we mod out by overall phases, any element $g$ for which $r_g \propto \unit$ is represented trivially. We will not allow this to occur, since in this case $G$ is not represented faithfully. Therefore we require that the $\{\g\}$ be such that 
			\be  \label{alphaclaim} \c_\g(g) =1 \, \, \forall \, \, \g  \, \implies g = e.\ee 
			We will call a representation obeying the above equation a `legit' representation. 
			As an example, for $G = \zn$ a legit representation is 
			\be r_g =1 \oplus \z_N^g\ee 
			where $\z_N= e^{\twp i /N}$. 
			For $\zn^2$, one requires at least three summands, with a legit representation being e.g. 
			\be r_\bfg = 1 \oplus \z_N^{g_1} \oplus \z_N^{g_2}.\ee 
			Note how we are always taking the first summand to be the trivial representation; this is done wolog because we can always redefine $r_g$ by an overall $g$-dependent phase. This is what allows us to use 1 on the RHS of \eqref{alphaclaim}. 
			
			Let $H$ be the subgroup of $G$ generated by the $\{\g\}$. We will first show that if $r_g$ is a legit representation, then $H=G$. Indeed, suppose that $H$ is a proper subgroup of $G$. We claim that this implies that $r_g$ is not a legit representation, i.e. that there is some $e \neq g_*\in G$ for which $r_{g_*} = \unit$. To show this we use the fact that for a subgroup $H\subset G$, every irrep of $H$ can be extended to an irrep of $G$ in $|G|/|H|$ different ways, so that in particular, there are $|G|/|H|$ irreps of $G$ whose characters are trivial when restricted to $H$ (as can be proved using e.g. Frobenius reciprocity). 
			Thus if $H$ is a proper subgroup of $G$, we can find $|G|/|H|>1$ different irreps of $G$ whose characters are trivial on $H$. Since each irrep of $G$ is identified with a group element, there exist $|G|/|H|$ choices of $g_*\in G$ such that $\c_\g(g_*)=1$ for all $\g$. Therefore $r_g$ does not generate a legit representation, as claimed.
			
			Finally, suppose $r_g$ is a legit representation. Taking $n$ tensor powers of $r_g$ gives 
			\be r_g^{\tp n} = \bigoplus_{\{\g_1,\dots,\g_n\}} \c_{\g_1 + \cdots + \g_n}(g).\ee 
			Since the $\{\g\}$ must generate $G$, for large enough $n$ (upper bounded by $|G|$), $r_g^{\tp n}$ will contain $\c_h(g)$ for all $h\in G$. Thus as long as $n$ is large enough, all of the degeneracies $N_h$ can be made positive, as we wanted to show. 
		\end{proof}

		\section{Sample complexity} \label{app:sample_complexity} 
		
		In this section we describe how to estimate the sample complexity of phase recognition using the expectation values of the multiscale string operators (MSOs) $\mcs_g^{(\o)}$ defined in the main text. Given an SPT ground state wavefunction $\k{\psi}$, the sample complexity $M_{\d}^\psi$ is defined as the number of measurements of the MSOs in the state $\k{\psi}$ needed to conclude whether or not $\k{\psi}$ is in the target phase $\spt_\o$ of interest, given some tolerance for error $\d$. 
		
		The eigenvalues of any $\mcs_g^{(\o)}$ are given by the phases $\c_g(h), h\in G$. This fact can be proven by demonstrating it in the case where $\mcs_g^{(\o)} = S^R_g \tp R_g \tp S^L_g$ is a `bare' 3-site string operator (which can in turn be shown using the graphical technology developed in previous appendices), and noting that the eigenvalues of $\mcs_g^{(\o)}$ do not change under conjugation by $\mcq_\o$. Since $G$ is a finite group, for any $g$ the character $\c_g(h)$ will be a multiple of some primitive $n$-th root of unity for all $h$, for some integer $n$ (with $n>1$ if $g\neq e$). We can thus select out the eigenspace of $\mcs^{(\o)}_g$ spanned by those eigenvectors with eigenvalue $\z_n^k$, $k\in \zz_n$ (recall $\z_n \equiv e^{\twp i /n}$) by forming the projector 
		\be \Pi_k^{(\o,g)} \equiv \frac1n \sum_{l\in \zz_n} \z_n^{kl} (\mcs^{(\o)}_g)^l.\ee 
		If $\k{\psi} \in \spt_\o$ is an RG fixed point, then $\k{\psi}$ is an eigenstate of $\mcs^{(\o)}_g$ with eigenvalue 1, and has support only on the image of $\Pi_e^{(\o,g)}$. When $\k{\psi}\in \spt_\o$ but $\k{\psi}$ is not a fixed point, $\lan \psi | \mcs^{(\o)}_g |\psi\ran$ will be some $O(1)$ number which approaches $1$ exponentially fast in the QCNN circuit depth, but which for small circuit depths and near phase boundaries can still be rather small, giving nonzero weight to the $\Pi_{k\neq e}^{(\o,g)}$. On the other hand, if $\k{\psi} \not\in \spt_\o$, then by the selection rules discussed earlier, $\lan \psi | \mcs^{(\o)}_{g;i\ra j} | \psi\ran$ will vanish as $e^{-|i-j|/\xi}$, meaning that the measurements of the $\Pi_{k}^{(\o,g)}$ will be very nearly uniformly distributed over all $k$. When doing numerics, this means that as soon as one sees a $\lan \psi | \mcs^{(\o)}_{g;i\ra j} | \psi\ran$ whose value is above machine precision and not exponentially small in the limit where $|i-j|/\xi \gg 1$, one can be confident that $\k{\psi} \in \spt_\o$. However in an experimental context there will always be some finite level of noise present in the measurement process, mandating that a stronger signal be observed before one can be confident about phase identification. 
		
		The criterion we will adopt is as follows. Define the binary random variable $\xi_\psi \in \{0,1\}$ to be the outcome of measuring the projector $\Pi^{(\o,g_\bullet)}_e$ in the state $|\psi\ran$, 
		where $g_\bullet$ is some fixed element of maximal order $n$ in $G$. Let $p \equiv \EE[\xi_\psi]$. Assuming $\k\psi\in\spt_\o$, the sample complexity $M^\psi_\d$ is defined as the number of samples needed so that $p$ is separated from $1/n$ (which would be the expectation value of $\xi$ if $\k\psi\not\in\spt_\o$) by $\d$ times the standard deviation of the emperical average of $\xi_\psi$ (in practice, one may be able to do slightly better by using a more sophisticated estimator of the average \cite{cong2019quantum}). That is, $M^\psi_\d$ is the minimum integer satisfying 
		\be p - \d \frac{\sqrt{p(1-p)}}{\sqrt{M^\psi_\d}} > 1/n.\ee 
		Thus 
		\be M^\psi_\d = \lceil \frac{\d^2 p(1-p) }{(p-1/n)^2}\rceil.\ee 
		In the setting of the $\zz_3\times \zz_3$ SPT considered in Fig. \ref{fig:dmrg} of the main text, we take $g_\bullet = (1,0)$, so that $n=3$ and 
		\be p = \frac13(1 + 2S),\qq S \equiv {\rm Re}[\lan \psi | \mcs^{(\o)}_{g_\bullet} | \psi \ran],\ee 
		giving
		\be M_\d^\psi =  \lceil \d^2 \frac{(1+2S)(1-S)}{2S^2} \rceil.\ee 
		Note that $M^\psi_\d =1$ as $S \ra 1$ from below; in this case one is nearly at the fixed point $\k{\Psi_\o}$, which can be recognized immediately. In the opposite limit of $S \ra 0$ the string operator contains no signal of the SPT order, and $M^\psi_1$ accordingly diverges. The inset of Fig. \ref{fig:dmrg} shows $M^\psi_3$ for $\k\psi$ the DMRG ground state close to the phase boundary, illustrating the exponential improvement of the sample complexity with increasing depth $d$.

		\section{The case when $\omega$ is not MNC} \label{app:non_mnc} 
		
		In this appendix we explain how the technology used in the construction of our RG circuit can be generalized to the case when the cohomology class $\o$ in question is not maximally non-commutative (MNC). This allows us to give a provable performance guarantee for identifying {\it any} SPT phase protected by an internal finite Abelian symmetry. 
		
		A summary of the material contained in this appendix is as follows. We begin in Sec. \ref{ss:entangledpairs} by discussing an alternate way of formulating the tools (reference states, shift operators, etc.) used to construct the RG circuits $\mcq_\o$ in the MNC case, which will allow us to more easily generalize to non-MNC $\o$. In Sec. \ref{ss:generalentangledpairs} we extend this technology to general $\o$. The main result of this section is to prove that for any SPT phase, the protecting symmetry $G$ can always be factored into a `MNC part' and a `flavor' part. The Hilbert space can similarly be factored as $\mch = \mch_{MNC} \tp \mch_{flav}$, with $G$ acting projectively on $\mch_{MNC}$ and linearly on $\mch_{flav}$. The degrees of freeom in $\mch_{MNC}$ are the ones responsible for the phase's protected edge modes, while those in $\mch_{flav}$ are in some sense trivial (in our reference states, the degrees of freedom in $\mch_{flav}$ are frozen out in a product state and unentangled with the rest of the system). 
		
		In Sec. \ref{ss:generalwfs} we discuss how general wavefunctions $\k{\psi_\o}$ can be represented in terms of collections of `errors' applied to the reference states $\k{\Psi_\o}$ constructed in Sec. \ref{ss:generalentangledpairs}. Some errors involve local actions of the symmetry that act only on $\mch_{MNC}$. These errors can be corrected using the same framework developed above for the case of MNC $\o$. Other errors act only on $\mch_{flav}$; these errors are unimportant and can essentially simply be ignored. The main difficulty is that there exist errors which involve a {\it mixed} action on $\mch_{MNC}$ and $\mch_{flav}$. From the perspective of either one of these subsystems, these mixed errors look as though they break the symmetry; that they are in fact symmetric is only seen when both subsystems are considered. To deal with these errors, the strategy we adopt is to convert them into errors that act {\it only} on $\mch_{MNC}$ and to then run an MNC RG circuit on $\mch_{MNC}$, with $\mch_{flav}$ being ignored for the remaineder of phase recognition procedure. 
		
		The conversion process is done by measuring the degrees of freedom on $\mch_{flav}$, performing an $O(L)$ amount of classical data processing on the measurement outcome, and then using the result of this data processing to construct a certain depth-1 local symmetric circuit $\mcu$; we will show that $\mcu$ can always be chosen to turn mixed errors into errors that act only on $\mch_{MNC}$. After acting with $\mcu$, the wavefunction is left in a form amenable to be adressed with an MNC RG circuit applied to $\mch_{MNC}$, and phase recognition proceeds as in the MNC case developed previously. The upshot is that phases with non-MNC $\o$ can be dealt with using more or less the same framework as the MNC case, provided that one pays an $O(L)$ amount of classical data processing as overhead.

		\ss{Entangled pairs representation of SPT ground states for MNC $\o$} \label{ss:entangledpairs}
		
		Before getting into the construction for general $\o$, we first discuss an alternate representation of the symmetry action and reference fixed-point wavefunctions in the MNC case, which will be helpful to use when we extend our technology to general $\o$. This representation can be thought of as a generalization of the usual representation of the AKLT ground state wavefunction, where spin-1 degrees of freedom are decomposed into two spin 1/2 sites that are then entangled into Bell pairs between sites. It is based on the fact \cite{verstraete2005renormalization,perez2006matrix} that any SPT ground state wavefunction with zero correlation length can be constructed by applying local isometries to a collection of Bell pairs. In this representation this construction is made manifest; for this reason we refer to it as the `entangled pairs representation'. 
		
		To get started, recall that in all of the previous discussion, we have worked in a basis in which the local symmetry action $R_g$ is diagonal and acts in the regular representation, viz. 
		\be \label{regrep} R_g = \sum_{h\in G} \c_g(h) \proj h.\ee 
		There is however another convenient basis in which the symmetry fractionalization that occurs in SPT phases is made more manifest.
		
		\begin{prop}[entangled pair representation, MNC case] \label{entangledpairrep}
			When $\o$ is MNC, the regular representation $R_g$ is unitarily equivalent to the representation
			\be \label{mcrdef} \mcr_g  = V_g^* \tp V_g,\ee 
			meaning that there exists a local unitary $U$ such that $U^\da R_g U = \mcr_g$ (as before, $V_g$ denotes the projective representation associated with $\o$). 
		\end{prop} 
		
		Recall that $\o$ being MNC implies that $G = G'\times G'$ is a square and $\dim(V_g) = |G'|$; thus the dimensions in \eqref{mcrdef} match. 
		
		\begin{proof}
			First, note that $\mcr_g$ is linear over $G$, viz. that $\mcr_g \mcr_h = \mcr_{g+h}$, since while $V_g V_h \neq V_{g+h}$, $V^*_g V^*_h \tp V_g V_h = V_{g+h}^* \tp V_{g+h}$. 
			
			We thus only need demonstrate that $R_g$ and $\mcr_g$ have identical eigenvalues. The eigenvalues of $R_g$ are of course enumerated by the characters $\c_g(h), h\in G$. That these agree with the eigenvalues of $V_g^*\tp V_g$ can be seen by explicitly constructing the corresponding eigenvectors, which are 
			\be \k{v_h} \equiv (V_h^* \tp \unit) \k{\G}, \qq \k{\G} \equiv \frac1{\sqrt{|G|}} \sum_{g\in G'} \k{g} \tp \k{g}. \ee 
			The eigenvalue of $\k{v_h}$ under $V_g^* \tp V_g$ is accordingly 
			\be ( V_g ^*\tp V_g) \k{v_h} = (V_g^* V_h^* \tp V_g) \k{\G} = \l_\o(h,g) (V_h^* V_g^* \tp V_g) \k{\G} = \c_h(\G_\o(g)) (V_h^* \tp V_g V_g^\da) \k{\G} = \c_h(\G_\o(g)) \k{v_h}.\ee 
			Since $\G_\o$ is an automorphism of $G$, the eigenvalues of $\mcr_g$ are indeed enumerated by the characters $\c_g(h),h\in G$, as claimed. 
		\end{proof}
		
		In this representation, the left and right shift operators $S^{L/R}_g$ are espeically simple. Indeed, from the fact that the shift operators `fractionalize' the symmetry action into left and right projective parts, it is not surprising that under the unitary transformation in the above proposition, they become 
		\be S^L_g = V^*_g \tp \unit,\qq S^R_g = \unit \tp V_g,\ee 
		so that $\mcr_g = S^L_g S^R_g$, as desired. It is easy to check that the $S^{L/R}_g$ obey all the algebraic relations derived previously. 
		
		\sss{MPS representation of reference states $\k{\Psi_\o}$} 
		The reference state $\k{\Psi_\o}$ also takes a rather simple form in this basis. With the decomposition of the physical space $\mch$ into the two tensor factors as above, the MPS tensors of $\k{\Psi_\o}$ simply associate the left tensor factor with the left virtual leg, and the right tensor factor with the right virtual leg: 
		\be\label{mnamps} \k{\Psi_\o} = \sum_{\bfg \in G^L} \Tr[A^{g_1} \cdots A^{g_L} ], \qq A^g = \kb{g_1}{g_2}.\ee 
		Graphically, 
		\be \igptc{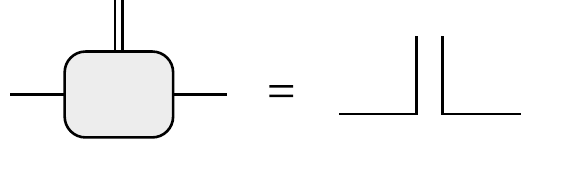} \ee 
		Chaining these tensors together, it is apparent that in $\k{\Psi_\o}$, the right $\tp$ factor of the physical Hilbert space at site $i$ forms a maximally entangled pair with the left $\tp$ factor at site $i+1$, generalizing the usual picture of the AKLT state to groups beyond $\zt$. 
		
		\sss{Parent Hamiltonian}
		
		The parent Hamiltonian of the reference state $\k{\Psi_\o}$ has the same representation as before, viz. 
		\be \label{hepair} H = \sum_i H_i,\qq H_i = -\sum_{g\in G} S^R_{g,i} S^L_{g,i+1}.\ee 
		In the present representation, this is 
		\be H_i = - \sum_{g\in G} (\unit \tp V_g)_{i} \tp (V_g^* \tp \unit)_{i+1},\ee
		with each $H_i$ easily being seen to stabilize $\k{\Psi_\o}$. Graphically, the action of e.g. $H_2$ on $\k{\Psi_\o}$ looks like 
		\be \label{graphical_pairham} \igpfc{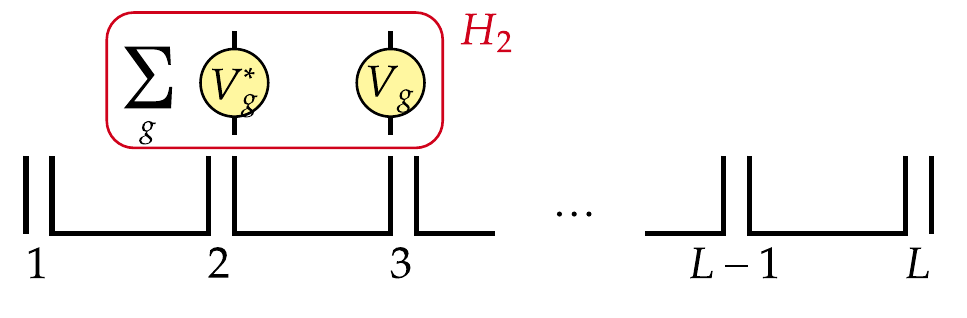}\ee 
		The existence of $\dim V_g = \sqrt{|G|}$ protected edge modes is particularly transparent in this representation, since $H$ as given above contains no nontrivial action on the left tensor factor of the leftmost site (straight line on far left of \eqref{graphical_pairham}), nor on the right tensor factor of the rightmost site (straight line on far right). Thus on a finite chain, $H$ has $\sqrt{|G|}$ degenerate modes on each edge, corresponding to the $\sqrt{|G|}$ independent states of each free tensor factor. It is easy to check that these degenerate modes cannot be lifted by any symmetry-preserving perturbation. Indeed, the symmetry action restricted to the edges is the projective action $V_g$, and $\{ V_g\}$ generate the full matrix algebra on $L(\cc^{\sqrt{|G|}})$ --- thus there are no nontrivial terms that can be added to the Hamiltonian which commute with the symmetry and act locally on the free tensor factors.

		\ss{Entangled pairs representation for general $\o$} \label{ss:generalentangledpairs}
		
		We now generalize the formalism of the previous subsection to the case where $\o$ is not MNC. We will show that for general $\o$, all of the tools involved (projective representations, reference-state MPS wavefunctions, etc.) can in some sense be split into an MNC part and an extra `flavor' part, with the latter not playing an essential role in distinguishing one SPT phase from another. This decomposition is inspired by previous studies of non-MNC SPT phases \cite{stephen2017computational,stephen2017computational2,prakash2014ground,de2021symmetry} and in particular by the masters thesis of David Stephen \cite{stephen2017computational}, in which a general prescription was worked out for constructing zero correlation length MPS tensors for non-MNC SPT ground states. 
		
		\sss{Mathematical preliminaries} 
		
		First, recall from App. \ref{app:math} that for general $\o$, the {\it projective center} $C_\o$ is defined as the subgroup of all elements $g$ for which $V_g$ is a scalar matrix. This means that $\o$ is MNC when restricted to the quotient $\mcg_\o \equiv G / C_\o$ (which is always non-trivial; if $C_\o = G$ then there are no nontrivial projective representations, and hence no nontrivial SPT phases). For future reference, we will write the relation between $\mcg, G,$ and $C_\o$ as the short exact sequence 
		\be \label{ogses} 0 \ra C_\o \xra{\iota} G \xra{\pi} \mcg_\o \ra 0.\ee 
		In what follows we will drop the subscripts on $\mcg_\o, C_\o$ for simplicity of notation. 
		
		Recall from the math facts listed in App. \ref{app:math} that MNC cohomology classes exist for a given finite Abelian group iff that group takes the form $G'\times G'$ for some factor $G'$. Since $\o$ is MNC on $\ob \mcg$, we thus always have 
		\be \mcg = G' \times G'\ee 
		for some $G'$. Since the dimension of every projective irrep with factor set $\o$ is $\sqrt{|\mcg|} = |G'|$, the projective irreps will always satisfy $\dim(V) = |G'|$ (readers interested in seeing an explicit example can skip ahead to Sec. \ref{sec:znex}). Because the number of topologically protected edge modes is determined by $\dim(V)$ (see e.g. \ref{pollmann2010entanglement} and the following), their number is thus equal to $|G'|$. While this is the same number of edge modes as a $\mcg$-SPT with MNC factor set, the $G$-SPT under consideration is nevertheless generically physically distinct from a $\mcg$-SPT with an `accidentally' enlarged symmetry, as we will see explicitly later. 
		
		Let us now fix some notation. In situations where confusion could potentially arise, we will use subscripts to explicitly denote the group in which operations are being performed, with e.g. $+_K$ denoting addition in the group $K$ (we use `$+$' for the group action since all groups under consideration are Abelian). Furthermore, we will adopt notation in which bars denote membership in $\mcg$, while tildes denote membership in $C$.
		We thus will write a general $g\in G$ in terms of the pair $(\ub g, \ut g) \in \mcg \times C$ as 
		\be \label{oggdecomp} g = \pi\inv(\ub g) +_G  \iota(\ut g),\ee 
		where the map $\pi\inv$ will always be defined using a certain fixed section of $G$.\footnote{Physicists which find this notation strange can think of $\pi$ as an isometry: $\pi (\pi \inv) = \unit_C$, but $(\pi\inv) \pi\neq \unit_G$.}  
		We see that $\ub g$ and $\ut g$ are obtained from $g$ as 
		\be \label{gdecomp} \ub g = \pi(g),\qq \ut g = \iota\inv(g -_G \pi\inv(\pi(g))),\ee 
		where $\ut g$ is well-defined since $g - \pi\inv(\pi(g)) \in \ker(\pi) = {\rm im}(\iota)$.

		Note that when writing elements pairwise like this, the group addition law acts in a twisted way: 
		\be \label{twistedgroup} (\ub g , \ut g) +_G (\ub h , \ut h ) = (\ub g +_{\mcg} \ub h, \,\,\ut h +_{C} \ut g +_{C} \ut \d(\ub h, \ub g)),\ee 
		where we have defined the function
		\be \ut\d : \mcg \times \mcg \ra C,\quad \ut\d(\ub h , \ub g) = \iota\inv(\d\pi\inv(\ub h, \ub g)),\ee 
		with 
		\be \d \pi\inv(\ub h, \ub g  ) = \pi\inv(\ub h) +_G \pi\inv(\ub g) -_G \pi\inv(\ub h+_{\mcg} \ub g),\ee 
		with $\ut \d$ well-defined for the same reason as $\ut g$ in \eqref{gdecomp}. Since $\d\pi\inv$ is a coboundary, its exterior derivative vanishes, with 
		\be \label{ext_der} \d\pi\inv(a,b) + \d\pi\inv(a+b,c) = \d\pi\inv(a,b+c) + \d\pi\inv(b,c) = \pi\inv(a) + \pi\inv(b) + \pi\inv(c) - \pi\inv(a+b+c).\ee

		We will also need a way of talking about relations between the characters of $G,\mcg,$ and $C$, which we will denote respectively as $\c, \c^\mcg$, and $\c^C$. 
		To find the relations between the characters, we apply the Hom functor to the short exact sequence \eqref{ogses}. Letting $K^* \equiv {\rm Hom}(K,U(1))$ denote the group of characters on $K$, the contravariance of Hom gives the left exact sequence 
		\be 1 \ra \mcg^* \xra{\pi_*} G^* \xra{\iota_*} C^*,\ee 
		where $\pi_*,\iota_*$ satisfy $\ker(\iota_*) = {\rm im}(\pi_*)$ and are defined by pre-composing with $\pi,\iota$, respectively (e.g. $\pi_*(\c^\mcg)(\cdot) = \c^\mcg(\pi(\cdot)) : G \ra U(1)$). 
		
		Since $K^*\cong K$ for all finite Abelian groups $K$, just as we could write any $g\in G$ as the pair $(\ub g, \ut g)\in \mcg \times C$, we can perform a similar decomposition on characters in $G^*$, with any $\c_*\in G^*$ equivalently being given in terms of a pair $(\c_{\ob g}^\mcg, \c_{\wt g}^C )\in \mcg^* \times C^*$ --- here the group elements $(\bar g, \wt g) \in \mcg \times C$ label their respective characters, and we have placed the bars and tildes on top to distinguish these group elements from $\ub g, \ut g$. 
		
		This is formalized as the following proposition, which is essentially the dual of \eqref{oggdecomp}:
		\begin{prop}[decomposition of characters]\label{prop:chidecomp}
			Any character $\c_g\in G^*$ can be written as 
			\be \c_g = \pi_*(\c^\mcg_{\ob g}) \iota_*\inv(\c^C_{\wt g}),\ee 
			where the map $\iota_*\inv$ is as usual defined by picking a section of $G^*$.\footnote{Note that it does not simply work by pre-composing by $\iota\inv$, since $\iota$ is not invertible except on $\ker(\pi)$.} Here the characters $\c^\mcg_{\ob g}\in\mch^*$ and $\c^C_{\wt g}\in C^*$ are determined from $\c_g$ as 
			\be \c^C_{\wt g} = \iota_*(\c_g),\qq \c^\mcg_{\ob g} = \pi_*\inv\( \frac{\c_g}{\iota\inv_*(\iota_*(\c_g))}\),\ee 
			where the expression in parenthesis is invertible by $\pi_*$ since it is in $\ker(\iota_*)$. 
			
		\end{prop}
		
		We will simplify notation by defining the map 
		\be \t \equiv \iota_*\inv\ee
		and by writing $G^*\ni \t(\c^C_{\wt g}) \equiv \c_{\t(\wt g)}$. Explicitly, when evaluated on $h\in G$, Prop. \ref{prop:chidecomp} reads 
		\be \label{explicitchar} \c_g(h) = \c^{\mcg}_{\bar g}(\ub h) \c_{\t(\wt g)}(h).\ee 
		Since characters can be identified with group elements, we can also decompose any $g\in G$ as 
		\be g = \pi_*(\bar g) +_G \t(\wt g),\ee 
		where by following logic similar to \eqref{gdecomp}, 
		\be \bar g = \pi_*\inv(g -_G \t(\iota_*(g))),\qq \wt g = \iota_*(g).\ee 
		In terms of the pairs $(\ob g, \wt g)$, group addition is performed in a way analogous to \eqref{twistedgroup}: 
		\be \label{dualtwistedgroup} (\bar g, \wt g) +_G (\bar h , \wt h) = (\bar g +_\mcg \bar h +_\mcg \bar \d(\wt g, \wt h), \wt g +_C \wt h),\ee 
		where 
		\be \bar \d : C \times C \ra \mcg , \qq \ob \d(\wt h,\wt g) =\pi_*\inv( \d \t(\wt h, \wt g)),\ee
		where we have defined the coboundary 
		\be \d\t : C \times C \ra G, \qq  \d \t(\wt h,\wt g) = \t(\wt h) + \t(\wt g) - \t(\wt h + \wt g),\ee  
		which also obeys \eqref{ext_der}. 

		\ms 
		
		{\bf Example: $G = \zz_{pq}^2$}
		
		\ms 
		
		To illustrate the above symbol pushing with an example, consider the case when $G = \zz_{pq}^2$ with $p,q$ prime, and take $\o=q$. Then 
		\be C = p\zz_{pq}^2 \cong \zz_q^2,\qq \mcg = \zz_{pq}^2 / p\zz_{pq}^2 \cong \zz_p^2,\ee 
		giving the SES 
		\be  \label{zpqses} 0 \ra \zz_q^2 \xra{\times p} \zz_{pq}^2 \xra{\mod p} \zz_p^2 \ra 0,\ee 
		so that $\iota$ multiplies by $p$ and $\pi$ takes its argument mod $p$,  with $\pi\inv$ defined by emedding the elements of $\mcg$ as those elements of $G$ which have both factors less than $p$. Letting $(a+_xb) \equiv (a+b) \mod x$, we may then write any $g \in G$ as 
		\be \label{gdec_exp} g = \ub g +_{pq} p \ut g, \ee 
		where $ \ub g \in \zz_p^2,\,\ut g \in \zz_q^2,$ and the function $\ut \d(\ub g,\ub h)$ is simply 
		\be \ut \d (\ub g, \ub h ) = \frac{(\ub g +_{pq} \ub h) - (\ub g+_p \ub h)}p.\ee 
		
		In this simple setting, it is straightforward to check that taking the dual merely exchanges $p$ and $q$. Thus the map $\pi_*=\times q$ is multiplication by $q$, while $\iota_*$ takes its argument modulo $q$. The inverse $\t=\iota_*\inv$ embedds the elements of $C$ as those elements of $G$ with both factors less than $q$. We may correspondingly also write any $g\in G$ as 
		\be g = q  \bar g +_{pq} \wt g,\ee 
		with 
		\be \bar \d(\wt g, \wt h) = \frac{(\wt g +_{pq} \wt h) - (\wt g +_q \wt h)}q. \ee

		\sss{Form of $\mcr_g$} 
		
		We now generalize the entangled pair representation $\mcr_g$ of the symmetry operators \eqref{mcrdef} to the case of general $\o$. 
		The idea is to split $\mcr_g$ into a projective action controlled by the part $\ub g$ of $G$ belonging to $\mcg$, together with an ordinary linear action controlled by the part $\wt g$ belonging to $C$. 
		To facilitate this decomposition, we will break up the $|G|$-dimensional physical Hilbert space $\mch$ at each site as 
		\be \mch = \mch_l \tp \mch_{flav} \tp \mch_{re},\ee 
		where $\dim \mch_l = \dim \mch_{re} = \sqrt{|\mcg|} = |G'|$ and $\dim\mch_{flav} = |C|$. We will refer to $\mch_l \tp \mch_{re}$ as the `MNC' subspace $\mch_{MNC}$; it is on this subspace that the symmetry fractionalization occurs. We will also use the symbols $\mch_{flav},\mch_{MNC}$ to refer to the tensor products of the single-site flavor / MNC Hilbert spaces across {\it all} sites; the distinction between the two should be clear from context. 
		
		In order to write down $\mcr_g$, which will split as a tensor product according to the above decomposition, we need a further definition: 

		\begin{definition}
			For any $g\in G$, the matrix $\wt R_g \in L(\cc^{|C|})$ is defined as the linear represetation $R_g$ restricted to $\t(C)$ (again with the same fixed section used to define $\t$). Explicitly, for $\{\k{\wt h}\, : \, \wt h \in C\}$ a basis of $\mch_{flav}$, 
			\be \wt R_g \equiv \sum_{\wt h \in C} \c_{\t(\wt h)}(g) \proj{\wt h}. \ee 
			
		\end{definition}
		
		We can now state the main result of this subsection, which shows that $R_g$ can be equivalently written in a form that explicitly separates the symmetry action into a projective action on $\mch_{MNC}$ and a linear action on $\mch_{flav}$: 
		
		\begin{theorem}[entangled pair representation, general case]
			For a factor set $\o$ with projective center $C$, the regular representation $R_g$ is unitarily equivalent to the representation
			\be \label{mcrdef2} \mcr_g  = V^*_{\ub g} \tp \wt R_g \tp V_{\ub g}  .\ee 
		\end{theorem} 
		
		\begin{proof}
			It is easy to check that $\mcr_g$ is linear, viz. that $\mcr_g \mcr_h = \mcr_{g+_Gh}$.		
			Thus as with the case of Prop. \ref{entangledpairrep}, it sufficies simply to show that the eigenvalues of $R_g$ are identical to the eigenvalues of $\mcr_g$. Since the eigenvalues of $V_{\ub g}^* \tp \unit_{|C|} \tp V_{\ub g}$ are given by $\c_{\ub g}^{\mcg}(\mcg)$ by virtue of $\o$ being MNC when restricted to $\mcg$, from the definition of $\wt R_g$ we have 
			\be \label{rgeigs} \mathsf{Eigvals}(\mcr_g) = \{ \c^{\mcg}_{\bar h}(\ub g) \c_{\t(\wt h)}(g) \,\, | \, \, (\bar h, \wt h) \in \mcg \times C\}.\ee  
			From the character decomposition formula \eqref{explicitchar}, this set is simply equivalent to the set $\c_g(G)$, which is indeed equal to the eigenvalues of the regular representation $R_g$. 
		\end{proof}

		\sss{MPS representation of reference states $\k{\Psi_\o}$} 
		
		We are finally in the position to apply the above group theory technology to the problem at hand. Our goal is to construct a wavefunction $\k{\Psi_\o}$ that will serve as our chosen RG fixed point for the phase $\spt_\o$. $\k{\Psi_\o}$ should be symmetric, and should also be stabilized by a collection of string operators.
		Furthermore, $\k{\Psi_\o}$ should reduce to the simple form of the generalized AKLT state \eqref{mnamps} in the MNC case, i.e. in the case when $C$ is trivial. We claim that the following wavefunction does the job: 
		\be\label{can_ref_state_general}\k{\Psi_\o} = \sum_{\bfg \in G^L} \Tr[ A^{g_1}_{\wt g_\bullet} \cdots A^{g_L}_{\wt g_\bullet} ],\ee 
		where the MPS tensors are (writing $g\in G$ as the triple $(\bar g_1,\wt g, \bar g_2) \in G' \times C \times G'$ in accordance with the factorization $\mch_l \tp \mch_{flav} \tp \mch_{re}$)
		\be  \label{nonmnamps}  A^g_{\wt g_\bullet} =  A^{(\bar g_1,\wt g,\bar g_2)}_{\wt g_\bullet} = \kb{\bar g_1}{\bar g_2} \d_{\wt g, \wt g_\bullet} ,\ee 
		where the virtual space has dimension $\dim \mch_{virt} = |G'|$ and where $\wt g_\bullet$ is an arbitrary fixed element of $C$. 
		$A^{(\bar g_1,\wt g,\bar g_2)}_{\wt g_\bullet}$ is simply given by the generalized AKLT state MPS tensor used in the MNC case \eqref{mnamps}, plus a projector onto a trivial product state for the degrees of freedom in $\mch_{flav}$. Graphically, 
		\be A^g_{\wt g_\bullet} = \igptc{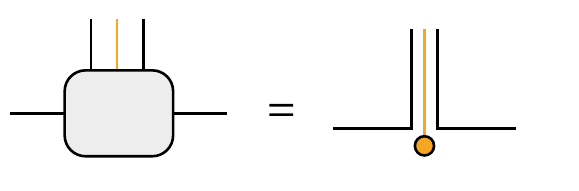},\ee 
		where the orange leg lives in $\mch_{flav}$ and the circle denotes a projection onto the state $\wt g_\bullet$. 
		
		$\k{\Psi_\o}$ has zero correlation length, is stabilized (modulo $\wt g_\bullet$-dependent phases) by all of the global symmetry operators $\mcr_g^{\tp L}$, and has the correct pattern of symmetry fractionalization; hence it defines a suitable reference state for the phase $\spt_\o$. It is clear from this presentation that the degrees of freedom in $\mch_{flav}$ are not essential to defining the topological properties of the SPT phase, and as such it is convenient to work with a reference state which trivializes them away. Relatedly, we see that on a finite chain, $\k{\Psi_\o}$ has $|G'|$ symmetry-protected modes localized to each edge; this shows that the number of protected edge modes is determined {\it only} by the dimension of the MNC subgroup $\mcg = G' \times G'$ (see also e.g. \cite{pollmann2010entanglement,de2021symmetry}).
		
		Physically, the purpose of string operators is that they allow one to move between distinct degenerate edge states. Since there are only $|G'|$ symmetry-protected modes, the space of operators which acts on the zero mode subspace has dimension $|G'|^2 = |\mcg|$. Thus we will only define $|\mcg|$ (and not $|G|$) different string operators, done in the obvious way by focusing only on the action of the symmetry on the $\mch_{MNC}$ factor of $\mch$: 
		
		\begin{definition}[string operators for general $\o$]
			For $\ub g\in \mcg$, the string operators $\mcs_{\ub g;i\ra j}$ are defined as 
			\be \mcs_{\ub g; i\ra j} \equiv S^R_{\ub g,i} \tp \bot_{l = i+1}^{j-1} \mcr_{\pi\inv(\ub g),l} \tp S^L_{\ub g, j},\ee 
			where the left and right shift operators are defined as 
			\be \label{shiftdef} S^L_{\ub g} \equiv V_{\ub g}^* \tp \unit_{flav} \tp \unit_r,\qq S^R_{\ub g} \equiv \unit_l \tp \unit_{flav} \tp V_{\ub g}.\ee 
		\end{definition}
		Note that we have chosen the linear symmetry operators $\mcr_g$ appearing in the interior of the string operator's support to be defined with a particular section of $G$ picked out by $\pi\inv(\mcg)$ --- when evaluated on states where the $\mch_{flav}$ degrees of freedom are in a product state (such as $|\Psi_\o\ran$), the exact choice of section only modifies the action of $\mcs_{\ub g;i\ra j}$ by an unimportant overall phase factor. It is easy to check that the choice of $\k{\Psi_\o}$ in \eqref{nonmnamps} gives a reference state stabilized (up to a phase) by all $|\mcg|$ string operators:
		\be \mcs_{\ub g; i\ra j} \k{\Psi_\o} = \c_{\wt g_\bullet}(\pi\inv(\ub g))^{|j-i|} \k{\Psi_\o} \, \, \forall \,\, \ub g \in \mcg, \, i<j.\ee 
		
		A parent Hamiltonian for the wavefunction $\k{\Psi_\o}$ (given a fixed choice of $\wt g_\bullet$) is as in \eqref{hepair}, but with an extra term to fix all of the degrees of freedom in $\mch_{flav}$ to equal $\k{\wt g_\bullet}$: 
		\be H = \sum_i H_i, \qq H_i = - \sum_{\ub g \in \mcg} S^R_{\ub g,i} S^L_{\ub g, i+1} - \unit_l \tp |\wt g_\bullet\ran_i \lan \wt g_\bullet|_i \tp \unit_r.\ee 
		
		When restricted to the MNC subspace, $\k{\Psi_\o}$ looks exactly like the references states constructed for MNC $\o$. Using the technology developed previously, we can therefore construct a quantum circuit $\bar \mcq_\o$ which acts only on $\mch_{MNC}$, and which corrects any symmetric errors whose actions are restricted to $\mch_{MNC}$. The circuit $\bar \mcq_\o$ nevertheless does not by itself provide a way to implement phase recognition, since --- as we will see shortly --- more complicated types of errors are possible. 		
		
		\sss{edge mode counting and accidental symmetries}
		
		Since the number of edge modes is determined only by $|\mcg|$, it is natural to ask if in fact a $G$-SPT with MNC subgroup $\mcg$ is always secretly equivalent to a fine-tuned version of a $\mcg$-SPT with MNC factor set and an `accidental' enlarged symmetry group $G$. It turns out that this is sometimes --- but not generically --- the case. 
		
		If a $G$-SPT is to be equivalent to a fine-tuned version of a $\mcg$-SPT, one should be able to break the symmetry down as $G \leadsto \mcg$, without reducing the number of protected edge modes. The only natural embedding of $\mcg$ into $G$ that exists is the one given by the map $\pi_*$; thus this pattern of symmetry breaking is acheived by preserving the subgroup $\pi_*(\mcg) \subset G$, and breaking the rest of $G$. In order to determine the number of edge modes that remain after symmetry breaking, we need simply to compute the projective center of $\pi_*(\mcg)$, defined as 
		\be C^{\pi_*(\mcg)}_\o = \{ \ub g \in \pi_*(\mcg) \, : \, \l_\o(g,h) = 1 \, \, \forall \, \, \ub h \in \pi_*(\mcg) \}.\ee 
		The number of protected edge modes (on a given edge) is then determined by 
		\be {\rm \# \, edge \, \, modes } = \sqrt{\frac{|\mcg|}{| C^{\pi_*(\mcg)}_\o|}}.\ee 
		Depending on the scenario in question, the number of protected edge modes could be either unchanged (in which case $G$ should indeed be viewed as an`accidental' symmetry) or eliminated entirely. We will see how both possibilities can be realized in the following example.  
		
		\sss{Example: $G=\zz_{pq}^2$} \label{sec:znex}
		
		We return to our example of $G=\zz_{pq}^2$, with $p,q$ prime and the non-MNC cohomology class $\o = q$.  As a reminder, in this case we have $C = p\zz_{pq}^2 \cong \zz_q^2,\mcg\cong \zz_p^2$. 
		
		As in our discussion of the MNC case for $G=\zz_N^2$, we will choose the projective representation $V_{\ub g}$ of $\mcg$ as $V_{\ub g} = X^{\ub g_1} Z^{\ub g_2}$, where $X,Z$ are the $\zz_p$ clock and shift matrices. Thus the explicit form of $\mcr_g$ is 
		\be \mcr_g =  V_{\ub g}^*  \tp \wt R_{g} \tp V_{\ub g}=  X^{\ub g_1} Z^{\ub g_2} \tp \(  \sum_{\wt h \in \zz_q^2}Z^{\wt h_1 (\ub g_1 + q \ut g_1) /q} \tp Z^{\wt h_2 (\ub g_2+ q \ut g_2)/q}\) \tp X^{\ub g_1} Z^{-\ub g_2},\ee
		where our choice of lift $\t : C \ra G$ is manifested in the factors of $1/q$ appearing in the exponents of the operators constituting $\wt R_g$.

		On a given edge, our $G$-SPT has only $\sqrt{|\mcg|} = p$ protected zero modes. Is this SPT then secretely a $\zz_p^2$ SPT with an accidentally enhanced symmetry? To answer this question, we consider breaking the symmetry down as $G \leadsto \pi_*(\mcg) = 	q \zz_{pq}^2$ ($i.e.$ the generators of $G$ are broken, but their $q$-fold powers are not). The projective center of $\pi_*(\mcg)$ is 
		\be C_q^{q\zz_{pq}^2} = \frac p{\gcd(p,q)} \zz^2_p = \zz^2_{\gcd(p,q)} =\begin{dcases} \zz_1 & q\neq p \\ \zz_p^2 & q=p \end{dcases} \ee 
		where in the last equality we used that $p,q$ are prime. Thus the number of edge modes is unchanged upon breaking $\zz_{pq}^2 \leadsto \zz_p^2$ if $p\neq q$ (and hence in this case the full $G$ symmetry should be viewed as being accidental), while the edge modes are completely eliminated if $q=p$ (so that in this case $G$ cannot be broken to any subgroup without completely trivializing the theory). The latter fact is especially clear when one realizes that when $q=p$, $\pi_*(\mcg) = p \zz_{p^2}^2 = C$. 

		\ss{Picture of general SPT ground states}\label{ss:generalwfs} 
		
		We now need the counterpart to theorem \ref{thm:simpleU}, which established that for MNC $\o$, any $\k{\psi_\o} \in \spt_\o$ can be written as a collection of linear symmetry actions (`errors') acting on the reference state $\k{\Psi_\o}$. In what follows we will set $\wt g_\bullet = e$ to be the identity for simplicity. 

		We begin by characterizing a basis for operators which is particularly well-suited for the task at hand. Our previous operator decomposition formula can be generalized as follows:
		
		\begin{prop}
			Any single-site operator $\mco$ can be decomposed as 
			\be 
			\mco = \sum_{g,h\in G} C_{g,h}\mcr_g S^R_{\bar h} \mcw_{\wt h} ,\ee 
			where $C_{g,h}$ are complex coefficients, and the operators $\mcw_{\wt h}$ are defined as 
			\be  \mcw_{\wt h} \equiv \sum_{\wt l \in C} \unit_l \tp \kb{\wt l + \wt h}{\wt l} \tp V_{\G\inv_\o(\bar \d(\wt l, \wt h))}.\ee
			
		\end{prop}
		
		\begin{proof}
			Consider working in the original basis, in which $R_g$ is diagonal and acts in the regular representation \eqref{regrep}. This representation is naturally associated with the basis $\{ \k g \, : \, g\in G\}$ for the local Hilbert space. Define the shift operators 
			\be X_g = \sum_{h \in G} \kb{h+g}h.\ee 
			It is then straightforward to check that the set $\{X_g R_h\, : \, g,h\in G\}$ generates all on-site operators. Indeed, this set clearly contains $|G|^2 = \dim(\mch)^2$ operators which are all linearly independent, as 
			\be \Tr[ X_g R_h R^\da_{h'} X_{g'}^\da ] = \Tr[X_{g-g'} R_{h-h'}]  = |G|\d_{g,g'} \d_{h,h'}.\ee 
			
			The $X_g$ and $R_h$ satisfy the commutation relation 
			\be R_g X_h = \c_g(h) X_h R_g.\ee 
			To find an operator basis better suited for the entangled pairs representation, one possibility is to simply transform the $X_h$ into the entangled pairs basis. We will do something that is nearly equivalent to this, by looking for operators that pick up the same phase $\c_g(h)$ upon commuting past $\mcr_g$. Define the operator 
			\be \mcx_h \equiv  \sum_{\wt l\in C}\unit_l \tp \kb{\wt l+\wt h}{\wt l} \tp V_{\G_\o\inv(\bar h)} V_{\G\inv_\o(\bar\d(\wt l,\wt h))}.\ee 
			Commuting this past $\mcr_g$, we find 
			\bea \label{xrcomm} \mcr_g \mcx_h & = \sum_{\wt l\in C} V_{\ub g}^* \tp \c_{\t(\wt l + \wt h)}(g) \kb{\wt l + \wt h}{\wt l} \tp V_{\ub g} V_{\G_\o\inv(\bar h)} V_{\G\inv_\o(\bar\d(\wt l,\wt h))} \\ 
			& = \sum_{\wt l\in C} V_{\ub g}^* \tp \c_{\t(\wt l) + \t(\wt h) - \d\t(\wt l,\wt h)}(g) \kb{\wt l + \wt h}{\wt l} \tp V_{\G_\o\inv(\bar h)} V_{\G\inv_\o(\bar\d(\wt l,\wt h))} V_{\ub g} \c^{\mcg}_{\bar g}(\ub h + \bar\d(\wt l,\wt h)) \\ 
			& = \c_{\bar g}^\mcg(\ub h) \c_{\t(\wt g)}(h) \mcx_h \mcr_g \\ 
			& = \c_g(h) \mcx_h \mcr_g.\eea 
			where the last line follows from our character decomposition formula \eqref{explicitchar}, and where in the third equality we used\footnote{For the example of $G = \zz_{pq}^2,\o =q$ used earlier, this string of equalities reads 
				\be \z_{pq}^{\d\t(\wt l,\wt h) \cdot g} = \z_{pq}^{\d\t(\wt l,\wt h)\cdot[g]_p} = \z_p^{(\d \t(\wt l,\wt h)/q) [g]_p},\ee 
				with $\ub g = [g]_p$ and $\d \t(\wt l, \wt h) \in q\zz_{pq}^2$.}
			\bea \c_{\d \t(\wt l,\wt h)}(g) & = \c_{\d \t(\wt l, \wt h)} (\pi\inv \pi (g)) = \c^\mcg_{\pi_*\inv(\d\t(\wt l, \wt h))}(\ub g) = \c^\mcg_{\bar \d(\wt l, \wt h)}(\ub g).\eea
			This suggests that $\{ \mcr_g\mcx_h \, : \, g,h\in G\}$ forms a complete basis of operators. We check this by calculating the trace $\Tr[\mcr_g \mcx_h \mcx_{h'}^\da \mcr_{g'}^\da]$. To do so we first write 
			\bea \mcx_h \mcx_{h'}^\da & = \sum_{\wt l\in C} \unit \tp \kb{\wt l + \wt h}{\wt l + \wt h'} \tp V_{\G_\o\inv(\bar h)} V_{\G\inv_\o(\bar\d(\wt l,\wt h))}V_{\G\inv_\o(\bar\d(\wt l,\wt h'))}^\da V_{\G_\o\inv(\bar h')}^\da \\ 
			& =  \sum_{\wt l \in C} \l_{h,h',\wt l}  \unit \tp \kb{\wt l + \wt h}{\wt l + \wt h'} \tp V_{\G\inv_\o(\bar h - \bar h' + \bar \d(\wt l, \wt h) - \d(\wt l, \wt h'))}, \eea 
			where we have defined the rather ugly-looking function 
			\be \l_{h,h',\wt l} \equiv \vs_{\G\inv_\o(\bar h' + \bar \d(\wt l,\wt h'))} \o(\G\inv_\o(\bar h + \bar \d(\wt l,\wt h)),-\G_\o\inv(\bar h'+\bar \d(\wt l,\wt h'))) \o(\G_\o\inv(\bar h),\G_\o\inv(\bar \d(\wt l, \wt h)) \o^*(\G_\o\inv(\bar h'),\G_\o\inv(\bar \d(\wt l, \wt h'))),\ee
			which satisfies $\l_{h,h,\wt l} = 1$.  
			Thus 
			\bea \Tr[\mcr_g \mcx_h \mcx_{h'}^\da \mcr_{g'}^\da] & =\d_{\wt h,\wt h'} \sum_{\wt l \in C}  \l_{h,h',\wt l} \c_{\t(\wt l)}(g-g')\Tr[V^*_{\ub  g} V^T_{\ub g'} \tp V_{\ub g} V^\da_{\ub g'} V_{\G\inv_\o(\bar h - \bar h')}] \\ 
			& = \d_{\wt h, \wt h'} \d_{\ub g,\ub g'} \sqrt{|\mcg|} \sum_{\wt l \in C}  \l_{h,h',\wt l}  \c_{\t (\wt l)}(\iota(\ut g) - \iota (\ut g')) \Tr[V_{\G\inv_\o(\bar h - \bar h')}]\\ 
			& = \d_{ \wt h, \wt h'} \d_{\bar h, \bar h'}\d_{\ub g, \ub g'} |\mcg| \sum_{\wt l \in C} \c^C_{\wt l}(\ut g - \ut g') \\ 
			& = \d_{h,h'} \d_{g,g'}|G|,  \eea 
			as desired. 
			
			We now simply realize that $\mcx_h$ can be broken up as 
			\be \label{mcx_breakup} \mcx_h = S^R_{\G_\o\inv(\ob h)} \mcw_{\wt h},\ee 
			The algebra $L^2(\mch)$ of onsite linear operators is then generated as 
			\be L^2(\mch) = \lan \mcr_g S^R_{\bar h} \mcw_{\wt h}  \, \, : \, \, g,h \in G \ran, \ee
			as claimed.  
		\end{proof}

		For our analysis of SPT ground states, we further need to know how {\it symmetric} operators can be decomposed. The following proposition follows easily from the commutation relation \eqref{xrcomm} and character orthogonality: 
		
		\begin{prop}
			Any multi-site symmetric operator $\mco_{sym}$ can be decmposed as  
			\be \label{mcosym2} \mco_{sym} = \sum_{\bfg,\bfh \in G^L} C_{\bfg,\bfh}  \d\(\sum_i h_i\) \bot_{i=1}^L  \mcr_{g_i}\mcx_{h_i}, \ee 
			where the $C_{\bfg,\bfh}$ are complex coefficients. 
		\end{prop}
		
		\begin{proof}
			The proof is straightforward, and simply involves making use of the commutation relation \eqref{xrcomm}: for a general operator $\mco$, we have 
			\be \mco \mcr_g^{\tp L} = \mcr_g^{\tp L}  \sum_{\bfg,\bfh \in G^L} \c_g(\sum_i h_i)C_{\bfg,\bfh}  \bot_{i=1}^L  \mcr_{g_i}\mcx_{h_i}.\ee 
			Since all of the terms in this linear combination are linearly independent, character orthgoonality means that $\mco$ is symmetric iff $\sum_i h_i = e$. 
		\end{proof}
		
		We can now state the analogue of theorem \ref{thm:simpleU} for the case of general $\o$: 
		
		\begin{theorem}[relation between SPT states in the same phase, general $\o$] \label{thm:simpleUgeneral} 
			Let $\k{\psi_\o} \in \spt_\o$ be a ground state in the SPT phase identified with the cohomology class $\o$. Then for $\k{\Psi_\o}$ the canonical reference state defined in \eqref{can_ref_state_general}, we may always write
			\be\label{psicircgeneral} \k{\psi_\o} = \sum_{\bfg,\bfh\in G^L} C_{\bfg,\bfh} \mcr_\bfg S^R_{\bar \bfh} \mcw_{\wt \bfh}  \k{\Psi_\o},\ee 
			where the $C_{\bfg,\bfh}$ are complex coefficients satisfying 
			\be \sum_{\bfg,\bfh} |C_{\bfg,\bfh}|^2 = 1,\ee 
			and the $\bfh$ are restricted so that 
			\be \label{hconds} \sum_i \wt h_i = e,\qq \sum_i \bar h_i  = - \sum_i \rem(\wt\bfh)_i , \ee 
			with the first equality holding in $C$ and the second holding in $\mcg$, and where we have defined 
			\be \rem(\wt \bfh)_i \equiv \bar \d\( \sum_{j<i} \wt h_j,\wt h_i\)  \ee 
			to be the remainder in $\mcg$ accumulated when adding the components of $\wt \bfh$ up to the position $i$. 
		\end{theorem}
		
		\begin{proof}
			The proof is a straightforward consequence of the decomposition \eqref{mcosym2}, with the conditions in \eqref{hconds} just coming from the fact that the sum over $\bfh$ in \eqref{mcosym2} is restricted to those vectors such that $\sum_i h_i = e$, together with the addition rule \eqref{dualtwistedgroup}. 
			
		\end{proof}

		Note that because we have chosen $\wt g_\bullet = e$, the $\mcw_{\wt \bfh}$ operators in \eqref{psicircgeneral} only act as shift matrices on $\mch_{flav}$ (that is, they act trivially on $\mch_{MNC}$), since their action on $\mch_{re}$ is controlled by $\bar \d(\wt g_\bullet ,\wt h_i)$, which vanishes for all $\wt h_i$ if $\wt g_\bullet = e$. 
		
		\ms

		As written in \eqref{psicircgeneral}, $\k{\psi_\o}$ is still not in a form that is amenable to be analyzed with the EC machinery developed in previous sections. This is because unlike in the case where $\o$ is MNC, here the errors acting on $\k{\Psi_\o}$ in \eqref{can_ref_state_general} do {\it not} consist solely of onsite symmetry operators. The problem here is ultimately the fault of the $\mcw_{\wt \bfh}$ operators, which entangle the flavor degrees of freedom on $\mch_{flav}$ with those on $\mch_{MNC}$. They arise becuase of the addition rule \eqref{dualtwistedgroup}, which says that shifting the flavor degrees of freedom by amounts $\wt h_i$ with $\sum_i \wt h_i $ trivial in $C$ is actually {\it not} in general a symmetry --- rather, one has a symmetry only if the sum $\sum_i \t(\wt h_i) \in \pi_*(\mcg)$ is trivial in $G$. When $\sum_i \t(\wt h_i) \neq e$, symmetry is restored only when one performs a compensating action on $\mch_{MNC}$, which is reflected in the constraint on $\bar \bfh$ in \eqref{hconds}. 
		
		This means that when we look only at $\mch_{MNC}$, the errors in \eqref{can_ref_state_general} do {\it not} appear to be symmetric, since $\sum \bar h_i$ is {\it not} in general equal to the identity. Thus directly acting on $\mch_{MNC}$ with the RG circuit $\bar \mcq_\o$ --- whose performance is predicated on the assumption that only symmetric errors are present --- will yield garbage. In order to use $\bar \mcq_\o$ to perform phase recognition, we need to convert the errors appearing in \eqref{psicircgeneral} into errors that {\it are} symmetric when restricted to $\mch_{MNC}$. Doing this is the task addressed in the following subsection.

		\ss{Removing the mixed errors}
		
		Our goal in this subsection is to construct a procedure by which the errors in \eqref{psicircgeneral} can be converted into a set of errors which are symmetric when restricted to $\mch_{MNC}$. 

		Suppose we are given the value of $\wt \bfh$ for a particular term in the linear combination of \eqref{psicircgeneral}. We can then proceed as follows. Starting from site $i=1$, we move along the chain and compute a running sum $\sum_{j=1}^i \t(\wt h_j)$. Each time this sum `overflows' in $C$ and generates a remainder of an amount $\rem(\wt \bfh)_i \in \mcg$, one acts on the wavefunction with the operator $S^R_{- \G\inv_\o(\rem(\wt \bfh)_i),i}$. Sweeping across the entire chain in this matter thus ensures that one is left with a state in which $\mch_{MNC}$ is acted on by $S^R_{\bar \bfh'},$ where $\sum_i\bar h_i'=e$. Adding up the values of $\wt h_i$ in a local way ensures that the $S^R_{- \G\inv_\o(\rem(\wt \bfh)_i),i}$ operators are inserted in such a way that the combined product of shift operators acting on $\k{\Psi_\o}$ always combine into charge-neutral blocks over distances of a correlation length $\xi$ (with $\xi$ set by the coefficients $C_{\bfg,\bfh}$ as usual). 

		The process just described is results in the following protocol to eliminate the mixed errors and perform phase recognition:	
		\begin{enumerate}
			\item Measure the flavor degrees of freedom in the computational basis of $\mch_{flav}$. Let the outcome of this measurement be $\wt \bfh$ (after measuring, the subsystem $\mch_{flav}$ can be discarded, as it will be decoupled from $\mch_{MNC}$ for the remainder of the protocol).
			\item Use an $O(L)$ amount of classical data processing to compute $\rem(\wt\bfh)$. 
			\item Apply $\mcu_{\rem}(\wt \bfh)$ to the measured state, where
			\be \mcu_{\rem}(\wt\bfh) \equiv \bot_{i=1}^L S^R_{-\G\inv_\o(\rem(\wt\bfh)_i),i}\ee
			is a product of single-site unitaries that converts the mixed errors into errors acting only on $\mch_{MNC}$. 
			\item Proceed by applying the error correction circuit $\bar\mcq_\o$ to $\mch_{MNC}$.
		\end{enumerate}
		
		While $\mcu_{\rem}(\wt\bfh)$ is simply a tensor product of single-qudit rotations, the data processing involved in computing $\rem(\wt\bfh)_i$ is non-local, since $\rem(\wt\bfh)_i$ depends on $\wt \bfh_j$ for all $j\leq i$. Also note that we do not need to post-select on the measurement outcomes $\wt \bfh$: any outcome is as good as any other for the purposes of performing phase recognition, as each set of terms with a fixed $\wt \bfh$ in the linear combination \eqref{psicircgeneral} contain all of the information about the universal aspects of the SPT phase. 
		
		That step 4 of the above protocol sufficies to perform phase recognition is established by the following proposition: 
		
		\begin{prop}
			For any $\k{\psi_\o} \in \spt_\o$, the state $\k{\psi_\o^{\wt \bfh}}$ obtained after measuring the flavor degrees of freedom to obtain the outcome $\wt\bfh$ and applying the corrective circuit $\mcu_{\rem}(\wt \bfh)$ can be written as 
			\be \label{finalpsi} \k{\psi_\o^{\wt \bfh}} = \sum_{\ub \bfg \in G^L} C_{\ub \bfg}(\wt \bfh) \mcr^\mcg_{\ub \bfg} \k{\Psi_\o},\ee
			where $\mcr_{\ub g}^\mcg$ is the action of the symmetry restricted to $\mch_{MNC}$: 
			\be \mcr^\mcg_{\ub g} = V_{\ub g}^* \tp \unit_{flav} \tp V_{\ub g}.\ee 
		\end{prop}

		\begin{proof} 

			By construction, we have 
			\be \label{newpsi} \k{\psi^{\wt \bfh}_\o} = \sum_{\bfg \in G^L, \bar\bfh\in \mcg^L} C_{\bfg,\bar\bfh}(\wt \bfh)\mcr_{\bfg} S^R_{\bar \bfh} \k{\Psi_\o^{\wt \bfh}},\ee 
			where the $ C_{\bfg,\bar\bfh}(\wt \bfh)$ are again some set of complex coefficients, and 
			\be \k{\Psi_\o^{\wt \bfh}} = \sum_{\bfg\in G^L} \Tr[A^{g_1}_{\wt h_1} \cdots A^{g_L}_{\wt h_L}] \ee 
			is the canonical reference state with the flavor degrees of freedom fixed to be $\wt \bfh$. Since the action of the corrective circuit $\mcu_{\rem}(\wt\bfh)$ yields a state in which the errors are symmetric when restricted to $\mch_{MNC}$, the coefficients $C_{\bfg,\bar\bfh}(\wt \bfh)$ appearing in \eqref{newpsi} are now constrained to satisfy $\sum_i \bar h_i = e$ (c.f. \eqref{hconds}). Thus when acting on $\k{\Psi^{\wt\bfh}_\o}$, $S^R_{\bar \bfh}$ can be replaced by a product of $\mcr^\mcg_{\bar \bfh}$ operators for the same reasons as discussed in the proof of theorem \ref{thm:simpleU}. When acting on $\k{\Psi^{\wt \bfh}_\o}$ the distinction between between $\mcr_{\bfg}$ and $\mcr^\mcg_{\ub \bfg}$ is not important, since the action of the two differs only by an unimportant overall phase. Thus there is no important dependence on $\wt\bfg$ in the sum, and so after changing notation for the complex coefficients appearing in the linear combination, we indeed arrive at \eqref{finalpsi}.

		\end{proof}
		
		That this allows us to perform phase recognition then follows from the fact that $\k{\psi_\o^{\wt \bfh}}$ has the same structure as the groud state wavefunction of a MNC SPT phase. This means that our previously developed RG circuit $\bar \mcq_\o$ can simply be applied directly to $\k{\psi_\o^{\wt \bfh}}$, with the remainder of the phase recognition procedure proceeding as in the MNC case.

		\bs 
		
	\end{widetext}

\end{document}